\documentclass[english,fleqn,a4paper,final]{article}

\newcommand{\fullonly}[1]{#1}
\newcommand{\fctonly}[1]{}
\newcommand{\fctsubmissiononly}[1]{}

\newcommand{\sisubmissiononly}[1]{#1}
\newcommand{\nosisubmission}[1]{}
\newcommand{\omittext}[1]{}

\fullonly{\usepackage{a4}}
\fullonly{\usepackage{amsmath}
\usepackage{amsthm}%
}
\usepackage{babel}
\usepackage{color}
\usepackage{float}
\usepackage{graphs}
\usepackage{hyphenat}
\usepackage{mathtools}
\usepackage{procalg}
\usepackage{url}

\newcommand{\N}{\ensuremath{\mathbb{N}}}
\renewcommand{\Act}[1][A]{\ensuremath{\mathcal{#1}}}
\newcommand{\ACP}{\ensuremath{\mathrm{ACP}}}
\newcommand{\ACPt}{\ensuremath{\ACP_{\silent}}}
\newcommand{\CCS}{\ensuremath{\mathrm{CCS}}}
\newcommand{\CSP}{\ensuremath{\mathrm{CSP}}}
\newcommand{\Chan}[1][C]{\ensuremath{\mathcal{#1}}}
\newcommand{\Data}[1][D]{\ensuremath{\mathcal{#1}}}
\newcommand{\Dataseqs}[1][D]{\ensuremath{\mathcal{#1}^{*}}}
\newcommand{\Datab}{\ensuremath{\Data_{\tpblank}}}
\newcommand{\Databseqs}{\ensuremath{(\Data_{\tpblank})^{*}}}
\newcommand{\IAct}[1][\Chan]{\ensuremath{\mathcal{I}_{#1}}}
\newcommand{\LTS}[1][T]{\ensuremath{\mathalpha{#1}}}
 \newcommand{\LTST}[1][]{\ensuremath{\LTS[T_{#1}]}}
\newcommand{\Name}[1][N]{\ensuremath{\mathcal{#1}}}
\newcommand{\PExp}[1][P]{\ensuremath{\mathalpha{\mathcal{#1}}}}
\newcommand{\RSpec}[1][E]{\ensuremath{\mathalpha{#1}}}
 
 \newcommand{\RSpecTM}{\ensuremath{\mathalpha{E_{\TM}}}}
 \newcommand{\RSpecTMinf}{\ensuremath{\mathalpha{E_{\TM}^{\infty}}}}
 \newcommand{\RSpecFC}{\ensuremath{\mathalpha{E_{\name[fc]}}}}
 \newcommand{\RSpecT}{\ensuremath{\mathalpha{E_{\name[T]}}}}
 \newcommand{\RSpecTinf}{\ensuremath{\mathalpha{E_{\name[T]}^{\infty}}}}
 \newcommand{\RSpecQ}{\ensuremath{\mathalpha{E_{\name[Q]}}}}
 \newcommand{\RSpecQinf}{\ensuremath{\mathalpha{E_{\name[Q]}^{\infty}}}}
\newcommand{\States}[1][S]{\ensuremath{\mathalpha{\mathcal{#1}}}}
 \newcommand{\StatesS}[1][]{\ensuremath{\States[S]_{#1}}}
\newcommand{\TCP}{\ensuremath{\mathrm{TCP}}}
\newcommand{\TCPt}{\ensuremath{\mathrm{\TCP_{\silent}}}}
\newcommand{\TM}[1][M]{\ensuremath{\mathalpha{\mathcal{#1}}}}
 \newcommand{\TMM}[1][]{\ensuremath{\mathalpha{\TM{_{#1}}}}}
 \newcommand{\URTM}[1][]{\ensuremath{\TM[U]{_{#1}}}}
 \newcommand{\UbRTM}{\ensuremath{\URTM[\bdegbound]}}
  \newcommand{\UbRTMStates}{\ensuremath{\StatesS[{\UbRTM}]}}
  \newcommand{\UbRTMtrel}{\ensuremath{\trel[{\UbRTM}]}}
  \newcommand{\UbRTMinit}{\ensuremath{\init[{\UbRTM}]}}
  \newcommand{\UbRTMfinal}{\ensuremath{\final[{\UbRTM}]}}
  
\newcommand{\rtmpar}[1][]{\ensuremath{\mathbin{\|_{#1}}}}
\newcommand{\act}[1][a]{\ensuremath{\mathalpha{#1}}}
 \newcommand{\acta}[1][]{\act[a_{#1}]}
 \newcommand{\actb}[1][]{\act[b_{#1}]}
 \newcommand{\actc}[1][]{\act[c_{#1}]}
\newcommand{\actt}[1][a]{\ensuremath{\mathalpha{#1}}}
  \newcommand{\actta}[1][]{\act[a_{#1}]}
  
\newcommand{\aep}[2]{\left[{#2}\right]_{#1}}
\newcommand{\chan}[1][c]{\ensuremath{\mathalpha{#1}}}
 
 \newcommand{\chani}{\ensuremath{\chan[i]}}
 \newcommand{\chano}{\ensuremath{\chan[o]}}
 \newcommand{\chanio}{\ensuremath{\{\chani, \chano\}}}
 \newcommand{\chanj}{\ensuremath{\chan[j]}}
 \newcommand{\chank}{\ensuremath{\chan[k]}}
 \newcommand{\chanp}{\ensuremath{\chan[p]}}
 \newcommand{\chanl}{\ensuremath{\chan[l]}}
 \newcommand{\chanr}{\ensuremath{\chan[r]}}
 \newcommand{\chanw}{\ensuremath{\chan[w]}}
 \newcommand{\chanm}{\ensuremath{\chan[m]}}
 \newcommand{\chanrwm}{\ensuremath{\{\chanr, \chanw, \chanm\}}}
 \newcommand{\chanrtm}{\ensuremath{\chan[u]}}

\renewcommand{\comm}[2]{{#1}\mathord{?\kern-0.4em!}{#2}}
\newcommand{\commm}[2]{{#1}\mathord{?\kern-0.3em!}{#2}}

\newcommand{\datum}[1][d]{\ensuremath{\mathalpha{#1}}}
  \newcommand{\datumd}[1][]{\ensuremath{\datum[d_{#1}]}}
  \newcommand{\datume}[1][]{\ensuremath{\datum[e_{#1}]}}
  \newcommand{\datumf}[1][]{\ensuremath{\datum[f_{#1}]}}
\newcommand{\defeq}{\ensuremath{\mathrel{\stackrel{\text{\tiny def}}{=}}}}
\newcommand{\emptystr}{\ensuremath{\mathalpha{\varepsilon}}}

\newcommand{\final}[1][]{\ensuremath{\mathalpha{\downarrow_{#1}}}}
\newcommand{\init}[1][]{\ensuremath{\mathalpha{\uparrow_{#1}}}}
\newcommand{\name}[1][N]{\ensuremath{\mathit{#1}}}

\newcommand{\pexp}[1][p]{\ensuremath{\mathalpha{#1}}}
  \newcommand{\pexpp}[1][]{\ensuremath{\pexp[p_{#1}]}}
  \newcommand{\pexpq}[1][]{\ensuremath{\pexp[q_{#1}]}}
\newcommand{\pref}[2]{\ensuremath{\mathalpha{#1}.{#2}}}
\newcommand{\qend}{\mathord{\bot}}
\newcommand{\qmark}{\mathord{\$}}
\newcommand{\recv}[2]{{#1}\mathord{?}{#2}}
\newcommand{\send}[2]{{#1}\mathord{!}{#2}}
\newcommand{\silent}{\ensuremath{\act[\tau]}}
\newcommand{\state}[1][s]{\ensuremath{\mathalpha{#1}}}
 \newcommand{\states}[1][]{\ensuremath{\state[s_{#1}]}}
 \newcommand{\statet}[1][]{\ensuremath{\state[t_{#1}]}}
 \newcommand{\stateu}[1][]{\ensuremath{\state[u_{#1}]}}
 \newcommand{\statev}[1][]{\ensuremath{\state[v_{#1}]}}
 \newcommand{\statew}[1][]{\ensuremath{\state[w_{#1}]}}
\renewcommand{\step}[2][]{\ensuremath{\mathbin{\arrow{#2}_{#1}}}}
  \newdimen\boxwdplusemdimen
  \def\arrow#1{{
     \boxwdplusemdimen=1em%
     \setbox0=\hbox{$\scriptstyle#1$}%
     \advance\boxwdplusemdimen by \wd0\relax%
     \ifdim\boxwdplusemdimen<16.11119pt%
       \boxwdplusemdimen=16.11119pt%
     \fi%
     \buildrel{#1}\over%
       {\setbox1=\hbox to \boxwdplusemdimen{\rightarrowfill}%
     \ht1=0.3em\relax\box1}%
   }}
\makeatletter
   \def\twoheadrightarrowfill{$\m@th\smash-\mkern-7mu%
     \cleaders\hbox{$\mkern-2mu\smash-\mkern-2mu$}\hfill
     \mkern-7mu\mathord\twoheadrightarrow$}
\makeatother
   \def\darrow#1{{
     \boxwdplusemdimen=1em%
     \setbox0=\hbox{$\scriptstyle#1$}%
     \advance\boxwdplusemdimen by \wd0\relax%
     \ifdim\boxwdplusemdimen<16.11119pt%
       \boxwdplusemdimen=16.11119pt%
     \fi%
     \buildrel{#1}\over%
       {\setbox1=\hbox to \boxwdplusemdimen{\twoheadrightarrowfill}%
     \ht1=0.3em\relax\box1}%
   }}
 \newcommand{\TMstep}[5][]{\step[#1]{#3{[#2/#4]}#5}}
\renewcommand{\steps}[2][]{\ensuremath{\mathbin{\arrow{#2}^{*}_{#1}}}}
\newcommand{\stepsplus}[2][]{\ensuremath{\mathbin{\arrow{#2}^{+}_{#1}}}}

\renewcommand{\term}[2][]{\ensuremath{{#2}\mathclose{\downarrow_{#1}}}}
\newcommand{\opt}[1]{\mbox{\tiny\rm(}#1\mbox{\tiny\rm)}} 
\newcommand{\optstep}[2][]{\step[#1]{\opt{#2}}}
\newcommand{\brelsym}[1][]{\ensuremath{\mathalpha{\mathcal{R}_{#1}}}}
\newcommand{\brel}[1][]{\ensuremath{\mathrel{\brelsym[#1]}}}
\newcommand{\todo}[1]{}
\newcommand{\comment}[1]{}
\newcommand{\qstn}[1]{}
\newcommand{\RTM}{RTM}
\newcommand{\tpunit}{\ensuremath{\mathalpha{1}}}
\newcommand{\tpblank}{\ensuremath{\mathalpha{\boxempty}}}
\newcommand{\tpsep}{\ensuremath{\#}}
 \newcommand{\tpsepc}{\ensuremath{\mathord{\mid}}}
 
\newcommand{\tpleft}{\ensuremath{\mathalpha{\llbracket}}}
\newcommand{\tpright}{\ensuremath{\mathalpha{\rrbracket}}}

\newcommand{\tphd}[1]{\check{#1}}
 \newcommand{\tphdL}[1]{{#1 \!}^{\scriptscriptstyle <}}
 \newcommand{\tphdR}[1]{\prescript{\scriptscriptstyle >}{}{\! #1}}
\newcommand{\tpmv}[1]{\ensuremath{\mathalpha{#1}}}
 \newcommand{\tpmvL}{\tpmv{L}}
 \newcommand{\tpmvR}{\tpmv{R}}
\newcommand{\tpstr}[1][\delta]{\ensuremath{\mathalpha{#1}}}
 \newcommand{\tpstrd}[1][]{\ensuremath{\tpstr[\delta_{#1}]}}
 \newcommand{\tpstrz}[1][]{\ensuremath{\tpstr[\zeta_{#1}]}}
 
 \newcommand{\tpstrdL}{\ensuremath{\tpstrd[\tpmvL]}}
 \newcommand{\tpstrdR}{\ensuremath{\tpstrd[\tpmvR]}}
 \newcommand{\tpstrzL}{\ensuremath{\tpstrz[\tpmvL]}}
 \newcommand{\tpstrzR}{\ensuremath{\tpstrz[\tpmvR]}}

\newcommand{\trel}[1][]{\ensuremath{\mathalpha{\rightarrow_{#1}}}}

\newcommand{\bbisimed}{\bisim_{\text{b}}^{\Delta}}

\newcommand{\lts}[2][]{\ensuremath{\mathcal{T}_{#1}(#2)}}

\newcommand{\codesym}{\ensuremath{\ulcorner{\_}\urcorner}}
\newcommand{\code}[1]{\ensuremath{\ulcorner{#1}\urcorner}}
\newcommand{\actcodesym}{\ensuremath{\codesym}}

\newcommand{\statecodesym}{\ensuremath{\codesym}}
\newcommand{\statecode}[1]{\ensuremath{\code{#1}}}

\newcommand{\finaltest}[1]{\ensuremath{\finaltestsym(#1)}}
\newcommand{\finaltestsym}{\ensuremath{\mathit{fin}}}
\newcommand{\outset}[1]{\ensuremath{\outsetsym(#1)}}
\newcommand{\outsetsym}{\ensuremath{\mathit{out}}}

\newcommand{\TMcode}[1]{\ensuremath{\code{#1}}}
\newcommand{\InF}{\ensuremath{\mathsf{Init}}}
  \newcommand{\InFStates}{\ensuremath{\StatesS[\InF]}}
  \newcommand{\InFinit}{\ensuremath{\init[\InF]}}
  \newcommand{\InFtrel}{\ensuremath{\trel[\InF]}}
  \newcommand{\InFstate}[1][]{\ensuremath{\mathit{in}_{#1}}}
\newcommand{\InUF}{\ensuremath{\mathsf{InitU}}}
  \newcommand{\InUFStates}{\ensuremath{\StatesS[\InUF]}}
  \newcommand{\InUFinit}{\ensuremath{\init[\InUF]}}
  \newcommand{\InUFtrel}{\ensuremath{\trel[\InUF]}}
  \newcommand{\InUFstate}[1][]{\ensuremath{\mathit{in}_{#1}}}
\newcommand{\StF}{\ensuremath{\mathsf{State}}}
  \newcommand{\StFStates}{\ensuremath{\StatesS[\StF]}}
  \newcommand{\StFinit}{\ensuremath{\init[\StF]}}
  \newcommand{\StFtrel}{\ensuremath{\trel[\StF]}}
  \newcommand{\StFstate}[1][]{\ensuremath{\mathit{st}_{#1}}}

\newcommand{\SpF}{\ensuremath{\mathsf{Step}}}
  \newcommand{\SpFStates}{\ensuremath{\StatesS[\SpF]}}
  \newcommand{\SpFinit}{\ensuremath{\init[\SpF]}}
  \newcommand{\SpFtrel}{\ensuremath{\trel[\SpF]}}
  \newcommand{\SpFstate}[1][]{\ensuremath{\mathit{sp}_{#1}}}

\newcommand{\BBCTSS}[1][]{\ensuremath{\mathsf{Sim}}}
  \newcommand{\BBCTSSStates}[1][]{\ensuremath{\StatesS[{\BBCTSS[#1]}]}}
  \newcommand{\BBCTSStrel}[1][]{\ensuremath{\trel[{\BBCTSS[#1]}]}}
  \newcommand{\BBCTSSinit}[1][]{\ensuremath{\init[{\BBCTSS[#1]}]}}
  \newcommand{\BBCTSSfinal}[1][]{\ensuremath{\final[{\BBCTSS[#1]}]}}
  
\newcommand{\INIT}{\ensuremath{\mathcal{I}}}
\newcommand{\REINIT}[1][]{\ensuremath{\mathcal{R}_{#1}}}

\newcommand{\STATE}{\ensuremath{\mathcal{U}}}
\newcommand{\STEP}{\ensuremath{\mathcal{D}}}

\newcommand{\TRANSF}{\ensuremath{\mathcal{F}}}
\makeatletter
  \def\nrightarrowfill{$\m@th\smash-\mkern-7mu%
    \cleaders\hbox{$\mkern-2mu\smash-\mkern-2mu$}\hfill
    \mkern-7mu\mathord\nrightarrow$}
\makeatother
  \def\narrow#1{{
     \boxwdplusemdimen=1em%
     \setbox0=\hbox{$\scriptstyle#1$}%
     \advance\boxwdplusemdimen by \wd0\relax%
     \ifdim\boxwdplusemdimen<16.11119pt%
       \boxwdplusemdimen=16.11119pt%
     \fi%
     \buildrel{#1}\over%
       {\setbox1=\hbox to \boxwdplusemdimen{\nrightarrowfill}%
     \ht1=0.3em\relax\box1}%
   }}
\newcommand{\nstep}[2][]{\ensuremath{\mathbin{\narrow{#2}_{#1}}\mbox{}}}
\newcommand{\bdegbound}{\ensuremath{B}}
\newcommand{\TMsend}[1]{\ensuremath{\overline{#1}}}
\newcommand{\fdcomp}[1][]{\ensuremath{\mathrel{\mapsto_{#1}}}}

\newcommand{\IS}[1]{\ensuremath{\mathalpha{\mathsf{IS}(#1)}}}
\newcommand{\initc}[1][]{\ensuremath{\mathit{ci}_{#1}}}
\newcommand{\finalc}[1][]{\ensuremath{\mathit{cf}_{#1}}}
\newcommand{\TMcomp}[1][]{\ensuremath{\mathrel{\mapsto_{#1}}}}

\floatstyle{boxed}
\newfloat{framedtable}{htbp}{lot}
\floatname{framedtable}{Table}

\fullonly{\theoremstyle{definition}
 \newtheorem{theorem}{Theorem}[section]
 \newtheorem{corollary}[theorem]{Corollary}
 \newtheorem{lemma}[theorem]{Lemma}
 \newtheorem{fact}[theorem]{Fact}
 \newtheorem{proposition}[theorem]{Proposition}
 \newtheorem{definition}[theorem]{Definition}
 \newtheorem{example}[theorem]{Example}
 \newtheorem{remark}[theorem]{Remark}
}

\makeatletter
\newenvironment{spec}%
 {\everymath{\displaystyle}\start@align\@ne\st@rredtrue\m@ne}
 {\endalign}
\makeatother

\title{Reactive Turing Machines}
  \author{Jos C.\ M.\ Baeten \and Bas Luttik \and Paul van Tilburg}
  \date{}%

\nosisubmission{\fctonly{%
  \institute{Eindhoven University of Technology\\
              \email{\{j.c.m.baeten,s.p.luttik,p.j.a.v.tilburg\}@tue.nl}}}
}

\begin{document}

\maketitle

\begin{abstract}
  We propose reactive Turing machines (\RTM{}s), extending classical
  Turing machines with a process-theoretical notion of interaction,
  and use it to define a notion of executable transition system.  We
  show that every computable transition system with a bounded
  branching degree is simulated modulo divergence-preserving branching
  bisimilarity by an \RTM{}, and that every effective transition
  system is simulated modulo the variant of branching bisimilarity
  that does not require divergence preservation. We conclude from
  these results that the   parallel composition of (communicating)
  \RTM{}s can be simulated by a single \RTM{}. We prove that there
  exist universal \RTM{}s modulo branching bisimilarity, but these
  essentially employ divergence to be able to simulate an \RTM{} of
  arbitrary branching degree. We also prove that modulo
  divergence-preserving branching bisimilarity there are \RTM{}s that
  are universal up to their own branching degree. We
  establish a correspondence between executability and finite
  definability in a simple process calculus. Finally, we establish that
  \RTM{}s are at least as expressive as \emph{persistent Turing machines}.
\end{abstract}

\sisubmissiononly{\maketitle}

\section{Introduction}\label{sec:intro}

The Turing machine~\cite{Tur36} is widely accepted as a computational
model suitable for exploring the theoretical boundaries of
computing. Motivated by the existence of universal Turing machines,
many textbooks on the theory of computation 
present the Turing machine not just as a theoretical model to explain
which functions are computable, but as an accurate conceptual model of
the computer.
There is, however, a well-known limitation to this view
\cite{Weg97,GW08}. A Turing machine operates from the assumptions that: (1)
all input it needs for the computation is available on the tape
from the very beginning; (2) it performs a terminating computation;
and (3) it leaves the output on the tape at the very end.
\fullonly{That is, a Turing machine computes a function, and thus
  it}\fctonly{Thus, the notion of Turing machine} abstracts from two
key ingredients of computing: \emph{interaction} and
\emph{non-termination}. Nowadays, most computing systems are so-called
\emph{reactive systems}~\cite{HP85}, systems that are generally not
meant to terminate and consist of computing devices that interact with
each other and with their environment.



Concurrency theory emerged from the early work of Petri \cite{Pet62}
and has now developed into a mature theory of reactive
systems.  We mention three of its contributions particularly relevant
for our work.  Firstly, it installed the notion of transition system%
\fullonly{ ---a generalisation of the notion of finite-state automaton
  from classical automata theory---} as the prime mathematical model
to represent discrete behaviour.  Secondly, it offered the insight
that language equivalence%
\fullonly{ ---the underlying equivalence in classical automata
  theory---} is too coarse in a setting with interacting automata;
instead one should consider automata up to some form of
bisimilarity. Thirdly, it yielded many algebraic process calculi
facilitating the formal specification and verification of reactive
systems.

Several proposals have been made in the literature extending Turing
machines with a notion of interaction (e.g., the \emph{persistent}
Turing machines of \cite{GSAS04} and the \emph{interactive} Turing
machines of \cite{LW01}). Since the purpose of these works is studying
the effect of interaction on the expressiveness of sequential
computation, interaction is added through an ad hoc input-output
facility of the Turing machine. In this paper, we propose to add
interaction as an orthogonal facility, in line with the way
interaction is studied in concurrency theory. The result will be a
semantically refined model of interactive behaviour, integrating the
well-established classical theory of automata with the theory of
concurrency. The advantage of the integration is that interactive
behaviour can be studied both from a concurrency-theoretic perspective
(e.g., up to any of the many behavioural equivalences known from the
concurrency-theoretic literature, see \cite{Gla93}), while at the same
time there is an automata-based notion of executability associated
with it.


In Sect.~\ref{sec:rtm} we propose a notion of \emph{reactive} Turing
machine (\RTM{}), extending the classical notion of Turing machine
with interaction in the style of concurrency theory.  The extension
consists of a facility to declare every transition to be either
\emph{observable}, by labelling it with an action symbol, or
unobservable, by labelling it with~$\silent$. Typically, a transition
labelled with an action symbol models an interaction of the \RTM{}
with its environment (or some other \RTM{}), while a transition
labelled with~$\silent$ refers to an internal computation step. Thus,
a conventional Turing machine can be regarded as a special kind of
\RTM{} in which all transitions are declared unobservable by labelling
them with~$\silent$.

The semantic object associated with a conventional Turing machine is
either the function that it computes, or the formal language that it
accepts. The semantic object associated with an \RTM{} is a behaviour,
formally represented by a transition system. A function is said to be
effectively computable if it can be computed by a Turing machine. By
analogy, we say that a behaviour is effectively executable if it can
be exhibited by an \RTM{}.  In concurrency theory,
behaviours are usually considered modulo a suitable behavioural
equivalence. In this paper we use
\emph{(divergence-preserving) branching bisimilarity}~\cite{GW96},
which is the finest behavioural equivalence in Van Glabbeek's spectrum
(see~\cite{Gla93}).

\todo{BL: Establishing divergence-preservation means that most of our
  results do not depend on fairness assumptions. This should be
  explained either here or in the conclusions.}

\todo{BL: Mention here, or in the conclusion section, that
  bisimilarity gives a much finer perspective on the behaviour of
  Turing machines.}

In Sect.~\ref{sec:expr} we set out to investigate the expressiveness
of \RTM{}s up to divergence-preserving branching bisimilarity.
\fullonly{%
We present an example of a behaviour that is not executable up
to branching bisimilarity. Then, we}%
\fctonly{We}
establish that every computable
transition system with a bounded branching degree can be simulated, up
to divergence-preserving branching bisimilarity, by an \RTM{}. If the
divergence-preservation requirement is dropped, then every effective
transition system can be simulated. These results
then allow us to conclude that the behaviour of a parallel composition
of \RTM{}s can be simulated on a single \RTM{}.

In Sect.~\ref{sec:universal} we define a suitable notion of
universality for \RTM{}s and investigate the existence of universal
\RTM{}s. \todo{Reformulate: }We find that%
\fullonly{%
  , since bisimilarity is sensitive to branching,}
there are some subtleties pertaining to the branching
degree bound associated with each \RTM{}. Up to divergence-preserving
branching bisimilarity, an \RTM{} can at best simulate other \RTM{}s
with the same or a lower bound on their branching degree. If
divergence-preservation is not required, however, then
universal \RTM{}s do exist.

In Sect.~\ref{sec:expl-int}, we consider the correspondence between
\RTM{}s and a process calculus consisting of a few standard
process-theoretic constructions.  On the one hand, the process
calculus provides a convenient way to specify executable behaviour;
indeed, every guarded recursive specification gives rise to a
computable transition system \cite{Vaa92}, which can be simulated up
to branching bisimilarity with an \RTM{}. On the other hand, we
establish that every executable behaviour is, again up to
divergence-preserving branching bisimilarity, finitely definable in
our calculus. Recursive specifications are the concurrency-theoretic
counterparts of grammars in the theory of formal languages.  Thus, the
result in Sect.~\ref{sec:expl-int} may be considered as the
process-theoretic version of the correspondence between Turing
machines and unrestricted grammars.

In Sect.~\ref{sec:PTM}, we argue that reactive Turing machines are at
least as expressive as the persistent Turing machines of
\cite{GSAS04}, and in Sect.~\ref{sec:conclusion} we conclude the paper
with a discussion of related work and some ideas for future work.

\todo{BL: Mention our course on automata and process theory, and Jos'
  lecture notes.}
\todo{BL: Add remark about $\lambda$-calculus and $\pi$-calculus, and how
  our RTMs can be used to determine their expressiveness.}
\todo{BL: Add example about difference between language equivalence and
  bisimilarity.}
\todo{BL: Discuss Turing's own oracle machines.}
\todo{BL: Add result that there exist non-executable processes with a
 decidable language, and refer to it in the introduction.}
\todo{BL: Add reference to the book ``Interactive Computation''.}

\section{Reactive Turing Machines}\label{sec:rtm}

\todo{BL: Implement FoSSaCS reviewer's suggestion to use alternative
  notation for \RTM{}s. Perhaps we should go further and also use
  alternative terminology. For instance, denote its collection of
  states by $Q=\{q_0,q_1,\dots\}$. We should then probably also denote
  its transition relation by $\delta$.}

We presuppose a finite set $\Act{}$ of \emph{action symbols} that we
use to denote the observable events of a system. An unobservable event
is denoted with $\silent$, assuming that $\silent\not\in\Act{}$; we
henceforth denote the set $\Act{}\cup\{\silent\}$ by $\Actt$. We also
presuppose a finite set $\Data$ of \emph{data symbols}. We add to
$\Data$ a special symbol $\tpblank$ to denote a blank tape cell,
assuming that $\tpblank\not\in\Data$; we denote the set
$\Data\cup\{\tpblank\}$ of \emph{tape symbols} by $\Datab$. Mostly,
the precise contents of the sets $\Act{}$ and $\Data{}$ are
unimportant for the developments in this paper, but it will
occasionally be convenient to assume explicitly that certain special
symbols are included in them.

\begin{definition}
\label{def:RTM}
  A \emph{reactive Turing machine} (\RTM{}) $\TM$ is a quadruple
  $(\States, \trel, \init, \final)$ consisting of
    a finite set of \emph{states} $\States$,
    a distinguished \emph{initial state} $\init\in\States$,
    a subset of \emph{final states}
      $\final\subseteq\States$, and
    a $(\Datab\times\Actt\times\Datab\times\{\tpmvL,\tpmvR\})$-labelled
    transition relation
    \begin{equation*}
      \trel \subseteq
        \States
          \times
        \Datab\times\Actt\times\Datab\times\{\tpmv{L},\tpmv{R}\}
          \times
        \States
    \enskip.
    \end{equation*}
  An \RTM{} is \emph{deterministic} if
    $(\states,\datumd,\actt,\datume[1],\tpmv{M_1},\statet[1])\in\trel$
  and
    $(\states,\datumd,\actt,\datume[2],\tpmv{M_2},\statet[2])\in\trel$
  implies that
    $\datume[1]=\datume[2]$,
    $\statet[1]=\statet[2]$ and
    $\tpmv{M_1}=\tpmv{M_2}$ 
  for all
    $\states,\statet[1],\statet[2]\in\States$,
    $\datumd,\datume[1],\datume[2]\in\Datab$,
    $\actt\in\Actt$, and
    $\tpmv{M_1},\tpmv{M_2}\in\{\tpmv{L},\tpmv{R}\}$,
  and, moreover,
   $(\states,\datumd,\silent,\datume[1],\tpmv{M_1},\statet[1])\in\trel$
  implies that there do not exist $\actt\neq\silent$,
  $\datume[2],\tpmv{M_2},\statet[2]$ such that
     $(\states,\datumd,\actt,\datume[2],\tpmv{M_2},\statet[2])\in\trel$
\end{definition}

If $(\states,\datumd,\actta,\datume,M,\statet)\in\trel$, we write
  $\states \TMstep{\datumd}{\actta}{\datume}{\tpmv{M}} \statet$.
The intuitive meaning of such a transition is that whenever $\TM$ is
in state $\states$ and $\datumd$ is the symbol currently read by the
tape head, then it may
  execute the action $\actta$,
  write symbol $\datume$ on the tape (replacing $\datumd$),
  move the read/write head one position to the left or one
  position to the right on the tape (depending on whether
 $\tpmv{M} = \tpmv{L}$ or $\tpmv{M} = \tpmv{R}$),
and then end up in state $\statet$.
\RTM{}s extend conventional Turing machines by associating with every
transition an element $\acta\in\Actt$. The symbols in $\Act$ are thought
of as denoting observable activities; a transition labelled with an
action symbol in $\Act$ is semantically be treated as observable.
Observable transitions are used to model interactions of an \RTM{} with
its environment or some other \RTM{}, as will be explained more in detail
below when we introduce a notion of parallel composition for \RTM{}s
(see Definition~\ref{def:RTMparc} and Example~\ref{def:RTMparc} below).
The symbol $\silent$ is used to declare that a transition is
unobservable. A classical Turing machine is an \RTM{} in which all
transitions are declared unobservable.

\begin{example}\label{ex:rtm}
  Assume that
    $\Act{}=\{\send{\chan{}}{\datumd},
            \recv{\chan{}}{\datumd}
              \mid \chan{}\in\{\chani,\chano\}\ \&\ \datumd\in\Datab\}$.
  Intuitively, $\chani$ and $\chano$ are the input/output communication
  channels by which the \RTM{} can interact with its environment.
  The action symbol
    $\send{\chan{}}{\datumd}$ ($\chan{}\in\{\chani,\chano\}$)
  denotes the event that a datum $\datumd$ is sent by the \RTM{} along
  channel $\chan{}$, and the action symbol $\recv{\chan{}}{\datumd}$
  ($\chan{}\in\{\chani,\chano\}$) denotes the event that a datum
  $\datumd$ is received by the \RTM{} along channel $\chan{}$.

  The left state-transition diagram in Fig.~\ref{fig:ex-rtm}
  specifies an \RTM{} that first inputs a string, consisting of
  an arbitrary number of $\tpunit$s followed by the symbol $\tpsep{}$,
  stores the string on the tape, and returns to the beginning of the
  string. Then, it performs a computation to determine if the number
  of $\tpunit$s is odd or even. In the first case, it simply removes the
  string from the tape and returns to the initial state. In the second
  case, it outputs the entire string, removes it from the tape, and
  returns to the initial state.
  Note that, ignoring occurrences of the action symbol $\silent$, the
  right state-transition diagram in Fig.~\ref{fig:ex-rtm} generates the
  sequence of actions
  \begin{equation*}
    \send{\chani}{\tpunit},
     \send{\chani}{\tpsep},
     \send{\chani}{\tpunit},\send{\chani}{\tpunit},
     \send{\chani}{\tpsep},
     \send{\chani}{\tpunit},\send{\chani}{\tpunit},\send{\chani}{\tpunit},
     \dots
  \end{equation*}
  which, according to concurrency-theoretic usage, should be thought
  of as modelling outputting on channel $\chani{}$ the infinite sequence
    $\tpunit\tpsep\tpunit\tpunit\tpsep\tpunit\tpunit\tpunit\tpsep\cdots\tpsep \tpunit^n \# \cdots$
  ($n\geq 1$). (How this particular concurrency-theoretic
  interpretation of action symbols leads to a formalisation of
  interaction will be clarified when we define the parallel
  composition of \RTM{}s.)
\end{example}

 \begin{figure}[t]
  \centering\small
  \begin{transsys}(109,42)(0,-3)
    \node[Nmarks=i](i1)(0,30){%
    }
    \node(i2)(22,30){%
    }
    \node(o)(55,30){%
    }
    \node(b)(55,15){%
    }
    \node(e)(40,30){%
    }
    \node(f)(40,15){%
    }
   {\scriptsize
    \edge(o,b){$\silent[\tpsep/\tpblank]\tpmvL$}
    \edge*[loopangle=0,ELpos=80](b){$\silent[\tpunit/\tpblank]\tpmvL$}
    \edge[ELside=r](e,f){$\silent[\tpsep/\tpblank]\tpmvL$}
    \edge*[loopangle=-90,ELpos=77](f){$\send{\chano}{\tpunit}[\tpunit/\tpblank]\tpmvL$}
    \drawbpedge[ELpos=40](f,180,32,i1,0,15){$\send{\chano}{\tpsep}[\tpblank/\tpblank]\tpmvR$}
    \drawbpedge[ELpos=20](b,-135,20,i1,-90,35){$\silent[\tpblank/\tpblank]\tpmvR$}
    \edge*[loopangle=90](i1){$\recv{\chani}{\tpunit}[\tpblank/\tpunit]\tpmvR$}
    \edge(i1,i2){$\recv{\chani}{\tpsep}[\tpblank/\tpsep]\tpmvL$}
    \edge*[loopangle=90](i2){$\silent[\tpunit/\tpunit]\tpmvL$}
    \edge(i2,e){$\silent[\tpblank/\tpblank]\tpmvR$}
    \graphset{curvedepth=-2,ELside=r}
    \edge(e,o){$\silent[\tpunit/\tpunit]\tpmvR$}
    \edge(o,e){$\silent[\tpunit/\tpunit]\tpmvR$}
   }
   \node[Nmarks=i](ei)(100,30){}
   \node(e1)(100,15){}
   \node(e2)(100,0){}
   \node(e3)(82,0){}
   {\scriptsize  
    \graphset{loopangle=0}
    \edge(ei,e1){$\silent[\tpblank/\tpunit]\tpmvR$}
    \edge(e1,e2){$\silent[\tpblank/\tpblank]\tpmvL$}
    \edge*(e2){$\silent[\tpunit/\tpunit]\tpmvL$}
    \edge(e2,e3){$\silent[\tpblank/\tpblank]\tpmvR$}
    \edge*[loopangle=180,ELpos=20](e3){$\send{\chani}{\tpunit}[\tpunit/\tpunit]\tpmvR$}
    \edge(e3,e1){$\send{\chani}{\tpsep}[\tpblank/\tpunit]\tpmvR$}
   }
  \end{transsys}
  \caption{Examples of reactive Turing machines.}\label{fig:ex-rtm}
 \end{figure}


To formalise our intuitive understanding of the operational behaviour
of \RTM{}s we shall below associate with every \RTM{} a transition
system.
\fullonly{\begin{definition}}%
An \emph{$\Actt$-labelled transition system} $\LTS$ is a
quadruple $(\States,\trel,\init,\final)$ consisting of a set of
\emph{states} $\States$, an \emph{initial} state $\init\in\States$, a
subset $\final\subseteq\States$ of \emph{final} states, and an
$\Actt$-labelled \emph{transition relation}
$\trel\subseteq\States\times\Actt\times\States$.  If $(\states, \act,
\statet) \in \trel$, we write $\states \step{\act} \statet$.  If
$\states$ is a final state, i.e., $\states \in \final$, we write
$\term{\states}$.
The transition system $\LTST$ is \emph{deterministic} if, for every
state $\states\in\States$ and for every $\acta\in\Actt$,
    $\states\step{\acta}\states[1]$ and $\states\step{\acta}\states[2]$
  implies
    $\states[1]=\states[2]$,
and, moreover, $\states\step{\silent}\states[1]$ implies that there do
not exist an action $\act\neq\silent$ and a state $\states[2]$ such that
$\states\step{\acta}\states[2]$.
%
\fullonly{\end{definition}}

With every \RTM{} $\TM{}$ we are going to associate a transition
system $\lts{\TM{}}$.  The states of $\lts{\TM{}}$ are the
configurations of the \RTM{}, consisting of a state of the \RTM{}, its
tape contents, and the position of the read/write head on the tape.
We represent the tape contents by an element of $\Databseqs$,
replacing precisely one occurrence of a tape symbol $\datumd$ by a
\emph{marked} symbol $\tphd{\datumd}$, indicating that the read/write
head is on this symbol.  We denote by
$\tphd{\Datab}=\{\tphd{\datumd}\mid\datumd\in\Datab\}$ the set of
\emph{marked} tape symbols; a \emph{tape instance}\label{tapeinstance} is a sequence
$\tpstrd\in{(\Datab\cup\tphd{\Datab})}^{*}$ such that $\tpstrd$
contains exactly one element of $\tphd{\Datab}$.  \fullonly{%
  Note that we do not use $\tpstrd$ exclusively for tape instances; we
  also use $\tpstrd$ for sequences over $\Data$.  A tape instance thus
  is a finite sequence of symbols that represents the contents of a
  two-way infinite tape. Henceforth, we do not distinguish between
  tape instances that are equal modulo the addition or removal of
  extra occurrences of the symbol $\tpblank$ at the left or right
  extremes of the sequence. That is, we do not distinguish tape
  instances $\tpstrd[1]$ and $\tpstrd[2]$ if
  $\tpblank^\omega\tpstrd[1]\tpblank^\omega
  =\tpblank^\omega\tpstrd[2]\tpblank^\omega$.  } \fullonly{%
\begin{definition}\label{def:RTMconf}
 A \emph{configuration} of an \RTM{}
    $\TM=(\States, \trel, \init, \final)$
  is a pair $(\states,\tpstrd)$ consisting of a state
    $\states\in\States$, and
  a tape instance $\tpstrd$.
\end{definition}%
}%
\fctonly{%
  Formally, a \emph{configuration} is now a pair $(\states,\tpstrd)$
  consisting of a state $\states\in\States$, anda tape instance $\tpstrd$.
}

Our transition system semantics defines an $\Actt$-labelled
transition relation on configurations such that an \RTM{}-transition
$\states\TMstep{\datumd}{\actta}{\datume}{M} \statet$ corresponds
with $\actta$-labelled transitions from configurations consisting of
the \RTM{}-state $\states$ and a tape instance in which some
occurrence of $\datumd$ is marked. The transitions lead to
configurations consisting of $\statet$ and a tape instance in which
the marked symbol $\datumd$ is replaced by $\datume$, and either the
symbol to the left or to right of this occurrence of $\datume$ is
replaced by its marked version, according to whether $M=L$ or
$M=R$. If $\datume$ happens to be the first symbol and $M=L$, or the
last symbol and $M=R$, then an additional blank symbol is appended at
the left or right end of the tape instance, respectively, to model the
movement of the head.

We introduce some notation to concisely denote the new placement of
the tape head marker.  Let $\tpstrd$ be an element of $\Datab^{*}$.
Then by $\tphdL{\tpstrd}$ we denote the element of
${(\Datab\cup\tphd{\Datab})}^{*}$ obtained by placing the tape head
marker on the right-most symbol of $\tpstrd$ if it exists, and
$\tphd{\tpblank}$ otherwise\fctonly{.}\fullonly{, i.e.,
\begin{equation*}
  \tphdL{\tpstrd} = 
    \begin{cases}
      \tpstrz\tphd{\datumd} & \text{if}\ \tpstrd = \tpstrz\datumd\quad
        (\datumd \in \Datab, \tpstrz\in\Datab^{*})\;,\ \text{and} \\
      \tphd{\tpblank}        & \text{if}\ \tpstrd = \emptystr\;.
    \end{cases}
\end{equation*}
(We use $\emptystr$ to denote the empty sequence.)}
Similarly, by $\tphdR{\tpstrd}$ we denote the element of
  ${(\Datab\cup\tphd{\Datab})}^{*}$
obtained by placing the tape head marker on the left-most symbol of
$\tpstrd$ if it exists, and $\tphd{\tpblank}$ otherwise\fctonly{.}\fullonly{, i.e.,
\begin{equation*}
  \tphdR{\tpstrd} = 
    \begin{cases}
      \tphd{\datumd}\tpstrz & \text{if}\ \tpstrd = \datumd\tpstrz\quad
        (\datumd \in \Datab, \tpstrz\in\Datab^{*})\;,\ \text{and} \\
      \tphd{\tpblank}        & \text{if}\ \tpstrd = \emptystr\;.
    \end{cases}
\end{equation*}}

\begin{definition}\label{def:RTMopsem}
  Let $\TM{}=(\States,\trel,\init,\final)$ be an \RTM{}.
  The \emph{transition system $\lts{\TM{}}$ associated with $\TM{}$} is defined
  as follows:
  \begin{enumerate}
  \item its set of states is the set of all configurations of $\TM{}$;
  \item its transition relation $\trel$ is the least relation
    satisfying, for all $\act\in\Actt$, $\datumd,\datume \in \Datab$ and
    $\tpstrdL,\tpstrdR\in\Datab^{*}$:
    \begin{gather*}
      (\states, \tpstrdL \tphd{\datumd} \tpstrdR)
        \step{\act}
      (\statet, \tphdL{\tpstrdL} \datume \tpstrdR)
         \text{ iff }
           \states \TMstep{\datumd}{\act}{\datume}{\tpmv{L}} \statet
    \enskip, \text{and}\\
      (\states, \tpstrdL \tphd{\datumd} \tpstrdR)
         \step{\act} 
      (\statet, \tpstrdL \datume \tphdR{\tpstrdR})
        \text{ iff }
          \states \TMstep{\datumd}{\act}{\datume}{\tpmv{R}} \statet
    \enskip;
    \end{gather*}
  \item its initial state is the configuration $(\init,\tphd{\tpblank})$;
    and
  \item its set of final states is the set of terminating configurations
      $\{(\states, \tpstrd) \mid \term{\states} \}$.
  \end{enumerate}
\end{definition}

Turing introduced his machines to define the notion of
\emph{effectively computable function}. By analogy, our notion of
\RTM{} can be used to define a notion of \emph{effectively executable
  behaviour}.
\begin{definition}
  A transition system is \emph{executable} if it is associated with
  an \RTM{}.
\end{definition}

\paragraph{Parallel composition}

To illustrate how \RTM{}s are suitable to model a form of interaction, we
proceed to define on \RTM{}s a notion of parallel composition, equipped
with a simple form of communication. 
\fullonly{(We are not trying to define the
most general or most suitable notion of parallel composition for
\RTM{}s here; the purpose of the notion of parallel composition
defined here is just to illustrate how \RTM{}s may run in parallel and
interact.)}
Let \Chan{} be a finite set of \emph{channels} for the
communication of data symbols between one \RTM{} and another.
Intuitively, $\send{\chan}{\datum}$
stands for the action of sending datum $\datum$ along channel $\chan$,
while $\recv{\chan}{\datum}$ stands for the action of receiving datum
$\datum$ along channel $\chan$.

First, we define a notion of parallel composition on transition systems.
Let
  $\LTST[1]=(\StatesS[1],\trel[1],\init[1],\final[1])$
and
  $\LTST[2]=(\StatesS[2],\trel[2],\init[2],\final[2])$
be transition systems, and let $\Chan{}'\subseteq\Chan{}$.
The \emph{parallel composition}
of $\LTST[1]$ and $\LTST[2]$ is the transition system
  $\aep{\Chan'}{\LTST[1]\parc \LTST[2]}=(\States,\trel,\init,\final)$,
with $\States$, $\trel$, $\init$ and $\final$ defined by
\begin{enumerate}
  \item $\States=\StatesS[1]\times\StatesS[2]$;
  \item $(\states[1],\states[2])\step{\acta}(\state[s_1'],\state[s_2'])$ iff
$\acta\in\Actt-\{\send{\chan}{\datum},\recv{\chan}{\datum}\mid \chan\in\Chan',\datum\in\Datab\}$ and either
    \begin{enumerate}
    \item   $\states[1]\step{\acta}\state[s_1']$ and $\states[2]=\state[s_2']$,
          or
             $\states[2]\step{\acta}\state[s_2']$ and $\states[1]=\state[s_1']$,
  or
    \item $\acta=\silent$ and
      either $\states[1]\step{\send{\chan}{\datum}}\state[s_1']$
               and
             $\states[2]\step{\recv{\chan}{\datum}}\state[s_2']$,
          or
             $\states[1]\step{\recv{\chan}{\datum}}\state[s_1']$
               and
             $\states[2]\step{\send{\chan}{\datum}}\state[s_2']$
      for some $\chan\in\Chan'$ and $\datum\in\Datab$;
    \end{enumerate}
  \item $\init=(\init[1],\init[2])$; and
  \item $\final=\{(\states[1],\states[2])\mid\states[1]\in\final[1]\ \&\ \states[2]\in\final[2]\}$.
\end{enumerate}

\begin{definition} \label{def:RTMparc}
  Let
    $\TM{_1}=(\StatesS[1],\trel[1],\init[1],\final[1])$
  and
    $\TM{_2}=(\StatesS[2],\trel[2],\init[2],\final[2])$
  be RTMs, and let $\Chan'\subseteq\Chan$;
  by $\aep{\Chan'}{\TM{_1}\parc\TM{_2}}$ we denote the
  \emph{parallel composition} of $\TM{_1}$ and $\TM{_2}$.
  The transition system $\lts{\aep{\Chan'}{\TM{_1}\parc\TM{_2}}}$ associated
  with the parallel composition $\aep{\Chan'}{\TM{_1}\parc_{\Chan}\TM{_2}}$ of
  $\TM{_1}$ and $\TM{_2}$ is the parallel composition of the transition
  systems associated with $\TM{_1}$ and $\TM{_2}$, i.e.,
    $\lts{\aep{\Chan'}{\TM{_1}\parc\TM{_2}}}=
       \aep{\Chan'}{\lts{\TM{_1}}\parc\lts{\TM{_2}}}$.
\end{definition}

\begin{example} \label{exa:RTMparc}
  Let $\Act$ be as in Example~\ref{ex:rtm}, let $\TM{}$ denote the
  left-hand side \RTM{} in Fig.~\ref{fig:ex-rtm}, and let $\TM[E]$
  denote the right-hand side \RTM{} in Fig.~\ref{fig:ex-rtm}.
  Then the parallel composition
    $\aep{\chani}{\TM\parc \TM[E]}$
  exhibits the behaviour of outputting, along channel $\chano$, the string
    $\tpunit\tpunit\tpsep
      \tpunit\tpunit\tpunit\tpunit
        \tpsep\cdots\tpsep
      \tpunit^n
        \tpsep$ ($n\geq 2$, $n$ even).
\end{example}

\paragraph{Behavioural equivalence}

In automata theory, Turing machines that compute the same function or
accept the same language are generally considered equivalent. In fact,
functional or language equivalence is underlying many of the standard
notions and results in automata theory. Perhaps most notably, a
\emph{universal} Turing machine is a Turing machine that, when started
with the code of some Turing machine on its tape, simulates this
machine up to functional or language equivalence. A result from
concurrency theory is that functional and language equivalence are
arguably too coarse for reactive systems, because they abstract from
all moments of choice (see, e.g.,~\cite{BBR09}). In concurrency theory
many alternative behavioural equivalences have been proposed; we refer
to~\cite{Gla93} for a classification.

The results about \RTM{}s that are obtained in the remainder of
this paper are modulo \emph{branching bisimilarity}~\cite{GW96}%
\fullonly{, which is the finest behavioural equivalence in Van Glabbeek's linear time --
branching time spectrum~\cite{Gla93}}. We consider both the
divergence-insensitive and the divergence-preserving variant.
\fullonly{(The divergence-preserving variant is called \emph{branching
  bisimilarity with explicit divergence} in~\cite{GW96,Gla93}, but in
this paper we prefer the term \emph{divergence-preserving} branching
bisimilarity.)

}%
\fctonly{%
Let $\LTST[1]$ and $\LTST[2]$ be transition systems.
If $\LTST[1]$ and $\LTST[2]$ are \emph{branching bisimilar}, then we
write $\LTST[1]\bbisim\LTST[2]$.
If $\LTST[1]$ and $\LTST[2]$ are \emph{divergence-preserving branching
bisimilar}, then we write $\LTST[1]\bbisimed\LTST[2]$. 
(Due to space limitations, the formal definitions had to be omitted;
the reader is referred to the full version~\cite{BLT11full}, or to \cite{GLT09}, where the
divergence-preserving variant is called \emph{branching
  bisimilarity with explicit divergence}.%
\fctsubmissiononly{
  Furthermore, for the convenience of the reviewer we have also
  included the definition in Appendix~\ref{app:bbed}.}%
)}

\fullonly{%
We proceed to define the behavioural equivalences that we
employ in this paper to compare transition systems. Let $\trel$ be
an $\Actt$-labelled transition relation on a set $\States$, and let
$\actt\in\Actt$; we write $\states\optstep{\actt}\statet$ if
$\states\step{\actt}\statet$ or $\actt=\silent$ and
$\states=\statet$.  Furthermore, we denote the transitive closure of
$\step{\silent}$ by $\stepsplus{}$, and we denote the
reflexive-transitive closure of $\step{\silent}$ by $\steps{}$.

\begin{definition}\label{def:bbisim}
  Let
    $\LTST[1]=(\StatesS[1],\trel[1],\init[1],\final[1])$
  and
    $\LTST[2]=(\StatesS[2],\trel[2],\init[2],\final[2])$
  be transition systems.
  A \emph{branching bisimulation} from $\LTST[1]$ to $\LTST[2]$ is a
  binary relation $\brelsym\subseteq\StatesS[1]\times\StatesS[2]$ and, for all states $\states[1]$ and
  $\states[2]$, $\states[1]\brel\states[2]$ implies
  \begin{enumerate}
  \item if $\states[1]\step[1]{\actt}\state[s_1']$, then there exist
      $\state[s_2'],\state[s_2'']\in\StatesS[2]$ such that
      $\states[2]\steps[2]{}\state[s_2'']\optstep[2]{\actt}\state[s_2']$, 
      $\states[1]\brel\state[s_2'']$
    and 
      $\state[s_1']\brel\state[s_2']$;
  \item if $\states[2]\step[2]{\actt}\state[s_2']$, then there exist
    $\state[s_1'],\state[s_1'']\in\StatesS[1]$ such that
      $\states[1]\steps[1]{}\state[s_1'']\optstep[1]{\actt}\state[s_1']$, 
      $\state[s_1'']\brel\states[2]$
    and
      $\state[s_1']\brel\state[s_2']$;
  \item\label{def:bbisim:term1}
    if $\term[1]{\states[1]}$, then there exists $\state[s_2']$
    such that
    $\states[2]\steps[2]{}\state[s_2']$,
    $\states[1]\brel\state[s_2']$ and
    $\term[2]{\state[s_2']}$;
    and
  \item\label{def:bbisim:term2}
    if $\term[2]{\states[2]}$, then there exists $\state[s_1']$
    such that
      $\states[1]\steps[1]{}\state[s_1']$,
      $\state[s_1']\brel\states[2]$
    and
      $\term[1]{\state[s_1']}$.
  \end{enumerate}
  The transition systems $\LTST[1]$ and $\LTST[2]$ are
  \emph{branching bisimilar} (notation: $\LTST[1]\bbisim\LTST[2]$)
  if there exists a branching bisimulation from $\LTST[1]$ to
  $\LTST[2]$ such that $\init[1]\brel\init[2]$.

  A branching bisimulation $\brel$ from $\LTST[1]$ to $\LTST[2]$ is
  \emph{divergence-preserving} if, for all states $\states[1]$ and
  $\states[2]$, $\states[1]\brel\states[2]$ implies
  \begin{enumerate}
  \addtocounter{enumi}{4}
  \item if there exists an infinite sequence
    ${(\states[1,i])}_{i\in\N}$ such that $\states[1]=\states[1,0]$,
    $\states[1,i]\step{\silent}\states[1,i+1]$ and
    $\states[1,i]\brel\states[2]$ for all $i\in\N$, then there
    exists a state $\state[s_2']$ such that
    $\states[2]\stepsplus{}\state[s_2']$ and
    $\states[1,i]\brel\state[s_2']$ for some $i\in\N$; and
  \item if there exists an infinite sequence
    ${(\states[2,i])}_{i\in\N}$ such that $\states[2]=\states[2,0]$,
    $\states[2,i]\step{\silent}\states[2,i+1]$ and
    $\states[1]\brel\states[2,i]$ for all $i\in\N$, then there
    exists a state $\state[s_1']$ such that
    $\states[1]\stepsplus{}\state[s_1']$ and
    $\state[s_1']\brel\states[2,i]$ for some $i\in\N$.
  \end{enumerate}
  The transition systems $\LTST[1]$ and $\LTST[2]$ are
  \emph{divergence-preserving branching bisimilar} (notation:
  $\LTST[1]\bbisimed\LTST[2]$) if there exists a divergence-preserving
  branching bisimulation from $\LTST[1]$ to $\LTST[2]$ such
  that $\init[1]\brel\init[2]$.
\end{definition}

The notions of branching bisimilarity and divergence-preserving
branching bisimilarity originate with~\cite{GW96}. The particular
divergence conditions we use to define divergence-preserving
branching bisimulations here are discussed in~\cite{GLT09}, where it
is also proved that divergence-preserving branching bisimilarity is an
equivalence.
 


An unobservable transition of an \RTM{}, i.e., a transition labelled
with $\silent$, may be thought of as an internal computation
step. Divergence-preserving branching bisimilarity allows us to
abstract from internal computations as long as they do not discard
the option to execute a certain behaviour. The following notion 
is used as a technical tool in the remainder of the paper.
\begin{definition}\label{def:comp}
 Given some transition system $\LTST$, an \emph{internal computation from
 state $\states$ to $\states'$} is a sequence of states
   $\states[1], \cdots, \states[n]$
 in $\LTST$ such that 
   $\states=\states[1]\step{\silent}\dots\step{\silent}\states[n]=\states'$.
 An internal computation is called \emph{deterministic} iff, for every
 state $\states[i]$ ($1\leq i < n$),
 $\states[i]\step{\acta}\states[i]'$ implies $\acta=\silent$ and
 $\states[i]'=\states[i+1]$.
 If $\states[1], \cdots, \states[n]$  is a deterministic internal
 computation from $\states$ to $\states'$, then we refer to the set
\begin{equation*}
   \{\states[1], \cdots, \states[n]\}
\end{equation*}
 as the set of \emph{intermediate states} of the deterministic
 internal computation.
\end{definition}

\begin{proposition}
  Let $\LTST$ be a transition system and let $\states$ and $\statet$
  be two states in $\LTST$. If there exists a deterministic internal
  computation from $\states$ to $\states'$, then all its intermediate
  states are related by the maximal divergence-preserving branching
  bisimulation on $\LTST$.
\end{proposition}
}

\section{Expressiveness of \RTM{}s}\label{sec:expr}

\fullonly{%
Our notion of \RTM{}s defines the class of executable transition
systems. In this section, we investigate the expressiveness of this
notion up to branching bisimilarity, using the notions of effective
transition system and computable transition system as a tool.

In Sect.~\ref{sec:expr:eff}, we recall the definitions of effective
transition system and computable transition system, observe that
executable transition systems are necessarily computable and,
moreover, have a bounded branching degree. Then, we proceed to
consider executable transition systems modulo (divergence-preserving)
branching bisimilarity. We present an example of a (non-effective)
transition system that is not executable up to branching
bisimilarity. Finally, we adapt a result by Phillips
\cite{Phi93} showing that every effective transition system is
branching bisimilar to a computable transition system with branching
degree at most two. Phillips' proof introduces divergence, and we
present an example illustrating that this is unavoidable.

In Sect.~\ref{sec:expr:bbctss},  we construct, for an arbitrary
boundedly branching computable transition system, an \RTM{} that
simulates the behaviour represented by the transition system up to
divergence-preserving branching bisimilarity. Thus, we confirm the
expressiveness of \RTM{}s: modulo divergence-preserving branching
bisimilarity, which is the finest behavioural equivalence in van
Glabbeek's spectrum \cite{Gla93}, the class of executable transition
systems coincides with the class of boundedly branching computable
transition systems. Moreover, in view of Phillips' result, we obtain
as a corrollary that every effective transition system can be
simulated  up to branching bisimilarity at the cost of introducing
divergence.

We obtain two more interesting corollaries from the result
in Sect.~\ref{sec:expr:bbctss}. Firstly, if a transition system is
deterministic, then, by our assumption that the set $\Act$ of action
symbols is finite, it is clearly boundedly branching; hence, every 
deterministic computable transition systems can be simulated, up to
divergence-preserving branching bisimilarity, by a deterministic
\RTM{}. Secondly, the parallel composition of boundedly branching
computable transition systems is clearly boundedly branching and
computable; hence, a parallel composition of \RTM{}s can be simulated,
up to divergence-preserving branching bisimilarity, by a single \RTM{}.
}%

\fctonly{%
We establish in this section that every effective
transition system can be simulated by an \RTM{} up to branching
bisimilarity, and that every boundedly branching computable transition
system can be simulated up to divergence-preserving branching
bisimilarity.  We use this as an auxiliary result to establish that a
parallel composition of \RTM{}s can be simulated by a single \RTM{}.
}

\subsection{Effective and Computable Transition Systems}\label{sec:expr:eff}

Let $\LTST=(\States,\trel,\init,\final)$ be a transition system;
the mapping $\outsetsym:\States\rightarrow 2^{\Actt\times\States}$ 
associates with every state its set of outgoing transitions, i.e., for all
$\states\in\States$,
\fullonly{%
\begin{equation*}
  \outset{\states}=\{(\actt,\statet)\mid \states\step{\actt}\statet\}
\enskip,
\end{equation*}}%
\fctonly{%
  $\outset{\states}=\{(\actt,\statet)\mid \states\step{\actt}\statet\}$,
}
and $\finaltestsym$ denotes the characteristic function of
$\final$. We restrict our attention in this section to
\emph{finitely branching} transition systems, i.e., transition systems
for which it holds that $\outset{\states}$ is finite for all states
$\states$. (The restriction is convenient for our definition of
computable transition system below, but it is otherwise unimportant
since in all our results about computable transition systems we 
further restrict to boundedly branching transition systems. The
restriction is not necessary for the definition of effective
transition system, and, in fact, our results about effective
transition systems do not depend on it.)

\begin{definition}\label{def:lts}
  Let $\LTST=(\States,\trel,\init,\final)$ be an $\Actt$-labelled
  finitely branching transition system.
  We say that $\LTST$ is \emph{effective} if $\trel$ and $\final$ are
  recursively enumerable sets.
  We say that $\LTST$ is \emph{computable} if $\outsetsym$ and
  $\finaltestsym$ are recursive functions.
\end{definition}
The notion of effective transition system originates with Boudol
\cite{Bou85}.
For the notion of computable transition system we have reformulated the
definition in \cite{BBK87} to suit our needs.
\fullonly{%
We temporarily step over the fact that, in order for the formal theory of
recursiveness to make sense, we need suitable codings into natural
numbers of the concepts involved. For now, we rely on the intuition of
the reader; in Sect.~\ref{sec:expr:bbctss} we return to this
issue in more detail.  (The reader may already want to
consult~\cite[\S 1.10]{Rog67} for more explanations.)}
A transition system is effective iff there exists an algorithm that
enumerates its transitions and an algorithm that enumerates its final
states. Similarly, a transition system is computable iff there exists
an algorithm that lists the outgoing transitions of a state and also
determines if it is final.

\begin{proposition}\label{prop:lts-RTM-comput}
  The transition system associated with an \RTM{} is computable.
\end{proposition}
\begin{proof}
  We omit a formal proof, but note that Definition~\ref{def:RTMopsem}
  describes the essence of algorithms for computing the
  outgoing transitions of a configuration and for determining if a
  configuration is final.
\end{proof}

Hence, unsurpisingly, if a transition system is not computable, then
it is not executable either. It is easy to define
transition systems that are not computable,
so there exist behaviours that are not executable. The
\fctonly{full version of this paper \cite{BLT11full} contains an}%
\fullonly{following}
example
\fullonly{takes this a little further and}%
\fctonly{that}
illustrates that there exist behaviours that are not even executable up to branching
bisimilarity.

\fullonly{%
\begin{example}\label{ex:notexecbbisim}
   (In this and later examples, we denote by $\varphi_x$ the partial recursive
  function with index $x\in\N$ in some exhaustive enumeration of
  partial recursive functions, see, e.g.,~\cite{Rog67}.)
  Assume that $\Act=\{\acta,\actb,\actc\}$ and consider the
  $\Act$-labelled transition system
     $\LTST[0]=(\StatesS[0],\trel[0],\init[0],\final[0])$
  with $\StatesS[0]$, $\trel[0]$, $\init[0]$ and $\final[0]$ defined by
  \begin{gather*}
    \StatesS[0]=\{\states,\statet,\stateu,\statev,\statew\}
                          \cup\{\states[x]\mid x \in \N\}\enskip,
\fctonly{\quad\init[0]=\states\enskip,\quad
  \final[0]=\{\statev,\statew\}\enskip,\ \text{and}}
\\
    \trel[0]=\{(\states,\acta,\statet),(\statet,\acta,\statet),(\statet,\actb,\statev),
                 (\states,\acta,\stateu),(\stateu,\acta,\stateu),(\stateu,\actc,\statew)\}\\
\qquad\qquad\mbox{}\cup
                 \{(\states,\acta,\states[0])\}
                    \cup
                 \{(\states[x],\acta,\states[x+1])\mid x\in\N\} \\
\qquad\qquad\mbox{}\cup
                 \{(\states[x],\acta,\statet),(\states[x],\acta,\stateu)\mid
                   \text{$\varphi_x$ is a total function}\}
\fctonly{\enskip.}%
\fullonly{%
\enskip,\\
    \init[0]=\states\enskip,\ \text{and}\\
    \final[0]=\{\statev,\statew\}\enskip.}
  \end{gather*}}%
\fullonly{The transition system is depicted in Fig.~\ref{fig:ltst0}.

  \begin{figure}[htb]
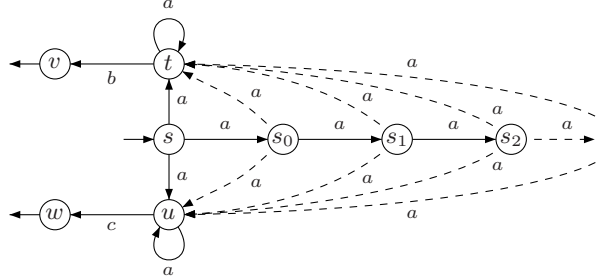

   \centering\small
   \begin{transsys}(75,30)(0,-5)
    \graphset{Nadjust=n,Nw=4,Nh=4,fangle=-180}
    \node[Nmarks=i](s)(15,10){$\states$}
    \node(t)(15,20){$\statet$}
    \node[Nmarks=f](v)(0,20){$\statev$}
    \node(s0)(30,10){$\states[0]$}
    \node(u)(15,0){$\stateu$}
    \node[Nmarks=f](w)(0,0){$\statew$}
    \node(s1)(45,10){$\states[1]$}
    \node(s2)(60,10){$\states[2]$}
    \node[Nframe=n,Nh=8,Nw=8](s_t)(75,10){}
    {\scriptsize
     \edge[ELside=r](s,t){$\acta$}
     \edge*[loopangle=90](t){$\acta$}
     \edge(t,v){$\actb$}
     \edge(s,u){$\acta$}
     \edge*[loopangle=-90](u){$\acta$}
     \edge(u,w){$\actc$}
     \edge(s,s0){$\acta$}
     \edge(s0,s1){$\acta$}
     \edge(s1,s2){$\acta$}
     \graphset{dash={1 0}{}}
     \edge(s2,s_t){$\acta$}
     \drawbpedge[ELpos=30,ELside=r](s0,135,10,t,-30,10){$\acta$}
     \drawbpedge[ELpos=30](s0,-135,10,u,30,10){$\acta$}
     \drawbpedge[ELpos=20,ELside=r](s1,135,10,t,0,10){$\acta$}
     \drawbpedge[ELpos=20](s1,-135,10,u,0,10){$\acta$}
     \drawbpedge[ELpos=10,ELside=r](s2,135,10,t,0,10){$\acta$}
     \drawbpedge[ELpos=10](s2,-135,10,u,0,10){$\acta$}
     \drawbpedge[ELside=r](s_t,120,10,t,0,10){$\acta$}
     \drawbpedge[dash={1 0}{}](s_t,-120,10,u,0,10){$\acta$}
    }
   \end{transsys}
   \caption{The transition system $\LTST[0]$.}\label{fig:ltst0}
  \end{figure}

  To argue that $\LTST[0]$ is not executable up to branching
  bisimilarity, we proceed by contradiction.
  Suppose that $\LTST[0]$ is executable up to branching bisimilarity.
  Then $\LTS[T_0]$ is branching
  bisimilar to a computable transition system $\LTS[T_0']$. Then, in
  $\LTS[T_0']$, the set of states reachable by a path that contains
  exactly $x$ $\acta$-transitions ($x\in\N$) and from which both a
  $\actb$- and a $\actc$-transition are still reachable, is
  recursively enumerable. It follows that the set of states in
  $\LTS[T_0']$ branching bisimilar to $\states[x]$ ($x\in\N$) is
  recursively enumerable. But then, since the problem of deciding
  whether from some state in $\LTS[T_0']$ there is a path containing
  exactly one $a$-transition and one $b$-transition such that the
  $a$-transition precedes the $b$-transition, is also recursively
  enumerable, it follows that the problem of deciding whether
  $\varphi_x$ is a total function must be recursively enumerable too, 
  which it is not.
  We conclude that $\LTST[0]$ is not executable up to branching
  bisimilarity. Incidentally, note that the language associated with
  $\LTST[0]$ is $\{a^n b,a^n c\mid n\geq 1\}$, which \emph{is}
  recursively enumerable (it is even context-free).
\end{example}
}

Phillips associates, in~\cite{Phi93}, with every effective transition
system a \emph{branching bisimilar} computable transition system of
which, moreover, every state has a branching degree of at most
$2$. (Phillips actually establishes weak bisimilarity, but it is easy to see
that branching bisimilarity holds.)

\begin{definition}
Let $\LTST=(\States,\trel,\init,\final)$ be a transition system, and
let $B$ be a natural number. We say that $\LTST$ has a branching degree
\emph{bounded by $\bdegbound$} if
  $|\outset{\states}|\leq\bdegbound$, for every state $\states\in\States$.
We say that $\LTST$ is \emph{boundedly branching} if there exists
$\bdegbound\in\N$ such that the branching degree of $\LTST$ is bounded
by $\bdegbound$.
\end{definition}

\begin{proposition}[Phillips]\label{prop:Phillips}
  For every effective transition system $\LTST$ there exists a boundedly
  branching computable transition system $\LTST'$ such that
    $\LTST\bbisim\LTST'$.
\end{proposition}

A crucial insight in Phillips' proof is that a divergence (i.e., an
infinite sequence of $\silent$-transitions) can be exploited to
simulate a state of which the set of outgoing transitions is
recursively enumerable, but not recursive. The following example,
inspired by \cite{Dar89}, shows that introducing divergence is
unavoidable.

\begin{example}\label{ex:divergence}
\fctonly{%
   (In this and later examples, we denote by $\varphi_x$ the partial recursive
  function with index $x\in\N$ in some exhaustive enumeration of
  partial recursive functions, see, e.g.,~\cite{Rog67}.)%
}
  Assume that $\Act=\{\acta,\actb\}$,
  and consider the transition system
    $\LTST[1]=(\StatesS[1],\trel[1],\init[1],\final[1])$
  with $\StatesS[1]$, $\trel[1]$, $\init[1]$ and $\final[1]$ defined by
  \begin{gather*}
    \StatesS[1]=\{\states[1,x],\ \statet[1,x]\mid x \in \N\}\enskip,
\fctonly{\quad\init[1]=\states[1,0]\enskip,}
\\
    \trel[1]=\{(\states[1,x],\acta,\states[1,x+1])\mid x\in\N\}\cup
             \{(\states[1,x],\actb,\statet[1,x])\mid x \in\N \}\enskip,
\fullonly{%
\\
    \init[1]=\states[1,0]\enskip,}\ \text{and}\\
    \final[1]=\{\statet[1,x]\mid\text{$\varphi_x(x)$ converges}\}\enskip.
  \end{gather*}
\fullonly{The transition system is depicted in Fig.~\ref{fig:ex-lts-eff}.

  \begin{figure}[htb]
    \begin{center}
    \begin{transsys}(60,15)(0,5)
      \graphset{fangle=-90}
      \node[Nmarks=i](s1_0)(0,16){$\states[1,0]$}
      \node(s1_1)(15,16){$\states[1,1]$}
      \node(s1_2)(30,16){$\states[1,2]$}
      \node(s1_3)(45,16){$\states[1,3]$}
      \node[Nframe=n](s1_t)(60,16){}
      \node(t_0)(0,6){$\statet[1,0]$}
      \node(t_1)(15,6){$\statet[1,1]$}
      \node(t_2)(30,6){$\statet[1,2]$}
      \node(t_3)(45,6){$\statet[1,3]$}
      \graphset{Nframe=n,Nw=0,Nh=0,Nadjust=n}
      \node(t_0t)(0,0){}
      \node(t_1t)(15,0){}
      \node(t_2t)(30,0){}
      \node(t_3t)(45,0){}
      {\scriptsize
       \edge(s1_0,s1_1){$\acta$}
       \edge(s1_1,s1_2){$\acta$}
       \edge(s1_2,s1_3){$\acta$}
       \edge[dash={1 0}{.2}](s1_3,s1_t){$\acta$}
       \edge(s1_0,t_0){$\actb$}
       \edge(s1_1,t_1){$\actb$}
       \edge(s1_2,t_2){$\actb$}
       \edge(s1_3,t_3){$\actb$}
       \graphset{dash={.5 0}{}}
       \edge(t_0,t_0t){}
       \edge(t_1,t_1t){}
       \edge(t_2,t_2t){}
       \edge(t_3,t_3t){}
      }
    \end{transsys}
    \end{center}
    \caption{The transition system $\LTST[1]$.}\label{fig:ex-lts-eff}
  \end{figure}}
\fctonly{%
  If $\LTST[2]$ is a transition system such that
    $\LTST[1]\bbisimed\LTST[2]$,
  as witnessed by some divergence-preserving branching bisimulation
  relation $\brel$, then it can be argued that $\LTST[2]$ is not
  computable. A detailed argument can be found in the full version~\cite{BLT11full}.
}%
\fctsubmissiononly{%
  For the convenience of the reviewer,
  the example is repeated in Appendix~\ref{app:effectivedivergence}
  with a more detailed argument.
}

\fullonly{%
  Now, suppose that $\LTST[2]$ is a transition system such that
    $\LTST[1]\bbisimed\LTST[2]$,
  as witnessed by some divergence-preserving branching bisimulation
  relation $\brel$; we argue that $\LTST[2]$ is not computable by
  deriving a contradiction from the assumption that it is.

  Clearly, since $\LTST[1]$ does not admit infinite
  sequences of $\silent$-transitions, if $\brel$ is
  divergence-preserving, then $\LTST[2]$ does not admit infinite
  sequences of $\silent$-transitions either. It follows that if
     $\states[1]\brel\states[2]$,
  then there exists a state $\state[s_2']$ in $\LTST[2]$ such that
    $\states[2]\steps[2]{}\state[s_2]$,
    $\states[1]\brel\state[s_2']$,
  and
    $\state[s_2']\nstep{\silent}$.
  Moreover, since $\LTST[2]$ is computable and does not admit infinite
  sequences of consecutive $\silent$-transitions, a state $\state[s_2']$
  satisfying the aforementioned properties is produced by the algorithm
  that, given a state of $\LTST[2]$, selects an enabled
  $\silent$-transition and recurses on the target of the transition
  until it reaches a state in which no $\silent$-transitions are enabled.

  But then we also have an algorithm that determines if $\varphi_x(x)$
  converges:
  \begin{enumerate}
  \item it starts from the initial state $\init[2]$ of $\LTST[2]$;
  \item it runs the algorithm to find a state without outgoing
      $\silent$-transitions, and then it repeats the following steps
      $x$ times:
    \begin{enumerate}
    \item execute the $\acta$-transition enabled in the reached state;
    \item run the algorithm to find a state without outgoing
      $\silent$-transitions again;
    \end{enumerate}
    since $\init[1]\brel\init[2]$, this yields a state $\states[2,x]$
    in $\LTST[2]$ such that $\states[1,x]\brel\states[2,x]$;
  \item it executes the $\actb$-transition that must be enabled
    in $\states[2,x]$, followed, again, by the algorithm to find a state
    without outgoing $\silent$-transitions; this yields a state
    $\statet[2,x]$, without any outgoing transitions, such that
    $\statet[1,x]\brel\statet[2,x]$.
  \end{enumerate}

  From $\statet[1,x]\brel\statet[2,x]$ it follows that
  $\statet[2,x]\in\final[2]$ iff $\varphi_x(x)$ converges, so the problem
  of deciding whether $\varphi_x(x)$ converges has been reduced to the
  problem of deciding whether $\statet[2,x]\in\final[2]$. Since it is
  undecidable if $\varphi_x(x)$ converges, it follows that $\final[2]$ is
  not recursive, which contradicts our assumption that $\LTST[2]$ is
  computable.
}
\end{example}

\subsection{Simulation of Boundedly Branching Computable Transition Systems}%
\label{sec:expr:bbctss}%

Let $\LTST = (\StatesS[\LTST], \trel[\LTST], \init[\LTST], \final[\LTST])$
be a boundedly branching computable transition system, say with branching
degree bounded by $\bdegbound$.
\fullonly{%
Our goal is to construct an \RTM{}
\fctonly{%
  $\BBCTSS = (\StatesS[\BBCTSS], \trel[\BBCTSS], \init[\BBCTSS],
  \final[\BBCTSS])$,
}\fullonly{%
\begin{equation*}
  \BBCTSS = (\StatesS[\BBCTSS], \trel[\BBCTSS], \init[\BBCTSS],
  \final[\BBCTSS])
\enskip,
\end{equation*}}
called the \emph{simulator} for $\LTST$, such that
 $\lts{\BBCTSS}\bbisimed\LTST$.

We have defined that $\LTST$ is computable if the associated mappings
  $\outsetsym$
and
  $\finaltestsym$
are recursive functions.
As is explained in \cite[\S 1.10]{Rog67}, the formal theory of
recursiveness can be applied to non-numerical functions (i.e., functions of which
domain and range are not the set of natural numbers), through codings
associating a unique natural number with every symbolic entity. In our
case, we fix codings of $\Actt$ and $\States$, i.e., injections
  $\actcodesym:\Actt\rightarrow\N$
and
  $\statecodesym:\States\rightarrow\N$
into the set of natural numbers $\N$.
We use these codings, and standard techniques for coding and decoding
tuples of natural numbers and finite sets of natural
numbers\footnote{%
  See, e.g., \cite[\S 5.3]{Rog67} and \cite[\S 5.6]{Rog67},
  respectively.
}
to define partial recursive functions $\outsetsym$ and $\finaltestsym$
on natural numbers:
\begin{itemize}
\item $\outsetsym:\N\rightharpoonup\N$ is the partial function
  that, for all states $\states$, maps $\statecode{\states}$ to the code
  associated with $\outset{\states}$ and is undefined for all natural
   numbers that are not codes of states; and
\item $\finaltestsym:\N\rightharpoonup\N$ is the partial function
     that maps $\statecode{\states}$ to $\finaltestsym(\states)$\footnote{Recall that
  $\finaltestsym$ on states is a characteristic function and hence
  already yields a natural number.} and is undefined on natural
  numbers that are not codes of states.
\end{itemize}
}

\fctonly{%
It is reasonably straightforward to construct an \RTM{}
 $\BBCTSS = (\StatesS[\BBCTSS], \trel[\BBCTSS], \init[\BBCTSS], \final[\BBCTSS])$,
which we call the \emph{simulator} for $\LTST$, such that
  $\lts{\BBCTSS}\bbisimed\LTST$.
The construction is detailed in~\cite{BLT11full}}\fctsubmissiononly{; here we  only give a brief
description}\fctonly{.}\fctsubmissiononly{ The simulator $\BBCTSS$ for $\LTST$ consists of three parts.
\begin{enumerate}
\item The \emph{initialiation fragment} prepares the tape of $\BBCTSS$
  for the simulation of $\LTST$. It starts by writing codes of the
  transition relation and the set of final states on a reserved
  portion of the tape. Then, the code of the initial state is written
  on another part of the tape, reserved for storing the \emph{current
      state} of the simulation.
\item The \emph{state fragment} computes what is called the
  \emph{menu} of the state, i.e., it computes whether the current state is
  terminating and what are its possible next transitions. This
  information is then written on the tape in coded form.
\item The \emph{step fragment} selects a transition to execute among all possible next transitions. Only in the construction of this fragment
  we actually use our assumption that $\LTST$ is boundedly
  branching. A crucial aspect of branching bisimilarity is that the
  choice for the next transition should be made in a single
  state. Thus, the state fragment includes a state $\states[m]$ for every possible
  menu $m$, and starts by decoding the menu of the tape, moving to the
  state corresponding with the menu on the tape. The state
  $\states[m]$ is a final state only if it is final according to the
  menu $m$, and it has exactly as many outgoing $\act$-transition as
  specified by the menu $m$. After executing one of the transitions
  specified by the menu, the step fragment computes the next current
  state and returns to the state fragment.
\end{enumerate}
}

\fullonly{%
For the remainder of this paper we fix an enumeration of the partial
recursive functions, and we denote by $\code{\outsetsym}$ and
$\code{\finaltestsym}$ indices of the partial recursive functions
$\outsetsym:\N\rightharpoonup\N$ and
$\finaltestsym:\N\rightharpoonup\N$ in this enumeration. Instead of
hardcoding computations of $\outsetsym$ and $\finaltestsym$ in the
\RTM{} $\BBCTSS$ to be constructed, we prefer to store their codes
$\code{\outsetsym}$ and $\code{\finaltestsym}$ on the tape and
interpret these codes.  This is slightly more generic than necessary
for simulation of the presupposed transition system $\LTST$, but the
genericity will pay off when we extend the simulator
to obtain a universal \RTM{} in Sect.~\ref{sec:universal}.

We are going to define $\BBCTSS$ as the union of three fragments, each
with a different purpose.
\begin{enumerate}
\item The \emph{initialiation fragment} prepares the tape of $\BBCTSS$
  for the simulation of $\LTST$, writing the codes $\code{\outsetsym}$,
  $\code{\finaltestsym}$ and $\statecode{\init[\LTST]}$ to the tape.
\item In the \emph{state fragment} the behaviour in the current state
  (i.e., whether it is terminating and what are its possible next
  transitions) is computed, and stored on the tape in coded form.
\item The \emph{step fragment} first decodes the information on the
  tape about the behaviour of the current state as computed in the
  state fragment, moving to a special selection state of $\BBCTSS$ that
  corresponds with the coded behaviour. (A crucial aspect of branching
  bisimilarity is that the choice for the next transition should be
  made in a single state. By our assumption that $\LTST$ is boundedly
  branching, we need to include only finitely many such selection states.)
  The aforementioned selection state has each of the possible outgoing
  transitions. After executing one of these transitions, the code of
  the target state, being the new current state, is written on the tape.
\end{enumerate}

\todo{Revise following paragraph}
Below we present a detailed description of the construction of
$\BBCTSS$. We first briefly discuss how we use the tape
to store information regarding the current state and its
behaviour. The implementation of the fragments involve several
straightforward computational tasks on the contents of 
the tape. We do not dwell on the details of implementing these
tasks; we just presuppose the existence of auxiliary
deterministic Turing machines capable of carrying them out.
Then, we discuss the implementation of the three
fragments decribed above. 

\paragraph{Tape}
In the above, we have declared codes for actions and states, and for
the partial recursive functions
  $\outsetsym:\N\rightharpoonup\N$
and
  $\finaltestsym:\N\rightharpoonup\N$.
The way in which natural numbers are represented as sequences over
some finite alphabet of tape symbols is largely irrelevant, but in our
construction below it is sometimes convenient to have an explicit
representation.  In such cases, we assume that numbers are stored in
unary notation using the symbol $\tpunit$.  That is, a natural number
$n$ is represented on the tape as the sequence $\tpunit^{n+1}$ of
$n+1$ occurrences of the symbol $\tpunit$.  
In addition to the symbol $\tpunit$, 
we use the symbols $\tpleft$ and $\tpright$ to delimit the codes
of $\outsetsym$ and $\finaltestsym$ that remain on the tape
throughout the simulation, $\tpsepc$ to separate the elements of a
tuple of natural numbers, and $\tpsep$ to separate tuples.
The simulator $\BBCTSS$ constructed below incorporates the
operation of some auxiliary Turing machines that may require the use
of some additional symbols; let $\Data'$ be the collection of all
these extra symbols. Then the tape alphabet $\Data$ of $\BBCTSS$ is
\begin{equation*}
  \Data=\{\tpunit,\tpleft,\tpright,\tpsepc,\tpsep\}\cup\Data'
 \enskip.
\end{equation*}

\paragraph{Auxiliary Turing machines}

\todo{Perhaps the definitions below should be integrated in Sect.~\ref{sec:rtm}.}

For our purposes, it is convenient to define a \emph{deterministic
  Turing machine} \TMM{} as a quadruple
  $\TMM=(\StatesS[\TMM{}],\trel[\TMM{}],\init[\TMM{}],\final[\TMM{}])$
with
  $\StatesS$ its set of \emph{states}, 
\begin{equation*}
  \trel[\TMM]\subseteq\StatesS[\TMM{}]\times\Datab\times\Datab\times\{\tpmvL,\tpmvR\}\times\StatesS[\TMM{}]
\end{equation*}
  its \emph{transition relation},
  $\init[\TMM{}]$ its \emph{initial state}, and
  $\final[\TMM{}]$ its (unique) \emph{final state}.
We assume that $\TMM{}$ satisfies the following requirements:
\begin{enumerate}
\renewcommand{\theenumi}{\roman{enumi}}
\renewcommand{\labelenumi}{(\theenumi)}
\item for every pair
  $(\states,\datumd)\in(\StatesS\setminus\{\final[\TMM]\})\times\Datab$
  there is precisely one triple
    $(\datume,\tpmv{M},\states')\in\Datab\times\{\tpmvL,\tpmvR\}\times\StatesS$ 
  such that
    $(\states,\datumd,\datume,\tpmv{M},\states')\in\trel[\TMM{}]$; and
\item there do not exist $\datumd, \datume\in\Datab$,
  $\tpmv{M}\in\{\tpmvL,\tpmvR\}$ and $\states\in\StatesS$ such that
  $(\final[\TMM{}],\datumd,\datume,\tpmv{M},\states)\in\trel[\TMM{}]$.
\end{enumerate}
Our definition of deterministic Turing machine is non-standard in
assuming that whenever it halts, it does so in the special
distinguished final state. This assumption is convenient when we
incorporate the functionality implemented by a Turing machine in the
definition of our simulator, below. To be convinced that our
assumption does not limit the computational expressiveness of our
notion of Turing machine, the reader may want to compare our variant
with the one described in \cite[p.\ 13--16]{Rog67}. The latter does not
have a distinguished halting state, but to convert it to one that
satisfies our requirements, it suffices to add (in the notation of
\cite{Rog67})  an internal state $q_f$ and a quadruple
$q_i,\datumd,\datumd,q_f$ for every combination of $q_i$ and $\datumd$
not already appearing as first two elements of a quadruple.

Note that a Turing machine can be viewed as an \RTM{} without the
$\Actt$ labels associated with transitions (see
Definition~\ref{def:RTM}) and with a singleton set of final
states. Similarly as for \RTM{}s (see Definition~\ref{def:RTMconf}), a
\emph{configuration} of a Turing machine is a pair $(\states,\tpstrd)$
consisting of a state $\states\in\StatesS$ and a tape instance
$\tpstrd$, and the transition relation $\trel[\TMM{}]$ of $\TMM{}$ induces
an unlabelled transition relation $\step{}$ on configurations, defined
as in Definition~\ref{def:RTMopsem} (simply omit all references to
$\Actt$ and its elements).

Let $\tpstrd$ and $\tpstrd'$ be tape instances.
By an \emph{$\TMM{}$-computation} from $\tpstrd$ to $\tpstrd'$ we
understand a sequence of configurations
\begin{equation*}
  (\states[0],\tpstrd[0]),\dots,(\states[n],\tpstrd[n])
\end{equation*}
such that
  $\states[0]=\init[\TMM{}]$,
  $\tpstrd[0]=\tpstrd$
  $\states[n]=\final[\TMM{}]$,
  $\tpstrd[n]=\tpstrd'$, and
  $(\states[i],\tpstrd[i])\step{}(\states[i+1],\tpstrd[i+1])$ for all
  $0\leq i< n$.
We write $\tpstrd\TMcomp[\TMM{}]\tpstrd'$ if there exists an
$\TMM{}$-computation from $\tpstrd$ to $\tpstrd'$.


\paragraph{Initialisation fragment}

Note that it is straightforward to define a conventional
deterministic Turing machine
  $\INIT=(\StatesS[\INIT],\trel[\INIT],\init[\INIT],\final[\INIT])$
that, when started on an empty tape, writes the given natural numbers
$\code{\outsetsym}$, $\code{\finaltestsym}$ and
$\statecode{\init[\LTST]}$ on the tape in a suitable representation,
yielding the tape instance
\begin{equation*}
  \tpleft\code{\outsetsym}\tpsepc\code{\finaltestsym}\tphd{\tpright}
    \statecode{\init[\LTST]}
\enskip.
\end{equation*}
We use $\INIT$ to define the \emph{initialisation fragment}
$\InF$. The set of states of $\InF$ is defined as
\begin{equation*}
  \InFStates=
    \StatesS[\INIT]\setminus\final[\INIT]
\enskip,
\end{equation*}
its initial state is defined as
\begin{equation*}
  \InFinit=\init[\INIT]
\enskip;\ \text{and}
\end{equation*}
its set of transitions is defined as
\begin{align*}
  \InFtrel=\mbox{}
    & \{(\InFstate,\datumd,\silent,\datume,\tpmv{M},\InFstate')
              \mid(\InFstate,\datumd,\datume,\tpmv{M},\InFstate')\in\trel[\INIT],\
                  \InFstate'\in\StatesS[\INIT]\setminus\final[\INIT]\}\\
    &\ \mbox{} \cup
           \{(\InFstate,\datumd,\silent,\datume,\tpmv{M},\StFinit)
             \mid(\InFstate,\datumd,\datume,\tpmv{M},\InFstate')\in\trel[\INIT],\
                  \InFstate'\in\final[\INIT]\} 
\enskip.
\end{align*}
(Note that $\StFinit$ is not a state in $\StatesS[\INIT]$; it is the
initial state of the state fragment to be defined next.)

%

\begin{fact}\label{fact:confs:init}
 The fragment $\InF$ gives rise to a deterministic internal computation from
   $(\InFinit, \tphd{\tpblank})$
 to
   $(\StFinit, \tpleft\code{\outsetsym}\tpsepc\code{\finaltestsym}
               \tphd{\tpright}
               \statecode{\init[\LTST]})$;
  we denote its set of intermediate states by $\IS{\InF}$.
\end{fact}

\paragraph{State fragment}

The \emph{state fragment} $\StF$ replaces the code of the current
state on the tape by a sequence of codes that represents the
behaviour of $\LTST$ in the current state. It is assumed that it starts
with a tape instance of the form
  $\tpleft\code{\outsetsym}\tpsepc\code{\finaltestsym}\tphd{\tpright}
   \statecode{\states}$
for some
  $\states\in\StatesS[\LTST]$.

Recall that $\code{\outsetsym}$ and $\code{\finaltestsym}$ are indices
of the partial recursive functions $\outsetsym:\N\rightarrow\N$ and
$\finaltestsym:\N\rightarrow\N$ in some fixed enumeration of the
partial recursive functions. Hence, there exists a Turing machine
  $\STATE=(\StatesS[\STATE],\trel[\STATE],\init[\STATE],\final[\STATE])$
that interprets the codes $\code{\outsetsym}$ and
$\code{\finaltestsym}$, applies the corresponding partial recursive
functions $\outsetsym$ and $\finaltestsym$ to $\statecode{\states}$,
and decodes the code of the finite set of pairs yielded by the
function $\outsetsym$ into a list of codes of actions and target
states. Without loss of generality, we may assume that $\STATE$, when
started on a tape instance of the form
\begin{equation*}
  \tpleft\code{\outsetsym}\tpsepc\code{\finaltestsym}\tphd{\tpright}
     \statecode{\states}
\enskip,
\end{equation*}
performs a terminating deterministic computation that yields the tape
instance 
\begin{equation*}
  \tpleft\code{\outsetsym}\tpsepc\code{\finaltestsym}\tphd{\tpright}
  \finaltest{\states}
  \tpsepc\statecode{\acta[1]}
  \tpsepc\cdots\tpsepc\statecode{\acta[k]}
  \tpsep\statecode{\states[1]}
  \tpsepc\cdots\tpsepc\statecode{\states[k]}
\enskip,
\end{equation*}
where
  $\outset{\states}=\{(\acta[i],\states[i])\mid 1\leq i \leq k\}$.
Note that, since the branching degree of $\LTST$ is bounded by $\bdegbound$,
we have that $k \leq \bdegbound$.  
Henceforth, we refer to the sequence
  $\finaltest{\states},\acta[1],\dots,\acta[k]$
generated and stored on the tape by $\STATE$ as the \emph{menu} in
$\states$.

The set of states of $\StF$ is defined as
\begin{equation*}
  \StFStates=\StatesS[\STATE{}]\setminus\final[\STATE{}]
\enskip;
\end{equation*}
its initial state is defined as
\begin{equation*}
  \StFinit=\init[\STATE{}]
\enskip;\ \text{and}
\end{equation*}
its set of transitions is defined as
\begin{align*}
  \StFtrel=\mbox{}
    & \{(\StFstate,\datumd,\silent,\datume,\tpmv{M},\StFstate')
              \mid(\StFstate,\datumd,\datume,\tpmv{M},\StFstate')\in\trel[\STATE],\
                  \StFstate'\in\StatesS[\STATE{}]\setminus\final[\STATE]\}\\
    &\ \mbox{} \cup
           \{(\StFstate,\datumd,\silent,\datume,\tpmv{M},\SpFinit)
              \mid(\StFstate,\datumd,\datume,\tpmv{M},\StFstate')\in\trel[\STATE],\
                  \StFstate'\in\final[\STATE]\}
\enskip.       
\end{align*}
(Again, note that $\SpFinit$ is not a state in $\StFStates$, but the
initial state of the step fragment to be defined next.)

%

\begin{fact}\label{fact:confs:state}
 Let $\states\in\StatesS[\LTST]$, let $0\leq k \leq B$, let
 $\acta[1],\dots,\acta[k]\in\Actt$ and
 $\states[1],\dots,\states[k]\in\StatesS[\LTST]$ such that
   $\outset{\states}=\{(\acta[i],\states[i])\mid 1\leq i \leq k\}$.
 Then the fragment $\StF$ gives rise to a deterministic computation
 from
 \begin{equation*}
   (\StFinit, \tpleft\code{\outsetsym}\tpsepc\code{\finaltestsym}
              \tphd{\tpright}
              \statecode{\states})
 \end{equation*}
 to
 \begin{equation*}
   (\SpFinit,
    \tpleft\code{\outsetsym}\tpsepc\code{\finaltestsym}\tphd{\tpright}
    \finaltest{\states}
    \tpsepc\statecode{\acta[1]}
    \tpsepc\cdots\tpsepc\statecode{\acta[k]}
    \tpsep\statecode{\states[1]}
    \tpsepc\cdots\tpsepc\statecode{\states[k]})
 \enskip;
 \end{equation*}
 we denote its set of intermediate states by
   $\IS{\StF,\states}$.
\end{fact}

\paragraph{Step fragment}

\newcommand{\SpFnext}[1][]{\ensuremath{\mathit{ne}_{#1}}}
\newcommand{\SpFcr}[1][]{\ensuremath{\mathit{cr}_{#1}}}

The purpose of the \emph{step fragment} $\SpF$ is to select a
transition enabled in the current state $\states$, execute the
corresponding action, and remove  
  $\finaltest{\states}$
and all codes of actions and states from the tape, except the code of
the target state of the selected transition.

The behaviour represented by the simulated transition system $\LTST$
when it is in state $\states$ consists of a non-deterministic choice
between its $k$ outgoing transitions
  $\states\step{\acta[1]}\states[1], \dots,
    \states\step{\acta[k]}\states[k]$
and it is terminating if $\finaltest{\states}=1$. To get a branching
bisimulation between $\LTST$ and the transition system associated with
$\BBCTSS$, the latter necessarily has to include a configuration
offering exactly the same choice of outgoing transitions and the same
termination behaviour. (It is important to note that branching
bisimilarity does not, e.g., allow the choice for one of the outgoing
transitions to be made by an internal computation that eliminates
options one by one.) The fragment $\SpF$ therefore includes one
special state
  $\SpFstate[{\finaltest{\states},\acta[1],\dots,\acta[k]}]$
for every potential menu.
Since $k\leq \bdegbound$, the branching degree bound of $\LTST$,
there are
  $N=\sum_{k=0}^{B} 2\cdot |\Actt|^k$
potential menus.

The functionality of the step fragment consists of two parts.  The
first part decodes the menu on the tape ending up in a state
  $\SpFstate[{\finaltest{\states},\acta[1],\dots,\acta[k]}]$.
The second part takes care of the execution of an enabled transition
and reinitialising the simulation with the target state of the
executed transition as the new current state.

Let $\STEP=(\StatesS[\STEP],\trel[\STEP],\init[\STEP],\final[\STEP])$
be a deterministic Turing machine with distinguished states
$\SpFstate[{\finaltest{\states},\acta[1],\dots,\acta[k]}]$ (one for
every potential menu) that, when started on a tape instance
 \begin{equation*}
   \tpleft\code{\outsetsym}\tpsepc\code{\finaltestsym}\tphd{\tpright}
   \finaltest{\states}
   \tpsepc\statecode{\acta[1]}
   \tpsepc\cdots\tpsepc\statecode{\acta[k]}
   \tpsep\statecode{\states[1]}
   \tpsepc\cdots\tpsepc\statecode{\states[k]}
 \end{equation*}
performs a deterministic computation that halts in the state
  $\SpFstate[{\finaltest{\states},\acta[1],\dots,\acta[k]}]$
with tape instance
  $\tpleft\code{\outsetsym}\tpsepc\code{\finaltestsym}\tphd{\tpright}
    \statecode{\states[1]}\tpsepc\cdots\tpsepc\statecode{\states[k]}$.
Note that we can assume, without loss of generality, that
\begin{equation*}
  \final[\STEP]=\{\SpFstate[{t,\acta[1],\dots,\acta[k]}]\mid
    t\in\{0,1\},\ 0\leq k \leq B,\ \acta[1],\dots,\acta[k]\in\Actt\}
\enskip.
\end{equation*}

The state $\SpFstate[{\finaltest{\states},\acta[1],\dots,\acta[k]}]$
is declared final iff $\finaltest{\states}=1$, and it has
$k$ outgoing transitions labelled $\acta[1],\dots,\acta[k]$,
respectively. After performing the $i$th transition labelled with
$\acta[i]$, the list of codes of states 
  $\statecode{\states[1]}\tpsepc\cdots\tpsepc\statecode{\states[k]}$
remaining on the tape should be replaced by the $i$th code in
the list, after which the simulation returns to the state
fragment. For each $1\leq i \leq B$, let
  $\REINIT[i]=(\StatesS[{\REINIT[i]}],\trel[{\REINIT[i]}],\init[{\REINIT[i]}],\final[{\REINIT[i]}])$
be a deterministic Turing machine that, when started on a tape instance of
the form
 \begin{equation*}
   \tpleft\code{\outsetsym}\tpsepc\code{\finaltestsym}\tpright
   \tphdR{\statecode{\states[1]}}
   \tpsepc\cdots\tpsepc\statecode{\states[k]}
\qquad (k\geq i)
 \end{equation*}
halts with a tape instance
 \begin{equation*}
   \tpleft\code{\outsetsym}\tpsepc\code{\finaltestsym}\tphd{\tpright}
   \statecode{\states[i]}
\enskip.
 \end{equation*}

The set of states of $\SpF$ is defined as
\begin{equation*}
  \SpFStates=
    (\StatesS[\STEP]
      \cup
     \textstyle\bigcup_{i=1}^B\StatesS[{\REINIT[i]}])
      \setminus\textstyle\bigcup_{i=1}^B\final[{\REINIT[i]}]
 \enskip;
\end{equation*}
its initial state is defined as
\begin{equation*}
  \SpFinit=\init[\STEP{}]
\enskip;\ \text{and}
\end{equation*}
its set of transitions is defined as
\begin{align*}
  \SpFtrel=\mbox{}
    & \{(\SpFstate,\datumd,\silent,\datume,\tpmv{M},\SpFstate')
              \mid(\SpFstate,\datumd,\datume,\tpmv{M},\SpFstate')\in\trel[\STEP]\}\\
    &\ \mbox{} \cup
           \{(\SpFstate[{t,\acta[1],\dots,\acta[k]}],\tpright,\acta[i],\tpright,\tpmvR,\init[{\REINIT[i]}])\\
    & \mbox{}\qquad\qquad\qquad\qquad
             \mid t\in\{0,1\},
                  \acta[1],\dots,\acta[k] \in \Actt,
                  k \leq \bdegbound, 1 \leq i\leq k\} \\
    &\ \mbox{} \cup \textstyle\bigcup_{i=1}^B
       \{(\SpFstate,\datumd,\silent,\datume,\tpmv{M},\SpFstate') \\
    & \mbox{}\qquad\qquad \qquad\qquad
              \mid(\SpFstate,\datumd,\datume,\tpmv{M},\SpFstate')\in\trel[{\REINIT[i]}],\
                  \SpFstate'\in\StatesS[{\REINIT[i]}]\setminus\final[{\REINIT[i]}]\}\\
    &\ \mbox{} \cup \textstyle\bigcup_{i=1}^B
           \{(\SpFstate,\datumd,\silent,\datume,\tpmv{M},\StFinit) \\
    & \mbox{}\qquad\qquad \qquad\qquad
              \mid(\SpFstate,\datumd,\datume,\tpmv{M},\SpFstate')\in\trel[{\REINIT[i]}],\
                  \SpFstate'\in\final[{\REINIT[i]}]\}
\enskip.       
\end{align*}

\begin{fact}\label{lem:confs:step}
 Let $\states\in\StatesS[\LTST]$, let $0\leq k \leq B$, let
 $\acta[1],\dots,\acta[k]\in\Actt$ and
 $\states[1],\dots,\states[k]\in\StatesS[\LTST]$ such that
   $\outset{\states}=\{(\acta[i],\states[i])\mid 1\leq i \leq k\}$.
  Then the fragment $\SpF$ gives rise to the following deterministic
  internal computations:
\begin{enumerate}
\renewcommand{\theenumi}{\roman{enumi}}
\renewcommand{\labelenumi}{(\theenumi)}
\item a deterministic internal computation from
  \begin{equation*}
    (\SpFinit,
     \tpleft\code{\outsetsym}\tpsepc\code{\finaltestsym}\tphd{\tpright}
     \finaltest{\states}
     \tpsepc\statecode{\acta[1]}
     \tpsepc\cdots\tpsepc\statecode{\acta[k]}
     \tpsep\statecode{\states[1]}
     \tpsepc\cdots\tpsepc\statecode{\states[k]})
  \end{equation*}
 to
  \begin{equation*}
   (\SpFstate[{\finaltest{\states},\acta[1],\dots,\acta[k]}],
    \tpleft\code{\outsetsym}\tpsepc\code{\finaltestsym}\tphd{\tpright}
    \statecode{\states[1]}
    \tpsepc\cdots\tpsepc\statecode{\states[k]})
\enskip;
  \end{equation*}
  we denote its set of intermediate states by $\IS{\SpF,1,\states}$;
  and
\item a deterministic internal computation from
  \begin{equation*}
    (\init[{\REINIT[i]}],
     \tpleft\code{\outsetsym}\tpsepc\code{\finaltestsym}\tpright
     \tphdR{\statecode{\states[1]}}
     \tpsepc\cdots\tpsepc\statecode{\states[k]})
  \end{equation*}
 to
  \begin{equation*}
    (\StFinit,
     \tpleft\code{\outsetsym}\tpsepc\code{\finaltestsym}\tphd{\tpright}
     \statecode{\states[i]})
  \enskip;
  \end{equation*}
  we denote its set of intermediate states by $\IS{\SpF,2,\states[i]}$.
\end{enumerate}
\end{fact}

\paragraph{Simulator}

The \emph{simulator}
  $\BBCTSS=(\BBCTSSStates,\BBCTSStrel,\BBCTSSinit,\BBCTSSfinal)$
for $\LTST$ is defined as the union of the fragments $\InF$, $\StF$
and $\SpF$ defined above:
the set of states of $\BBCTSS$ is defined as the union of the sets of
states of all fragments
\begin{equation*}
  \BBCTSSStates=
    \InFStates
  \cup
    \StFStates
  \cup
    \SpFStates
\enskip;
\end{equation*}
the transition relation of $\BBCTSS$ is the union of the
transition relations of all fragments
\begin{equation*}
  \BBCTSSStates=
    \InFtrel
  \cup
    \StFtrel
  \cup
    \SpFtrel
\enskip;
\end{equation*}
the initial state of $\BBCTSS$ is the initial state of $\InF$
\begin{equation*}
  \BBCTSSinit=\InFinit
\enskip;\ \text{and}
\end{equation*}
the set of final states $\BBCTSSfinal$ of $\BBCTSS$ is
\begin{equation*}
  \BBCTSSfinal=\{\SpFstate[{1,\acta[1],\dots,\acta[k]}]
                 \mid 0\leq k\leq B\ \&\ \acta[1],\dots,\acta[k]\in\Actt\}
\enskip.
\end{equation*}



}

\begin{theorem}\label{thm:BBCTSS}
  For every boundedly branching computable transition system $\LTST$
  there exists an \RTM{} $\BBCTSS$ such that
    $\LTST \bbisimed \lts{\BBCTSS}$.
\end{theorem}
\fullonly{%
\begin{proof}
 Consider the \RTM{} $\BBCTSS$ of which the definition is sketched
 above. Referring to Fact~\ref{fact:confs:init} we define the following relation:
 \begin{equation*}
   \brelsym[\init] = \{(\init[\LTST], \statet) \mid \statet \in \IS{\InF} \}
 \enskip,
 \end{equation*}
 and referring to Facts~\ref{fact:confs:state} and
 \ref{lem:confs:step}, we define, for every
 $\states\in\StatesS[\LTST]$, the relation
 \begin{equation*}
   \brelsym[\states] = 
      \{(\states, \statet) \mid \statet \in \IS{\StF,\states}\cup\IS{\SpF,1,\states}\cup\IS{\SpF,2,\states}\}
 \enskip.
 \end{equation*}
 Then it can be verified straightforwardly that the binary relation
 \begin{equation*}
    \brelsym = \brelsym[\init] \cup 
               \textstyle\bigcup_{\states\in\StatesS[\LTST]} \brelsym[\states]
  \end{equation*}
  is a divergence-preserving branching bisimulation from $\LTST$ to
  $\lts{\BBCTSS}$.
  Since $(\init[\BBCTSS],\tphd{\tpblank})\in\IS{\InF}$, it follows
  that $(\init[\LTST],(\init[\BBCTSS],\tphd{\tpblank}))\in\brelsym$,
  and hence $\LTST\bbisimed\lts{\BBCTSS}$.
\todo{Discuss divergence!}
\end{proof}%
}

Recall that, by Proposition~\ref{prop:Phillips}, every effective
transition system is branching bisimilar to a computable transition
system with branching degree bounded by $2$. According to
Theorem~\ref{thm:BBCTSS}, the resulting transition can be simulated
with an \RTM{} up to divergence-preserving branching
bisimilarity. We can conclude that \RTM{}s can simulate effective
transition systems up to branching bisimilarity, but, in view of
Example~\ref{ex:divergence}, not in a divergence-preserving manner.

\begin{corollary}\label{cor:BBCTSS-eff}
  For every effective transition system $\LTST$ there exists a reactive
  Turing machine $\BBCTSS$ such that
    $\LTST \bbisim \lts{\BBCTSS}$.
\end{corollary}

Note that if $\LTST$ is deterministic, then $|\outset{\states}|\leq
|\Actt|$ for every state $\states$ in $\LTST$, so every deterministic
transition system is, in fact, boundedly branching.
\fullonly{%
Furthermore, since all internal computations involved in the
simulation of a boundedly branching $\LTST$ are deterministic,
if $\BBCTSS$ is non-deterministic, then this can only be due to a state
   $\SpFstate[{\finaltest{\states},\acta[1],\dots,\acta[k]}]$%
with $\acta[i]=\acta[j]$ for some $1 \leq i < j\leq k$.}%
\fctonly{%
All computations involved in the simulation of $\LTST$ are
deterministic; if $\BBCTSS$ is non-deterministic, then this is due
to a state of which the menu includes some action $\acta$ more than
once.}
It follows that a deterministic computable transition system can be
simulated up to divergence-preserving branching bisimilarity by a
deterministic \RTM{}.
The following corollary to Theorem~\ref{thm:BBCTSS} summarises the
argument.
\begin{corollary}\label{cor:BBCTSS-dc}
  For every deterministic computable transition system $\LTST$ there exists
  a deterministic \RTM{} \TM{} such that
    $\lts{\TM{}}\bbisimed\LTST$.
\end{corollary}

Using Theorem~\ref{thm:BBCTSS} we can now also establish that a
parallel composition of \RTM{}s can be simulated, up to
divergence-preserving branching bisimilarity, by a single \RTM{}.  To
this end, note that the transition systems associated with \RTM{}s are
boundedly branching and computable. Further note that the parallel
composition of boundedly branching computable transition systems is
again computable. It follows that the transition system associated
with a parallel composition of \RTM{}s is boundedly branching and
computable, and hence, by Theorem~\ref{thm:BBCTSS}, there exists an
\RTM{} that simulates it up to divergence-preserving branching
bisimilarity. Thus we get the following corollary.

\begin{corollary}
  For every pair of \RTM{}s $\TM_{1}$ and $\TM_{2}$
  and for every set of communication channels $\Chan$
  there is an \RTM{} $\TM$ such that
    $\lts{\TM{}}\bbisimed\lts{\aep{\Chan}{\TM_{1}\rtmpar{}\TM_{2}}}$.
\end{corollary}

\section{Universality}\label{sec:universal}

Recall that a \emph{universal Turing machine} is a Turing machine that
can simulate an arbitrary Turing machine on arbitrary input. The
assumptions are that a finite description of the to be simulated
Turing machine (e.g., a G\"odel number, see \cite{Rog67}) as well as
its input are available on the tape of the universal Turing machine, and the
simulation is up to functional or language equivalence. We adapt this
scheme in two ways. Firstly, we let the simulation start by inputting
the description of an arbitrary \RTM{} \TMM{} along some dedicated
channel $\chanrtm{}$, rather than assuming its presence on the tape
right from the start. This is both conceptually desirable ---for our
aim is to give interaction a formal status---, and technically
necessary ---for in the semantics of \RTM{}s we have assumed that the
tape is initially empty.  Secondly, we require the behaviour of
$\TMM{}$ to be simulated up to divergence-preserving branching
bisimilarity.

Thus, we arrive at the following tentative definitions. For an
arbitrary \RTM{} \TMM{}, denote by $\TMsend{\TMM}$ a deterministic
\RTM{} with no other behaviour than outputting a G\"odel number
$\code{\TMM{}}$ of \TMM{} in an appropriate representation along
channel $\chanrtm{}$ after which it halts in its unique final state. A
\emph{universal} \RTM{} is then an \RTM{} $\URTM$ such that, for every
\RTM{}~$\TMM$, the parallel composition $\aep{\{\chanrtm\}}{\URTM\parc
  \TMsend{\TMM}}$ simulates $\lts{\TMM}$.

Although such a universal \RTM{} $\URTM$ exists up to branching
bisimilarity, as we shall see below, it does not exist up to
divergence-preserving branching bisimilarity. To see this, note that
the transition system associated with any particular \RTM{} \URTM{}
has a branching degree that is bounded by some natural number $B$.
It can then be established that, up to divergence-preserving branching
bisimilarity, that \URTM{} can only simulate \RTM{}s with a branching
degree bounded by $B$.
\fctonly{The argument is formalised in the following proposition;
  see~\cite{BLT11full} for a proof.}

\begin{proposition}
 There does not exist an \RTM{} $\URTM{}$ such that for all \RTM{}s
 \TMM{} it holds that
    $\aep{\{\chanrtm\}}{\URTM \parc \TMsend{\TMM}}\bbisimed\lts{\TMM}$.
\end{proposition}
\fullonly{%
\begin{proof}
 Suppose that $\URTM$ is an \RTM{} such that 
    $\aep{\{\chanrtm\}}{\URTM \parc
      \TMsend{\TMM}}\bbisimed\lts{\TMM}$
 holds for every \RTM{} \TMM{}.
 Then, by the way $\TMsend{\TMM}$ is defined, the branching degree of
 $\lts{\aep{\{\chanrtm\}}{\URTM \parc
     \TMsend{\TMM}}\bbisimed\lts{\TMM}}$
 is bounded by the branching degree bound on $\lts{\URTM}$, say
 $\bdegbound$.
 Now, consider the \RTM{}
   $\TMM=(\StatesS[\TMM],\trel[\TMM],\init[\TMM],\final[\TMM])$
  with
  \begin{align*}
    \StatesS[\TMM] & = \{\init[\TMM],0,\dots,\bdegbound+1\}
  \enskip,\\
    {\trel[\TMM]} & = \{(\init[\TMM],\tpblank,\acta,\tpblank,\tpmvR,i)\mid
    i=0,\dots,B+1\}\enskip,\ \text{and}\\
   {\final[\TMM]} & = \{0\}\enskip.
  \end{align*}
  Clearly, the configuration $(\init[\TMM],\tphd{\tpblank})$ in
  $\lts{\TMM}$ has branching degree $B+1$.
  Since $\aep{\{\chanrtm\}}{\URTM \parc
      \TMsend{\TMM}}\bbisimed\lts{\TMM}$,
  there exists a state $\states$ in
     $\lts{\aep{\{\chanrtm\}}{\URTM \parc
        \TMsend{\TMM}}}$
  that is related by a  divergence-preserving branching bisimulation
  to $(\init[\TMM],\tphd{\tpblank})$.
  Moreover, since $(\init[\TMM],\tphd{\tpblank})$ has no outgoing
  $\silent$-transitions, it follows from the definition of
  divergence-preserving branching bisimulation that from $\states$ we
  can execute at most finitely many $\silent$-transitions to state
  $\states'$ without outgoing $\silent$-transitions that must also
  related to $(\init[\TMM],\tphd{\tpblank})$ by the same
  divergence-preserving branching bisimulation. But then $\states'$
  must simulate each of the $B+1$ outgoing $\acta$-transitions to
  states that are pairwise not divergence-preserviung branching
  bisimilar, and therefore it has a branching degree of $\bdegbound + 1$.
 This is a contradiction, and we conclude that the supposed \RTM{}
 $\URTM$ cannot exist.
\end{proof}%
}

If we insist on simulation up to divergence-preserving branching
bisimilarity, then we need to relax the notion of universality.
\begin{definition}
  Let $\bdegbound$ be a natural number.
  An \RTM~$\UbRTM$ is \emph{universal up to $\bdegbound$} if for every
  \RTM~$\TMM$ of which the associated transition system $\lts{\TMM}$
  has a branching degree bounded by $\bdegbound$ it holds that
  $\lts{\TMM} \bbisimed \aep{\{\chanrtm\}}{\UbRTM\parc \TMsend{\TMM}}$.
\end{definition}

\fullonly{%
We now present the construction of a collection of \RTM{}s $\UbRTM$ for all
branching degree bounds $\bdegbound$. We now benefit from our generic
approach in Sect.~\ref{sec:expr:bbctss}: to obtain a definition of
$\UbRTM$, it is enough to adapt the initialisation fragment of the
simulator \BBCTSS{}.
}%
\fctonly{%
The construction of the simulator for a transition system of which the
branching degree is bounded by $\bdegbound$ in the proof of
Theorem~\ref{thm:BBCTSS} can be adapted to
get the definition of an \RTM{} $\UbRTM$ that is universal up to
$\bdegbound$. It suffices to slightly modify the initialisation
fragment. Instead of writing the codes of the functions $\outsetsym$
and $\finaltest{\_}$ and the initial state directly on the tape, the
\emph{initialisation fragment}\fullonly{$\InUF$} receives the code
$\TMcode{\TMM}$ of an arbitrary $\TMM$ along some dedicated
channel~$\chanrtm$. Then, it recursively computes the codes of the 
functions $\outsetsym$ and $\finaltestsym$, and the initial state
of $\lts{\TMM}$ and stores these on the tape.}
\fullonly{\paragraph{Initialisation fragment}%
Recall that the initialisation fragment \InF{} of the simulator
\BBCTSS{} is designed to write the codes $\code{\outsetsym}$,
$\code{\finaltestsym}$ and $\statecode{\init[{\lts{\TMM{}}}]}$ for a
fixed $\LTST$ on the tape. The initialisation fragment $\InUF$ of
$\UbRTM{}$ should, instead, input the G\"odel number $\code{\TMM}$ of
an arbitrary $\TMM{}$ along channel $\chanrtm{}$ and \emph{compute}
the codes $\code{\outsetsym}$, $\code{\finaltestsym}$ and
$\statecode{\init[{\lts{\TMM{}}}]}$ of the associated transition
system $\lts{\TMM{}}$. We do not elaborate on the details of
computing these codes from $\code{\TMM{}}$; their existence follows from
standard recursion-theoretic arguments.
Here, it suffices to declare a deterministic Turing machine
  $\TRANSF=(\StatesS[\TRANSF],\trel[\TRANSF],\init[\TRANSF],\final[\TRANSF])$
that, when started on a tape instance of the form
\begin{equation*}
  \tpleft\code{\TMM{}}\tphd{\tpright}
\enskip,
\end{equation*}
performs a terminating computation that yields the tape instance
\begin{equation*}
  \tpleft\code{\outsetsym}\tpsepc\code{\finaltestsym}\tphd{\tpright}
    \statecode{\init[{\lts{\TMM{}}}]}
\enskip,
\end{equation*}
where $\code{\outsetsym}$ and $\code{\finaltestsym}$ are indices of the
partial recursive functions $\outsetsym$ and $\finaltestsym$
associated with $\lts{\TMM{}}$, and $\init[{\lts{\TMM{}}}]$ is the
initial configuration of $\lts{\TMM{}}$.

Let us assume that $\TMsend{\TMM{}}$ outputs the G\"odel number
$\code{\TMM{}}$ of $\TMM{}$ along channel $\chanrtm{}$ as a sequence
of $\code{\TMM{}}+1$ $\tpunit$s delimited by $\tpleft$ and $\tpright$. Then the
initialisation fragement $\InUF$ should first receive along channel
$\chanrtm{}$ the symbol $\tpleft$, then a sequence of $\tpunit$s,
until it receives the symbol $\tpright$, and then continue with the
computation defined by the deterministic Turing machine $\TRANSF$.

The set of states of $\InUF$ is defined as
\begin{equation*}
  \InUFStates=
    \{\InUFstate[0],\InUFstate[1],\InUFstate[2]\}
      \cup
    (\StatesS[\TRANSF]\setminus \final[\TRANSF])
\enskip,
\end{equation*}
its initial state is defined as
\begin{equation*}
  \InUFinit=\InUFstate[0]
\enskip;\ \text{and}
\end{equation*}
its set of transitions is defined as
\begin{align*}
  \InUFtrel=\mbox{}
    &
    \{(\InUFstate[0],\tpblank,\recv{\chanrtm{}}{\tpleft},\tpleft,\tpmvR,\InUFstate[1]),\
       (\InUFstate[1],\tpblank,\recv{\chanrtm{}}{\tpunit},\tpunit,\tpmvR,\InUFstate[1]),\\
    & \qquad
       (\InUFstate[1],\tpblank,\recv{\chanrtm{}}{\tpright},\tpright,\tpmvR,\InUFstate[2]),\
       (\InUFstate[2],\tpblank,\silent,\tpblank,\tpleft,\tpmvL,\init[\TRANSF])
    \}\\
    &\ \mbox{} \cup
           \{(\InUFstate,\datumd,\silent,\datume,\tpmv{M},\init[\TRANSF])
             \mid(\InUFstate,\datumd,\datume,\tpmv{M},\InUFstate')\in\trel[\TRANSF],\
                 \InUFstate'\in\StatesS[\TRANSF]\setminus\final[\TRANSF]\}\\
    &\ \mbox{} \cup
           \{(\InUFstate,\datumd,\silent,\datume,\tpmv{M},\StFinit)
              \mid(\InUFstate,\datumd,\datume,\tpmv{M},\InUFstate')\in\trel[\TRANSF],\
                  \InUFstate'\in\final[\TRANSF]\}
\end{align*}



Note that $\InUF$ gives rise to an deterministic internal computation
only in parallel composition with an \RTM{} $\TMsend{\TMM{}}$ that
sends the G\"odel number of some \RTM{} $\TMM{}$.

\begin{fact}\label{fact:confs:initU}
 Let $\TMM{}$ be an arbitrary \RTM{}, let $\initc[\TMsend{\TMM{}}]$ denote the
 initial configuration of $\TMsend{\TMM{}}$ and let $\finalc[\TMsend{\TMM{}}]$
 denote the final configuration of $\TMsend{\TMM{}}$.
 Then the parallel composition
   $\aep{\{\chanrtm\}}{\TMsend{\TMM} \parc \UbRTM}$
 of $\TMsend{\TMM{}}$ with the fragment $\InUF$ gives rise to a
 deterministic internal computation from
\begin{equation*}
  (\initc[\TMsend{\TMM{}}], (\InFinit, \tphd{\tpblank}))
\end{equation*}
 to
\begin{equation*}
  (\finalc[\TMsend{\TMM{}}], (\StFinit, \tpleft\code{\outsetsym}\tpsepc\code{\finaltestsym}
               \tphd{\tpright}
               \statecode{\init[\LTST]}))
\enskip;
\end{equation*}
  we denote its set of intermediate states by $\IS{\InUF}$.
\end{fact}

\paragraph{A universal \RTM{} for branching degree bound $\bdegbound$}

For a fixed branching degree bound $\bdegbound$, we define the \RTM{}
  $\UbRTM=(\UbRTMStates,\UbRTMtrel,\UbRTMinit,\UbRTMfinal)$
as the union of the fragments $\InUF$, $\StF$ and $\SpF$ defined
above: the set of states of each particular $\UbRTM$ is defined as the
union of the sets of states of the fragments:
\begin{equation*}
  \UbRTMStates=
    \InUFStates
  \cup
    \StFStates
  \cup
    \SpFStates
\enskip;
\end{equation*}
the transition relation of $\UbRTM$ is the union of the
transition relations of all fragments:
\begin{equation*}
  \UbRTMStates=
    \InUFtrel
  \cup
    \StFtrel
  \cup
    \SpFtrel
\enskip;
\end{equation*}
the initial state of $\UbRTM$ is the initial state of $\InUF$:
\begin{equation*}
  \UbRTMinit=\InUFinit
\enskip;\ \text{and}
\end{equation*}
the set of final states of $\UbRTM$ is
\begin{equation*}
  \UbRTMfinal=\{\SpFstate[{1,\acta[1],\dots,\acta[k]}]
                \mid 0\leq k\leq B\ \&\ \acta[1],\dots,\acta[k]\in\Actt\}
\enskip.
\end{equation*}
The following theorem establishes that $\UbRTM{}$ is universal up to $\bdegbound$.
}

\begin{theorem}\label{thm:UbRTM}
 For all \RTM{}s $\TMM$ with a branching degree bounded by
 $\bdegbound$, it holds that
  $\lts{\TMM} \bbisimed 
   \aep{\{\chanrtm\}}{\TMsend{\TMM} \parc \UbRTM}$.
\end{theorem}
\fullonly{%
\begin{proof}
 Referring to Fact~\ref{fact:confs:initU} we define the following
 relation on configurations of the parallel composition
   $\aep{\{\chanrtm\}}{\TMsend{\TMM} \parc \UbRTM}$:
 \begin{equation*}
   \brelsym[\init] = \{(\init[{\lts{\TMM}}],(\statet[1],\statet[2])) \mid \statet \in \IS{\InUF} \}
 \enskip,
 \end{equation*}
 and referring to Facts~\ref{fact:confs:state} and
 \ref{lem:confs:step}, we define, for every
 $\states\in\StatesS[\LTST]$, the relation
 \begin{equation*}
   \brelsym[\states] = 
      \{(\states, (\finalc[{\TMsend{\TMM{}}}],\statet)) \mid \statet \in \IS{\StF,\states}\cup\IS{\SpF,1,\states}\cup\IS{\SpF,2,\states}\}
 \enskip.
 \end{equation*}
 Then it is straightforward to verify that the binary relation
 \begin{equation*}
    \brelsym = \brelsym[\init] \cup 
               \textstyle\bigcup_{\states\in\StatesS[\LTST]} \brelsym[\states]
  \end{equation*}
  is a divergence-preserving branching bisimulation from $\lts{\TMM{}}$ to
  $\lts{\UbRTM}$.
  Since
    $(\initc[{\TMsend{\TMM{}}}],(\init[\UbRTM{}],\tphd{\tpblank}))\in\IS{\InUF}$,
  it follows that $(\init[{\lts{\TMM}}],(\initc[{\TMsend{\TMM{}}}],(\init[\UbRTM{}],\tphd{\tpblank})) )\in\brelsym$,
  and hence $\lts{\TMM}\bbisimed \aep{\{\chanrtm\}}{\TMsend{\TMM} \parc \UbRTM}$.
\end{proof}%
}

At the expense of introducing divergence it is possible to define a
universal \RTM{}. Recall that, by Proposition~\ref{prop:Phillips},
every effective transition system is branching bisimilar to a
boundedly branching transition system. The proof of this result
exploits a trick, first described in~\cite{BBK87} and adapted by
Phillips in \cite{Phi93}, to use a divergence with (infinitely many)
states of at most a branching degree of $2$ to simulate, up to
branching bisimilarity, a state with arbitrary (even countably
infinite) branching degree.
\fullonly{%
The auxiliary Turing machine $\TRANSF$, used in the fragment $\InUF$
to compute the codes of $\outsetsym$, $\finaltestsym$ and
$\init[\lts{\TMM{}}]$ for $\lts{\TMM{}}$ can be adapted to deliver,
instead, the codes of functions $\outsetsym'$, $\finaltestsym'$ and
$\init[\LTST]$ of a computable transition system $\LTST$ with
branching degree bounded by $2$ such that
$\LTST\bbisim\lts{\TMM{}}$. Thus, we get the following corollary to Theorem~\ref{thm:UbRTM}.}

\begin{corollary}\label{cor:U2RTM}
 There exists an \RTM{} $\URTM$ such that
    $\lts{\TMM} \bbisim 
       \aep{\{\chanrtm\}}{\TMsend{\TMM} \parc \URTM}$
   for every \RTM{} $\TMM$.
\end{corollary}


\section{A process calculus}\label{sec:expl-int}

\fullonly{%
  We have presented reactive Turing machines and studied the ensued
  notion of executable behaviour modulo (divergence-preserving)
  branching bisimilarity. In process theory, behaviour is usually
  specified in some process calculus. In this section, we 
  present a simple process calculus with only standard
  process-theoretic notions and establish that every executable
  behaviour can be defined with a finite specification in our calculus
  up to divergence-preserving branching bisimilarity.  }
\fullonly{%
The process calculus we define below is closest to value-passing
{CCS} \cite{Mil89} for a finite set of data. It deviates from
value-passing {CCS} in that it combines parallel
composition and restriction in one construct (i.e., it includes the
special form of parallel composition already presented in
Sect.~\ref{sec:rtm}), omits the relabelling construction, and
distinguishes successful and unsuccessful termination. Our process
calculus may also be viewed as a special instance of the fragment of
$\TCPt$, excluding sequential composition (see~\cite{BBR09}).}

\fctonly{%
  A well-known result from the theory of automata and formal languages
  is that the formal languages accepted by Turing machines correspond
  with the languages generated by an unrestricted grammar. A
  \emph{grammar} is a formal system for describing a formal language. The
  corresponding notion in concurrency theory is the notion of
  \emph{recursive specification}, which is a formal system for
  describing behaviour. In this section, we show that the behaviours
  of \RTM{}s correspond with the behaviours described by so-called
  $\TCPt$ recursive specifications. The process theory $\TCPt$ is a
  general theory for describing behaviour, encompassing the key
  features of the well-known process theories \ACPt{}~\cite{BK85},
  \CCS{}~\cite{Mil89} and \CSP{}~\cite{Hoa85}.
}

\fctonly{%
  We briefly introduce the syntax of \TCPt{} and informally
  describe its operational semantics. We refer to the
  textbook~\cite{BBR09} for an elaborate treatment.}
Recall the finite sets $\Chan$ of channels and $\Datab$ of data
on which the notion of parallel composition defined in
Sect.~\ref{sec:rtm} is based.
For every subset $\Chan'$ of $\Chan$ we define a special set of
actions $\IAct[{\Chan'}]$ by:
\fullonly{\begin{equation*}}\fctonly{$}
  \IAct[{\Chan'}] = \{\recv{\chan}{\datum}, \send{\chan}{\datum}
                 \mid \datum \in \Datab, \chan \in \Chan' \}
\fctonly{$}
\fullonly{\enskip}.
\fullonly{\end{equation*}}
The actions $\recv{\chan}{\datum}$ and $\send{\chan}{\datum}$ denote
the events that a datum $\datum$ is received or sent along
\emph{channel}~$\chan$. Furthermore, let $\Name$ be a countably
infinite set of names.  The set of \emph{process expressions} $\PExp$
is generated by the following grammar ($\acta
\in \Actt \cup \IAct{}, \name \in \Name, \Chan'\subseteq \Chan$):
\begin{equation*}
   \pexp~::=~%
     \dl\ \mid\ \emp\ \mid\ 
     \pref{\acta}{\pexp}\ \mid\
     \pexp\altc\pexp\ \mid\
     \aep{\Chan'}{\pexp \parc \pexp}\ \mid\
     \name
\enskip.
\end{equation*}
\fullonly{Let us briefly comment on the operators in this syntax.}
The constant $\dl$ denotes \emph{deadlock}, the unsuccessfully
terminated process.
The constant $\emp$ denotes \emph{skip}, the successfully terminated process.
For each action $\act \in \Actt \cup \IAct$ there is a unary operator
$\pref{\act}{}$ denoting action prefix; the process denoted by
$\pref{\act}{\pexpp}$ can do an $\act$-transition to the process
denoted by $\pexpp$. 
The binary operator $\altc$ denotes \emph{alternative composition} or 
\emph{choice}.
The binary operator $\aep{\Chan{}'}{\_ \parc \_}$ denotes the special kind
of \emph{parallel composition} that we have also defined on \RTM{}s.  It
enforces communication along the channels in $\Chan{}'$, and communication
results in $\silent$.
\fctonly{(By including the restricted kind of parallel
composition, we deviate from the definition of \TCPt{} discussed
in~\cite{BBR09}, but we note that our notion of parallel composition is
definable with the operations $\parc$, $\encap{\_}{\_}$ and
$\abstr{\_}{\_}$ of \TCPt{} in~\cite{BBR09}.)}%

\fctonly{%
A \emph{recursive specification} $\RSpec$ is a set of
equations of the form: $\name \defeq \pexp$, with as left-hand side a
name $\name$ and as right-hand side a \TCPt{} process expression
$\pexp$.  It is required that a recursive specification $\RSpec$
contains, for every $\name \in \Name$, at most one equation with
$\name$ as left-hand side; this equation is referred to as the
\emph{defining equation} for $\name$ in $\Name$.  Furthermore, if some
name occurs in the right-hand side of some defining equation, then the
recursive specification must include a defining equation for it.
Let $\RSpec$ be a recursive specification, and let $\pexpp$ be a
process expression. There is a standard method to associate with
$\pexpp$ a transition system $\lts[\RSpec]{\pexpp}$. The details can
be found, e.g., in~\cite{BBR09}%
}%
\fctsubmissiononly{%
  , and, for the convenience of the reviewer, can also be found in
  Appendix~\ref{app:tcptos}%
}%
\fctonly{.}%
\fullonly{%
A \emph{recursive specification} $\RSpec$ is a set of equations
\begin{equation*}
\RSpec=
  \{\name \defeq \pexp \mid
        \name\in\Name\ \&\ \pexp\in\PExp
  \}
\end{equation*}
satisfying the requirements that
\begin{enumerate}
\renewcommand{\theenumi}{\roman{enumi}}
\renewcommand{\labelenumi}{(\theenumi)}
\item for every $\name \in \Name$ it includes at most one equation
  with $\name$ as left-hand side, which is referred to as the
  \emph{defining equation} for $\name$; and
\item if some name $\name$ occurs in the right-hand side $\pexp'$ of
  some equation $\name'=\pexp'$ in $\RSpec$, then $\RSpec$ must
  include a defining equation for $\name$.
\end{enumerate}
Let $\RSpec$ be a recursive specification and let $\pexp$ be a process
expression. We say that $\pexp$ is \emph{$\RSpec$-interpretable} if
all occurrences of names in $\pexp$ have a defining equation in
$\RSpec$.

We use Structural Operational Semantics~\cite{Plo04a} to associate a
transition relation with process expressions:
let $\trel$ be the $(\Actt\cup\IAct)$-labelled transition relation
induced on the set of process expressions by operational rules in
Table~\ref{tbl:sos-tcpt}. Note that the operational rules presuppose a
recursive specification $\RSpec$. 


\begin{framedtable}[htb]
 \begin{center}
 \begin{osrules}
   \osrule*{}{\term{\emp}} \qquad \qquad \qquad \qquad
     \osrule*{}{\pref{\act}{\pexp} \step{\act} \pexp} \\[2em]
   \osrule*{\pexpp \step{\act} \pexpp'}%
           {\pexpp \altc \pexpq \step{\act} \pexpp'} \qquad
     \osrule*{\pexpq \step{\act} \pexpq'}%
             {\pexpp \altc \pexpq \step{\act} \pexpq'} \qquad
     \osrule*{\term{\pexpp}}{\term{(\pexpp \altc \pexpq)}} \qquad
     \osrule*{\term{\pexpq}}{\term{(\pexpp \altc \pexpq)}} \\[2em]
   \osrule*{\pexpp \step{\act} \pexpp' & \act\not\in\IAct[\Chan']}%
           {\aep{\Chan'}{\pexpp \parc \pexpq} \step{\act} \aep{\Chan'}{\pexpp' \parc \pexpq}} \qquad
     \osrule*{\pexpq \step{\act} \pexpq' & \act\not\in\IAct[\Chan']}%
             {\aep{\Chan'}{\pexpp \parc \pexpq} \step{\act} \aep{\Chan'}{\pexpp \parc \pexpq'}} \qquad
     \osrule*{\term{\pexpp} & \term{\pexpq}}%
             {\term{\aep{\Chan'}{\pexpp \parc \pexpq}}} \\[2em]
     \osrule*{\pexpp \step{\recv{\chan}{\datum}} \pexpp' &
              \pexpq \step{\send{\chan}{\datum}} \pexpq'}%
             {\aep{\Chan'}{\pexpp \parc \pexpq}\step{\silent}
              \aep{\Chan'}{\pexpp' \parc \pexpq'}} 
\qquad
     \osrule*{\pexpp \step{\send{\chan}{\datum}} \pexpp' &
              \pexpq \step{\recv{\chan}{\datum}} \pexpq'}%
             {\aep{\Chan'}{\pexpp \parc \pexpq}\step{\silent}
              \aep{\Chan'}{\pexpp' \parc \pexpq'}} \\[2em]
   \osrule*{\pexpp \step{\act} \pexpp' &
            (\name \defeq \pexpp) \in \RSpec}%
           {\name \step{\act} \pexpp'} \qquad 
     \osrule*{\term{\pexpp} &
              (\name \defeq \pexpp) \in \RSpec}%
             {\term{\name}}
  \end{osrules}
 \caption{Operational rules for a recursive specification 
   $\RSpec$, with $\pexpp,\pexpp',\pexpq,\pexpq'\in\PExp$, $\acta
\in \Actt \cup \IAct{}$, $\name \in \Name$, $\chan\in\Chan$,
$\datum\in\Data$, and $\Chan'\subseteq \Chan$.}\label{tbl:sos-tcpt}
 \end{center}
\end{framedtable}

\begin{definition}\label{def:lts-rspec}
  Let $\RSpec$ be a recursive specification and let $\pexpp$ be an
  $\RSpec$-interpretable process expression.
  We define the labelled transition system
  \begin{equation*}
    \lts[\RSpec]{\pexpp}=
     (\StatesS[\pexpp],\trel[\pexpp],\init[\pexpp],\final[\pexpp])
  \end{equation*}
  associated with $\pexpp$ and $\RSpec$ as follows:
  \begin{enumerate}
  \item the set of states $\StatesS[\pexpp]$ consists of all process
    expressions reachable from $\pexpp$;
  \item the transition relation $\trel[\pexpp]$ is the restriction
    to $\StatesS[\pexpp]$ of the transition relation $\trel$ defined
    on all process expressions by the operational rules in
    Table~\ref{tbl:sos-tcpt}, i.e.,
      $\trel[\pexpp]=\trel\cap
        (\StatesS[\pexpp]\times (\Actt\cup\IAct)\times\StatesS[\pexpp])$.
  \item the process expression $\pexpp$ is the initial state,
    i.e. $\init[\pexpp]=\pexpp$; and
  \item the set of final states consists of all process expressions
    $\pexpq\in\StatesS[\pexpp]$ such that $\term{\pexpq}$, i.e.,
      $\final[\pexpp]=\final\cap\StatesS[\pexpp]$.
  \end{enumerate}
\end{definition}

It is straightforward to associate with every \RTM{}
  $\TM{}=(\States_{\TM{}},\trel_{\TM{}},\init_{\TM{}},\final_{\TM{}})$ 
a recursive specification $\RSpecTMinf{}$ and a process expression $\pexp$
such that   $\lts{\TM{}}$
is divergence-preserving branching bisimilar to
  $\lts[\RSpecTMinf{}]{\pexp}$:
\begin{enumerate}
\item Associate a distinct name $\name[M]_c$ with every configuration $c$
  of $\TM{}$; and
\item let $\RSpecTM{}$ consist of all equations
\begin{equation*}
  \name[M]_c\defeq{}\textstyle\sum_{(\acta,c')\in\outset{c}}\pref{\acta}{\name[M]_{c'}}\
    [\mbox{}\altc\emp{}]_{\final{c}}
    \qquad\text{($c$ a configuration of $\TM{}$)}
\enskip.
\end{equation*}
  (We use summation $\sum$ to abbreviate an $\outset{c}$-indexed
  alternative composition, and indicate by
  $[\mbox{}\altc\emp{}]_{\final{c}}$ that the summand $\emp{}$ is only
  included if $\final{c}$ holds.)
\end{enumerate}
It can be easily verified that
  $\lts{\TM{}}\bbisimed\lts[\RSpecTM{}]{\name_{({\init[\TM{}]},\tphd{\tpblank})}}$.

Our main goal in this section is to show that one does not have
to resort to infinite recursive specifications: we establish
that for every \RTM{} $\TM{}$ there exists a \emph{finite} recursive
specification $\RSpecTM{}$ and an $\RSpecTM{}$-interpretable process
expression $\pexp$ such that
  $\lts{\TM{}}$
is divergence-preserving branching bisimilar to
  $\lts[\RSpecTM{}]{\pexp}$.
Our specification consists of two parts: a generic finite
specification of the behaviour of a tape, and a finite specification
of a control process that is specific for the \RTM{} $\TM{}$ under
consideration. In the end we establish that an \RTM{} $\TM{}$ is
finitely specified by the parallel composition of its associated
control with a process modelling the tape. 

It is convenient to know that divergence-preserving branching
bisimilarity is compatible with the notion of parallel composition in
our calculus, allowing us to establish the correctness of both
components separately.
\begin{lemma} \label{lem:congruence}
  Let $\RSpec$ be a recursive specification and let $\pexpp[1]$,
  $\pexpp[2]$, $\pexpq[1]$ and $\pexpq[2]$ be $\RSpec$-interpretable
  process expressions.
  If $\lts[\RSpec]{\pexpp[1]}\bbisimed\lts[\RSpec]{\pexpp[2]}$ and
  $\lts[\RSpec]{\pexpq[1]}\bbisimed\lts[\RSpec]{\pexpq[2]}$, then
$\lts[\RSpec]{\aep{\Chan'}{\pexpp[1]\parc\pexpq[1]}}
     \bbisimed
   \lts[\RSpec]{\aep{\Chan'}{\pexpq[1]\parc\pexpq[2]}}$.
\end{lemma}
}

\fullonly{%
\subsection{Tape}  \label{sec:finitetapespec}

We want to present a finite specification of the behaviour of the tape
of an \RTM{}, but before we do so, we give a straightforward infinite
specification. As an intermediate correctness result, we then
establish that the behaviour defined by our finite specification is
divergence-preserving branching bisimilar to the behaviour induced by
our infinite specification. Recall that our definition of tape
instance (see p.~\pageref{tapeinstance}) uses a tape head marker to
indicate the position of the tape head; the state of the tape is,
therefore, uniquely represented by a tape instance. The behaviour of a
tape in the state represented by the tape instance
$\tpstrdL\tphd{\datumd}\tpstrdR$ is characterised by the following
equation:
\begin{multline} \label{eq:TapeSpec}
 \name[T]_{\tpstrdL\tphd{\datumd}\tpstrdR}  
   \defeq\send{\chanr}{\datumd}.\name[T]_{\tpstrdL\tphd{\datumd}\tpstrdR}
          \altc\sum_{\datume \in \Datab} 
          \recv{\chanw}{\datume}.\name[T]_{\tpstrdL\tphd{\datume}\tpstrdR} \\
          \altc\recv{\chanm}{\tpmv{L}}.\name[T]_{\tphdL{\tpstrdL}\datumd\tpstrdR}
          \altc\recv{\chanm}{\tpmv{R}}.\name[T]_{\tpstrdL\datumd\tphdR{\tpstrdR}}
          \altc\emp{}\;.
\end{multline}
The equation expresses that a tape, when it is in the state
represented by tape instance $\tpstrdL\tphd{\datumd}\tpstrdR$, can
either output the datum $\datumd$ under the head along its
\emph{read-channel} $\chanr$, input a new datum $\datume$ along its
\emph{write-channel} $\chanw$ which then replaces the datum $\datumd$
under the head, or receive over its \emph{move-channel} $\chanm$ the
instruction to move the head either one position to the left
($\tpstrdL$) or one position to the right ($\tpstrdR$). It is for
notational convenience and not essential that we here separate the
operations of reading, writing and moving, instead of combining them
in a single communication. (Note that separating the operations is
harmless, since our specification of the finite control specified
below (see Section~\ref{subsec:fincontrol}) will ensure that always an
appropriate combination of tape operations is carried out one after
the other, and the eventual parallel composition of the tape and the
finite control eventually abstracts from the actual calling of tape
operations.)
The additional summand $\emp$ of the
right-hand side of  the equation indicates that the state of the tape
represented by $\tpstrdL\tphd{\datumd}\tpstrdR$ is final; this is
needed to ensure that the parallel composition of a tape with a finite
control for an \RTM{} is final whenever the finite control is in a
final state.

We denote by $\RSpecTinf$ the recursive specification consisting of
all instantiations of Eqn.~\eqref{eq:TapeSpec} with concrete values
for $\datumd$, $\datume$, $\tpstrdL$ and $\tpstrdR$. It is easy to see
that the $\RSpecTinf$-interpretable process expression
$\name[T]_{\tphd{\tpblank}}$ completely specifies all possible
behaviour of the tape when it is started with an empty tape
instance. It is also clear that there are infinitely many combinations
of concrete values for $\datumd$, $\datume$, $\tpstrdL$ and
$\tpstrdR$, so $\RSpecTinf$ is infinite.

For our finite specification of the behaviour of the tape, we make use
of a seminal result in the process theoretic literature, to the effect that
the behaviour of a queue can be specified in the process calculus at
hand with finitely many equations. Note that the state of a queue is
uniquely represented by its contents, a string $\tpstrd$; we denote
the behaviour of a queue with contents $\tpstrd$ by
$\name[Q]_{\tpstrd}$. Then the behaviour of a queue in all its
possible states is specified by the following infinite recursive
specification $\RSpecQinf$ (with $\datumd \in \Datab$ and $\tpstrd \in
\Datab^{*}$, and $\emptystr$ denoting the empty string):
\begin{spec}
 \name[Q]_{\emptystr} & 
   \defeq \textstyle\sum_{\datumd \in \Data}
          \recv{\chani}{\datumd}.\name[Q]_{\datumd}
          \altc\emp
 \;, \\
 \name[Q]_{\tpstrd\datumd} & 
   \defeq \send{\chano}{\datumd}.\name[Q]_{\tpstrd}
          \altc\textstyle\sum_{\datume \in \Data} 
              \recv{\chani}{\datume}.\name[Q]_{\datume\tpstrd\datumd}
          \altc\emp \;.
\end{spec}
The equation for $\name[Q]_{\emptystr}$ expresses that the empty queue
can only receive an arbitrary datum $\datumd$ along its
\emph{input-channel} $\chani$; the equation for
$\name[Q]_{\tpstrd\datumd}$ expresses that a queue containing at least
one element $\datumd$ at the front of the queue may also output
$\datumd$ along its \emph{output-channel} $\chano$.

Bergstra and Klop \cite{BK86} discovered the following intricate
finite specification $\RSpecQ$, consisting of six equations,
completely defining the behaviour of a queue. Its correctness has been
formally established by Bezem and Ponse \cite{BP97}:
\begin{spec}
 {\name[Q]}^{\chanj\chank}_{\chanp} &
  \defeq \sum_{\datumd \in \Datab}
                  \recv{\chanj}{\datumd}.%
                  \aep{\{\chanp\}}{{\name[Q]}^{\chanj\chanp}_{\chank} \parc 
                           (\emp\altc
                           \send{\chank}{\datumd}.{\name[Q]}^{\chanp\chank}_{\chanj})}
              \altc\emp
 \quad\text{for all $\{\chanj,\chank,\chanp\} = \{\chani,\chano,\chanl\}$.}
\end{spec}
The $\emp$-summands are not part of Bergstra and Klop's
specification. We have added them to make sure that every state of our
queue is final. It is easy to see that they have no further influence
on the behaviour of the specified process. Also, it can be easily
verified that the proof in \cite{BP97} is still valid, and that, in
fact, the proof establishes that the finite and infinite
specifications of the empty queue are divergence-preserving branching bisimilar.
\begin{theorem}\label{thm:expl-int:queue-fin}
 $\lts[\RSpecQinf]{\name[Q]_{\emptystr}}
     \bbisimed
   \lts[\RSpecQ]{\name[Q]^{\chani \chano}_{\chanl}}$.
\end{theorem}

\begin{figure}[htb]
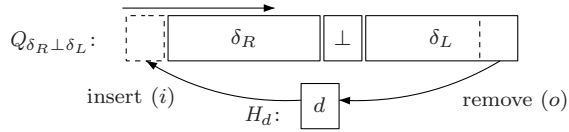

 {\footnotesize
  \begin{center}
  \begin{transsys}(51,15)(0,4)
   \graphset{Nadjust=n,Nmr=0,Nw=20,Nh=6}
   \node(right)(16,12){$\tpstrdR$}
   \node(left)(42,12){$\tpstrdL$}
   \graphset{Nw=5}
   \node(end)(29,12){$\qend$}
   \node(Hd)(26,3){$\datumd$}
   \put(16,1){$\name[H]_{\datumd}$:}
   \node[dash={.7 0}{}](new)(3,12){}
   \drawline[dash={.7 0}{},AHnb=0](47,15)(47,9)

   \graphset{curvedepth=3,ELside=l}
   \edge[ELpos=80,exo=-4,syo=1](Hd,new){insert ($\chani$)}
   \edge[ELpos=20,sxo=12,eyo=1](left,Hd){remove ($\chano$)}

   \put(-15,11){$\name[Q]_{\tpstrdR\qend\tpstrdL}$:}
   \drawline(0,16)(20,16)
  \end{transsys}
  \end{center}
 }
 \caption{Schematic representation of the tape process.}\label{fig:ov-tape}
\end{figure}

We proceed to explain how we use the queue to finitely specify
a tape. Our specification consists of a \emph{tape controller} that
implements the interface of the tape with its environment and uses the
queue to store information regarding the current tape instance. See
Fig.\ref{fig:ov-tape} for an illustration of how we use the
queue to store part of the tape instance
$\tpstrdL\tphd{\datumd}\tpstrdR$: $\tpstrdL$ constitutes the front
of the queue $\tpstrdL$, $\tpstrdR$ constitutes the tail of the
queue, and we use a special symbol $\qend$ to mark where $\tpstrdR$
ends and $\tpstrdL$ begins. The symbol $\datumd$ under the head is
maintained separately by the tape controller. The four operations on
the tape are implemented by the tape controller as follows:
\begin{enumerate}
\item If the tape controller receives the instruction to output the
  symbol under the head of the tape, then it suffices to send
  $\datumd$  along its read-channel. For this operation no interaction
  with the queue is needed.
\item If the tape controller receives the instruction to overwrite the
  symbol under its head with the symbol $\datume$, then it suffices to
  forget $\datumd$ and continue maintaining $\datume$. For
  this operation also no interaction with the queue is needed.
\item If the tape controller receives the instruction to move the head
  one position to the left, then this amounts to inserting the datum
  $\datumd$ at the tail of the queue and removing the datum at the front
  of the queue, which then becomes the new symbol maintained by
  the tape controller. There is one exeception, though: if the symbol
  at the front of the queue is $\qend$, then the left-most position so
  far has been reached. The new symbol maintained by the tape
  controller should become $\tpblank$, and the queue should be
  restored in its original state. This is achieved by first inserting
  a marker $\qmark$ at the tail of the queue, and then repeatedly
  moving symbols from the front to  the tail of the queue, until the
  symbol $\qmark$ is removed from the front of the queue.
\item If the tape controller receives the instruction to move the head
  one position to the right, then this is implemented by first
  placing the marker $\qmark$ at the tail of the queue, and then
  repeatedly moving symbols from the front of the queue to the tail of
  the queue, until the symbol $\qmark$ is removed from the front of
  the queue. The symbol that was removed before removing $\qmark$
  becomes the new symbol under the head.
\end{enumerate}

The tape controller is defined by the finite recursive specification
$\RSpecT$ consisting of the following equations:
\begin{spec}
 \displaybreak[1]
 \name[H]_{\datumd}   & \defeq \send{\chanr}{\datumd}.\name[H]_{\datumd}
                          \altc\textstyle\sum_{\datume \in \Datab} \recv{\chanw}{\datume}.\name[H]_{\datume}
                          \altc\recv{\chanm}{\tpmv{L}}.\name[H^L_{\mathit\datumd}]
                          \altc\recv{\chanm}{\tpmv{R}}.\name[H^R_{\mathit\datumd}] \altc\emp\;,\\
 \name[H^L_{\mathit\datumd}] & \defeq \send{\chani}{\datumd}.%
                                 \Bigr(\textstyle\sum_{\datume \in \Datab} \recv{\chano}{\datume}.\name[H]_{\datume}
                                       \altc\recv{\chano}{\qend}.\send{\chani}{\qmark}.\send{\chani}{\qend}.\name[Back]\Bigr) \;,\\
 \displaybreak[1]
 \name[Back]          & \defeq \textstyle\sum_{\datumd \in \Datab} \recv{\chano}{\datumd}.\send{\chani}{\datumd}.\name[Back]
                           \altc\recv{\chano}{\qmark}.\name[H]_{\tpblank} \;, \\
 \name[H^R_{\mathit\datumd}] & \defeq \send{\chani}{\qmark}.\send{\chani}{\datumd}.\Bigr(%
                                 \textstyle\sum_{\datume \in \Datab} \recv{\chano}{\datume}.\name[Fwd]_{\datume}
                                 \altc\recv{\chano}{\qend}.\name[Fwd]_{\qend}\Bigr)\;,\\
 \name[Fwd]_{\datumd} & \defeq \textstyle\sum_{\datume \in \Datab} \recv{\chano}{\datume}.%
                            \send{\chani}{\datumd}.\name[Fwd]_{\datume}
                          \altc\recv{\chano}{\qend}.\send{\chani}{\datumd}.\name[Fwd]_{\qend}
                          \altc\recv{\chano}{\qmark}.\name[H]_{\datumd} \;,\\
 \name[Fwd]_{\qend}   & \defeq \textstyle\sum_{\datume \in \Datab} \recv{\chano}{\datume}.%
                            \send{\chani}{\qend}.\name[Fwd]_{\datume}
                          \altc\recv{\chano}{\qmark}.\send{\chani}{\qend}.\name[H]_{\tpblank} \;.
\end{spec}

To establish the correctness of our finite recursive specification of
the tape controller, we first prove below that
\begin{equation*}
   \lts[\RSpecTinf]{\name[T]_{\tpstrdL\tphd{\datumd}\tpstrdR}}
      \bbisimed
    \lts[\RSpecT\cup\RSpecQinf]{\aep{\chanio}{\name[H]_{\datumd} \parc
        \name[Q]_{\tpstrdR\qend\tpstrdL}}}
\enskip.
\end{equation*}

The following two lemmas establish that the operation of moving the
tape head one position to the left, or to the right, is implemented
correctly by the tape controller processes
$\name[H^L_{\mathit\datumd}]$  and $\name[H^R_{\mathit\datumd}]$.

\begin{lemma}\label{lem:expl-int:tp-mvL}
 Let $\datumd \in \Datab$ and let $\tpstrdL,\tpstrdR \in \Datab^{*}$.
 \begin{enumerate}
 \renewcommand{\theenumi}{\roman{enumi}}
 \renewcommand{\labelenumi}{(\theenumi)}
 \item If $\tpstrdL=\tpstrzL\datumd[L]$ for some
   $\tpstrzL\in\Datab^{*}$ and $\datumd[L]\in\Datab$, then
     $\aep{\chanio}{\name[H^L_{\mathit\datumd}] \parc
       \name[Q]_{\tpstrdR\qend\tpstrdL}}$
   has a deterministic internal computation to
     $\aep{\chanio}{\name[H]_{\datumd[L]} \parc 
                         \name[Q]_{\datumd\tpstrdR\qend\tpstrzL}}$;
   we denote its set of intermediate states by
   $\IS{L,\tphd{\datumd[L]}\datumd\tpstrdR}$.
  \item If $\tpstrdL=\emptystr$, then
     $\aep{\chanio}{\name[H^L_{\mathit\datumd}] \parc
       \name[Q]_{\tpstrdR\qend\tpstrdL}}$
   has a deterministic internal computation to
     $\aep{\chanio}{\name[H]_{\tpblank} \parc
       \name[Q]_{\datumd\tpstrdR\qend}}$;
    we denote its set of intermediate states by
      $\IS{L,\tphd{\tpblank}\datumd\tpstrdR}$.
  \end{enumerate}
\end{lemma}
\begin{proof}
  The validity of the lemma is straightforwardly proved by computing
  a fragment of the transition system associated with
    $\aep{\chanio}{\name[H^L_{\mathit\datumd}] \parc
       \name[Q]_{\tpstrdR\qend\tpstrdL}}$.
\end{proof}

\begin{lemma}\label{lem:expl-int:tp-mvR}
 Let $\datumd \in \Datab$ and let $\tpstrdL,\tpstrdR \in \Datab^{*}$.
 \begin{enumerate}
 \renewcommand{\theenumi}{\roman{enumi}}
 \renewcommand{\labelenumi}{(\theenumi)}
 \item If $\tpstrdR = \datumd[R]\tpstrzR$ for some
   $\tpstrzR\in\Datab^{*}$ and $\datumd[R]\in\Datab$, then
     $\aep{\chanio}{\name[H^R_{\mathit\datumd}] \parc
       \name[Q]_{\tpstrdR\qend\tpstrdL}}$
   has a deterministic internal computation to
     $\aep{\chanio}{\name[H]_{\datumd[R]} \parc 
                         \name[Q]_{\tpstrzR\qend\tpstrzL\datumd}}$;
   we denote its set of intermediate states by
     $\IS{R,\tpstrzL\datumd\tphd{\datumd[R]}\tpstrzR}$.
  \item If $\tpstrdR=\emptystr$, then
     $\aep{\chanio}{\name[H^R_{\mathit\datumd}] \parc
       \name[Q]_{\tpstrdR\qend\tpstrdL}}$
   has a deterministic internal computation to
     $\aep{\chanio}{\name[H]_{\tpblank} \parc \name[Q]_{\qend\tpstrdL\datumd}}$;
   we denote its set of intermediate states by
     $\IS{R,\tpstrdL\datumd\tphd{\tpblank}}$.
  \end{enumerate}
\end{lemma}
\begin{proof}
  The validity of the lemma is straightforwardly proved by computing
  a fragment of the transition system associated with
    $\aep{\chanio}{\name[H^R_{\mathit\datumd}] \parc
       \name[Q]_{\tpstrdR\qend\tpstrdL}}$.
\end{proof}

The following theorem establishes the correct behaviour of our tape
controller in a parallel composition with a queue.
\begin{theorem}\label{thm:expl-int:tp-hd-queue}
 Let $\datumd\in\Datab$ and let $\tpstrdL, \tpstrdR\in\Datab^{*}$ .
 Then
\begin{equation*}
   \lts[\RSpecTinf]{\name[T]_{\tpstrdL\tphd{\datumd}\tpstrdR}}
      \bbisimed
    \lts[\RSpecT\cup\RSpecQinf]{\aep{\chanio}{\name[H]_{\datumd} \parc
        \name[Q]_{\tpstrdR\qend\tpstrdL}}}
\enskip.
\end{equation*}
\end{theorem}
\begin{proof}
  Referring to Lemmas~\ref{lem:expl-int:tp-mvL} and
  \ref{lem:expl-int:tp-mvR} for the definitions of $\IS{L,\tpstrd}$
  and $\IS{R,\tpstrd}$, we define the binary relation $\brel$ by
  \begin{multline*}
    \brelsym=
      \{(\name[T_{\tpstrdL\tphd{\datumd}\tpstrdR}],
          \aep{\chanio}{\name[H]_{\datumd} \parc
        \name[Q]_{\tpstrdR\qend\tpstrdL}})\mid \datumd\in\Datab\ \&\
      \tpstrdL,\tpstrdR\in\Datab^{*}\}\\
         \mbox{}\cup 
      \{(\name[T_{\tpstrd}],\states)\mid
           \text{$\tpstrd$ a tape instance and}\
           \states\in\IS{L,\tpstrd}\cup\IS{R,\tpstrd}
      \}
\enskip.
  \end{multline*}
  We leave it to the reader to verify that $\brelsym$ is a
  divergence-preserving branching bisimulation from
  $\lts[\RSpecTinf]{\name[T]_{\tpstrdL\tphd{\datumd}\tpstrdR}}$ to $\lts[\RSpecT\cup\RSpecQinf]{\aep{\chanio}{\name[H]_{\datumd} \parc
        \name[Q]_{\tpstrdR\qend\tpstrdL}}}$.
\end{proof}

Note that, as a direct consequence of
Theorem~\ref{thm:expl-int:queue-fin}, we get that
\begin{equation*}
  \lts[\RSpecQinf]{\name[Q]_{\qend}}
     \bbisimed
   \lts[\RSpecQ]{\aep{\{\chanl\}}{\name[Q]^{\chani
         \chanl}_{\chano}\parc\pref{\send{\chano}{\qend}}{\name[Q]^{\chanl\chano}_{\chani}}}}
\enskip.
\end{equation*}
Hence, by Lemma~\ref{lem:congruence}, we can replace the infinite
recursive specification of the queue in
Theorem~\ref{thm:expl-int:tp-hd-queue} by the finite recursive
specification due to Bergstra and Klop \cite{BK86}, to get the
following corollary.

\begin{corollary} \label{cor:finitetape}
 Let $\datumd\in\Datab$ and let $\tpstrdL, \tpstrdR\in\Datab^{*}$ .
 Then
\begin{equation*}
   \lts[\RSpecTinf]{\name[T]_{\tphd{\tpblank}}}
      \bbisimed
    \lts[\RSpecT\cup\RSpecQ]{\aep{\chanio}{\name[H]_{\tpblank} \parc \aep{\{\chanl\}}{\name[Q]^{\chani
         \chanl}_{\chano}\parc\pref{\send{\chano}{\qend}}{\name[Q]^{\chanl\chano}_{\chani}}}
        }}
\enskip.
\end{equation*}
\end{corollary}

\subsection{Finite control} \label{subsec:fincontrol}

It remains to associate with every \RTM{}
  $\TM{}=(\StatesS[\TM{}],\trel[\TM{}],\init[\TM{}],\final[\TM{}])$
a finite recursive specification $\RSpecFC$ that, intuitively,
implements the finite control of the \RTM{} defined by its transition relation.
For every state $\states \in \States$ and datum $\datumd \in \Datab$,
denote by $\name[C]_{\states,\datumd}$ the process that controls the
behaviour of \TM{} when it is in state $\states$ with $\datumd$ under
the head. We define the behaviour of $\name[C]_{\states,\datumd}$ by
the following equation ($\statet \in \States$, $\acta \in \Actt$,
$\datume \in \Datab$, and $\tpmv{M} \in \{\tpmvL, \tpmvR\}$): 
\begin{spec}
  \name[C]_{\states,\datumd}
    & \defeq \textstyle\sum_{(\states,\datumd,\acta,\datume,\tpmv{M},\statet) \in \trel}
                \left(\acta.\send{\chanw}{\datume}.\send{\chanm}{\tpmv{M}}.%
                  \textstyle\sum_{\datumf \in \Datab} \recv{\chanr}{\datumf}.
                    \name[C]_{\statet,\datumf} \right) [\altc\:\emp]_{\term{\states}} \;;
\end{spec}
we denote by $\RSpecFC$ the set of all instances with a concrete state
$\states\in\States$ and a concrete datum $\datumd\in\Datab$.

The following two lemmas establish that the sequence of instructions
from the finite control to write $\datume$ at the position of the head, move
the tape head one position to the left or right, and then read the
datum at the new position of the head has the desired effect.

\todo{Improve notation for sets of intermediate states (below and
  everywhere else in the paper).}
\begin{lemma}\label{lem:fcinstructionsmvL}
  Let $\datumd,\datume\in\Datab$ and let $\tpstrdL,\tpstrdR\in\Datab^{*}$.
  \begin{enumerate}
  \renewcommand{\theenumi}{\roman{enumi}}
  \renewcommand{\labelenumi}{(\theenumi)}
  \item
    If $\tpstrdL=\tpstrzL\datumd[L]$ for some $\datumd[L]\in\Datab$
    and $\tpstrzL\in\Datab^{*}$,  then 
  \begin{equation*}
    \aep{\chanrwm}{\send{\chanw}{\datume}.\send{\chanm}{\tpmv{L}}.%
     \textstyle\sum_{\datumf \in \Datab} \recv{\chanr}{\datumf}.\name[C]_{\statet,\datumf} \parc 
                               \name[T]_{\tpstrdL\tphd{\datumd}\tpstrdR}}
  \end{equation*}
  has a deterministic internal computation to
  \begin{equation*}
     \aep{\chanrwm}{\name[C]_{\statet,\datumd[L]} \parc 
                                    \name[T]_{\tpstrzL\tphd{\datumd[L]}\datume\tpstrdR}}
  \enskip;
  \end{equation*}
  we denote its set of intermediate states by
    $\IS{[\datumd/\datume]L,\statet,\tpstrzL\tphd{\datumd[L]}\datume\tpstrdR}$.
  \item
    If $\tpstrdL=\emptystr$,
    then 
  \begin{equation*}
    \aep{\chanrwm}{\send{\chanw}{\datume}.\send{\chanm}{\tpmv{L}}.%
     \textstyle\sum_{\datumf \in \Datab} \recv{\chanr}{\datumf}.\name[C]_{\statet,\datumf} \parc 
                               \name[T]_{\tpstrdL\tphd{\datumd}\tpstrdR}}
  \end{equation*}
  has a deterministic internal computation to
  \begin{equation*}
     \aep{\chanrwm}{\name[C]_{\statet,\tpblank} \parc 
                                    \name[T]_{\tphd{\tpblank}\datume\tpstrdR}}
  \enskip;
  \end{equation*}
  we denote its set of intermediate states by
    $\IS{[\datumd/\datume]L,\statet,\tphd{\tpblank}\datume\tpstrdR}$.
  \end{enumerate}
\end{lemma}

\begin{lemma}\label{lem:fcinstructionsmvR}
  Let $\datumd,\datume\in\Datab$ and let $\tpstrdL,\tpstrdR\in\Datab^{*}$.
  \begin{enumerate}
  \renewcommand{\theenumi}{\roman{enumi}}
  \renewcommand{\labelenumi}{(\theenumi)}
  \item
    If $\tpstrdR=\datumd[R]\tpstrzR$ for some $\datumd[R]\in\Datab$
    and $\tpstrzR\in\Datab^{*}$,
    then 
  \begin{equation*}
    \aep{\chanrwm}{\send{\chanw}{\datume}.\send{\chanm}{\tpmv{R}}.%
     \textstyle\sum_{\datumf \in \Datab} \recv{\chanr}{\datumf}.\name[C]_{\statet,\datumf} \parc 
                               \name[T]_{\tpstrdL\tphd{\datumd}\tpstrdR}}
  \end{equation*}
  has a deterministic internal computation to
  \begin{equation*}
     \aep{\chanrwm}{\name[C]_{\statet,\datumd[L]} \parc 
                                    \name[T]_{\tpstrdL\datume\tphd{\datumd[R]}\tpstrzR}}
  \enskip;
  \end{equation*}
  we denote its set of intermediate states by
    $\IS{[\datumd{}/\datume{}]R,\statet,\tpstrdL\datume\tphd{\datumd[R]}\tpstrzR}$.
  \item
    If $\tpstrdR=\emptystr$,
    then 
  \begin{equation*}
    \aep{\chanrwm}{\send{\chanw}{\datume}.\send{\chanm}{\tpmv{L}}.%
     \textstyle\sum_{\datumf \in \Datab} \recv{\chanr}{\datumf}.\name[C]_{\statet,\datumf} \parc 
                               \name[T]_{\tpstrdL\tphd{\datumd}\tpstrdR}}
  \end{equation*}
  has a deterministic internal computation to
  \begin{equation*}
     \aep{\chanrwm}{\name[C]_{\statet,\tpblank} \parc 
                                    \name[T]_{\tpstrdL\datume\tphd{\tpblank}}}
  \enskip;
  \end{equation*}
  we denote its set of intermediate states by
    $\IS{[\datumd{}/\datume{}]R,\statet,\tpstrdL\datume\tphd{\tpblank}}$.
  \end{enumerate}
\end{lemma}

 \begin{figure}[htb]
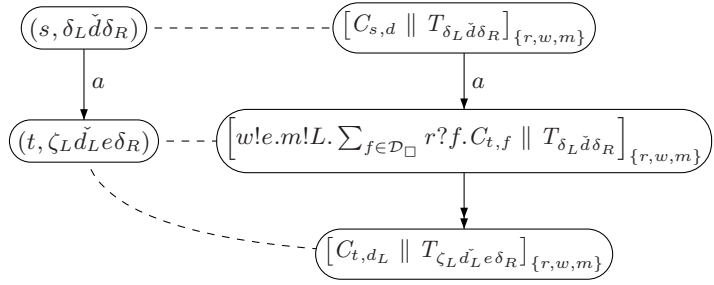

  \begin{center}
  \begin{transsys}(50,33)(10,1)
   {\small
    \node(s)(0,30){$(\states,\tpstrdL\tphd{\datumd}\tpstrdR)$}
    \node(t)(0,15){$(\statet,\tpstrzL\tphd{\datumd[L]}\datume\tpstrdR)$}

    \node(sC)(50,30){$\aep{\chanrwm}{\name[C]_{\states,\datumd} \parc 
                                     \name[T]_{\tpstrdL\tphd{\datumd}\tpstrdR}}$}
    \node(twmrC)(50,15){$\aep{\chanrwm}{\send{\chanw}{\datume}.\send{\chanm}{\tpmv{L}}.%
      \sum_{\datumf \in \Datab} \recv{\chanr}{\datumf}.\name[C]_{\statet,\datumf} \parc 
                                \name[T]_{\tpstrdL\tphd{\datumd}\tpstrdR}}$}
    \node(tC)(50,0){$\aep{\chanrwm}{\name[C]_{\statet,\datumd[L]} \parc 
                                     \name[T]_{\tpstrzL\tphd{\datumd[L]}\datume\tpstrdR}}$}

    \edge(s,t){$\acta$}
    \edge(sC,twmrC){$\acta$}
    \edge[AHnb=2,ELpos=60](twmrC,tC){}

    \graphset{AHnb=0,dash={1 0}{}}
    \edge(s,sC){}
    \edge(t,twmrC){}
    \drawqbedge(t,0,0,tC){}
   }
  \end{transsys}
  \end{center}
  \caption{Relation between an \RTM{} transition and specification
    transitions.}\label{fig:bbisim-exec-tcpt}
 \end{figure}

\begin{theorem} \label{thm:TMspecTinf}
  $\lts{\TM}\bbisimed
      \lts[\RSpecFC\cup\RSpecTinf]{\aep{\chanrwm}{\name[C]_{\init[\TM{}],\tpblank} \parc
   \name[T]_{\tphd{\tpblank}}}}$.
\end{theorem}
\begin{proof}
  Referring to Lemmas~\ref{lem:fcinstructionsmvL} and
  \ref{lem:fcinstructionsmvR}, we define a binary relation $\brelsym$
  by
  \begin{gather*}
    \brelsym=
      \{(\states,\tpstrdL\tphd{\datumd}\tpstrdR),
         \aep{\chanrwm}{\name[C]_{\states,\datumd} \parc
   \name[T]_{\tpstrdL\tphd{\datumd}\tpstrdR}})\mid
              \states\in\StatesS[\TM{}]\ \&\
              \tpstrdL,\tpstrdR\in\Datab^{*}\ \&\
              \datumd\in\Datab
      \}\\
      \qquad \qquad\mbox{}\cup
        \{((\statet,\tpstrzL\tphd{\datumd[L]}\datume\tpstrdR),\stateu)
            \mid \text{$\stateu\in
              \IS{[\datumd{}/\datume{}]L,\statet,\tpstrzL\tphd{\datumd[L]}\datume\tpstrdR}$
              for some $\datumd\in\Datab$}\}\\
      \qquad \qquad\mbox{}\cup
        \{((\statet,\tphd{\tpblank}\datume\tpstrdR),\stateu)
            \mid \text{$\stateu\in
              \IS{[\datumd{}/\datume{}]L,\statet,\tphd{\tpblank}\datume\tpstrdR}$
              for some $\datumd\in\Datab$}\}\\
      \qquad \qquad\mbox{}\cup
        \{((\statet,\tpstrdL\datume\tphd{\datumd[R]}\tpstrzR),\stateu)
            \mid \text{$\stateu\in
              \IS{[\datumd{}/\datume{}]L,\statet,\tpstrdL\datume\tphd{\datumd[R]}\tpstrzR}$
              for some $\datumd\in\Datab$}\}\\
      \qquad \qquad\mbox{}\cup
        \{((\statet,\tpstrdL\datume \tphd{\tpblank}),\stateu)
            \mid \text{$\stateu\in
              \IS{[\datumd{}/\datume{}]L,\statet,\tpstrdL\datume
                \tphd{\tpblank}}$ for some $\datumd\in\Datab$}\}
  \end{gather*}
  The relation $\brelsym$ is illustrated in Fig.~\ref{fig:bbisim-exec-tcpt}.
  We leave it to the reader to verify that $\brelsym$ is a
  divergence-preserving branching
  bisimulation. 
\end{proof}

By Corollary~\ref{cor:finitetape} and Lemma~\ref{lem:congruence}, we
can replace the infinite recursive specification of the tape in
Theorem~\ref{thm:TMspecTinf} by the finite recursive
specification we found in Section~\ref{sec:finitetapespec}. We thus get
the following corollary to Theorem~\ref{thm:TMspecTinf}.

\begin{corollary} \label{cor:execfinitespec}
  For every executable transition $\LTST$ there exists a finite
  recursive specification $\RSpec$ and an $\RSpec$-interpretable
  process expression $\pexp$ such that $\LTST\bbisimed\lts[\RSpec]{\pexp}$.
\end{corollary}
}

\fctonly{%
  In~\cite{BLT11full} we present the details of a construction that
  associates with an arbitrary \RTM{} a \TCPt{} recursive
  specification that defines its behaviour up to divergence-preserving
  branching bisimilarity. Thus, we get the following correspondence.}

Note that if $\RSpec$ is a finite recursive specification in our
calculus, and $\pexp$ is an $\RSpec$-interpretable process expression,
then $\lts[\RSpec]{\pexp}$ is a boundedly branching computable
transition system. Hence, up to divergence-preserving branching
bisimilarity, we get a one-to-one correspondence between executability
and finite definability in our process calculus.
\begin{corollary}
  A transition system is executable modulo divergence-preserving
  branching bisimilarity if, and only if, it is finitely definable
  modulo divergence-preserving branching bisimilarity in the process
  calculus with deadlock, skip, action prefix, alternative composition
  and parallel composition with value-passing handshaking
  communication.
\end{corollary}
\todo{Define `finitely definable up to divergence-preserving branching bisimilarity'.}

For the aforementioned corollary it is important that our calculus
does not include sequential composition. If sequential composition is
added to our calculus, then there exist recursive specifications with
an associated transition system that is unboundedly branching (see,
e.g.,~\cite{BCLT10fsen}).





\section{Persistent Turing Machines} \label{sec:PTM}

\newcommand{\eoi}{\tpsep}
\newcommand{\divout}{\ensuremath{\mathalpha{\mu}}}
\newcommand{\divstate}{\ensuremath{\infty}}
\newcommand{\ITS}[1][I]{\ensuremath{\mathalpha{#1}}}
 \newcommand{\ITSI}[1][]{\ensuremath{\ITS[I_{#1}]}}
\newcommand{\its}[2][]{\ensuremath{\mathcal{I}_{#1}(#2)}}

In~\cite{GSAS04}, the following notion of persistent Turing machine is
put forward.

\begin{definition}
  A \emph{persistent Turing machine} (PTM) $\TM$ is a nondeterministic
  Turing machine with three semi-infinite tapes: a read-only
  \emph{input} tape, a read/write \emph{work} tape and a write-only
  \emph{output} tape.
\end{definition}

The principal semantic notion associated with PTMs in \cite{GSAS04} is
the notion of \emph{macrostep}. Let $\Datab$ be the alphabet of
$\TM{}$.  Then $w\step{w_i/w_o}w'$ denotes that $\TM{}$, when started
in its initial state with its heads at the beginning of its input, work, and
output tapes, containing $w_i$, $w$, and $\emptystr$, respectively,
has a halting computation that produces $w_i$, $w'$, and $w_o$ as the
respective contents of its input, work and output tapes.  Furthermore,
$w\step{w_i/\mu}\divstate$ denotes that $\TM{}$, when started in its
initial state with its heads at the beginning of its input, work and
output tapes, containing $w_i$, $w$, and $\emptystr$, respectively,
has a diverging computation.

The macrostep semantics is used to define for every PTM a
\emph{persistent stream language} with an associated notion of
equivalence, and an \emph{interactive transition system} with
associated notions of isomorphism and bisimilarity. We proceed to
recall the definition of the latter.
\begin{definition}
  Let $\Data$ be a finite set of \emph{data symbols} (not containing
  the symbols $\eoi$ or $\divout$).
  An \emph{interactive transition system} (ITS) over $\Data$ is a
  quadruple $(\States,\trel{},\init)$ consisting of a set of states
  $\States$ containing a special state $\divstate$, a distinguished
  initial state $\init$, and a recursively enumerable 
  $(\Dataseqs\times (\Dataseqs\cup\{\divout\})$-labelled transition
  relation on states. It is assumed that all states in $\States$ are
  reachable, and for all $\states\in\States$ and $w_i\in\Dataseqs$,
  whenever $\states\step{w_i/w_o}\divstate$, then $w_o=\divout{}$, and
  whenever $\divstate\step{w_i/w_o}\states$, then $\states=\divstate$
  and $w_o=\divout$.
\end{definition}

The interactive transition system associated with a PTM $\TM$ is
defined as follows:
\begin{enumerate}
\item its set of states consists of all
  $w\in\Dataseqs\cup\{\divstate\}$ \emph{reachable}  from $\emptystr$
  in $\TM$ by macrosteps, i.e., all $w\in\Dataseqs\cup\{\divstate\}$
  such that, some $k\geq 0$, $w_{i,0},\dots,w_{i,k}\in\Dataseqs$,
  $w_{1},\dots,w_{k}\in\Dataseqs\cup\{\divstate\}$, and
  $w_{o,1},\dots,w_{o,k}\in\Dataseqs\cup\{\divout\}$ such that
  $w_{0}=w$, $w_{j}\step{w_{i,j+1}/w_{o,j+1}}w_{j+1}$ and $w_{k}=w$. 
\item its initial state is $\emptystr$; and
\item its $(\Dataseqs\times(\Dataseqs\cup\{\divout\})$-labelled
  transition relation is defined for all reachable 
    $w,w'\in\Dataseqs\cup\{\divstate\}$, and for all $w_i\in\Dataseqs$
    and $w_o\in\Dataseqs\cup\{\divout\}$ by
  $w \step{w_i/w_o} w'$ if this is a macrostep associated with \TM{}.
\end{enumerate}

It is established in \cite{GSAS04} that the above interpretation
establishes a one-to-one correspondence between PTMs up to macrostep
equivalence (PTMs $\TM{}_1$ and $\TM{}_2$ are \emph{macrostep
  equivalent} if there associated ITSs are isomorphic) and ITSs up to
isomorphism.

To show that PTMs can be simulated by \RTM{}s, we associate with every
ITS an effective transition system, which can then, according to
Theorem~\ref{thm:BBCTSS}, be simulated up to branching bisimilarity by an
\RTM{}. Let $\ITSI$ be an ITS; the transition system $\lts{\ITSI}$
associated with an ITS is defined as follows:
\begin{enumerate}
\item for every state $\states$ of $\ITSI$ it includes an infinite
  collection of states $\states_{i,w}$, one for every $w\in\Dataseqs$,
  and transitions
    $\states_{i,w}\step{\recv{\chani}{\datumd}}\states_{i,w\datumd}$
  modelling that the symbol $\datumd$ is received along the input
  channel $\chani$ and appended to the string $w$;
\item for every state $\states$ of $\ITSI$ it includes an
  infinite collection of states $\states_{o,w}$, one for every
  $w\in\Dataseqs\cup\{\divout\}$, and transitions
    $\states_{o,\datumd{}w}\step{\send{\chano}{\datumd}}\states_{o,w}$
  and
    $\states_{o,\divout}\step{\silent}\states_{o,\divout}$;
\item whenever $\ITSI$ has a transition
  $\states\step{w_i/w_o}\states'$, then $\lts{\ITSI}$ has a transition
    $\states_{i,w_i}\step{\recv{\chani}{\tpsep}}\states_{o,w_o}'$
  (the symbol $\tpsep$ is used to signal the end of the input and its
  receipt starts the procedure to output of $w_o$);
\item for every state $\states$ of $\ITSI$, the associated transition
  system $\lts{\ITSI}$ has a transition
  $\states_{o,\emptystr}\step{\send{\chano}{\tpsep}}\states_{i,\emptystr}$
  (the symbol $\tpsep$ is used to signal the end of the output and its
  sending returns the transition system in input mode).
\end{enumerate}

To illustrate that there is no loss of information in the encoding of
ITSs into transition systems, we define, on a transition system
$\LTST$ of which the set of actions is 
     $\Act=
       \{\recv{\chani}{\datumd},\send{\chano}{\datumd}\mid
            \datumd\in\Data\cup\{\eoi\}\}$,
a reverse procedure. First, we associate macrosteps with sequences of
transitions in \LTST{} as follows:
\begin{enumerate}
\item for states $\states$ and $\states'$ of $\LTST$ and
  $w_i,w_o\in\Dataseqs$ such that
    $w_i=\datumd[0],\datumd[1],\dots,\datumd[k]$
  and
    $w_o=\datume[0],\datume[1],\dots,\datume[\ell]$, let us write
    $\states\step{w_i/w_o}\states'$
  if there exist states $\states_{i,0},\dots,\states_{i,k}$ and
  $\states_{o,0},\dots,\states_{o,\ell+1}$ such that
  \begin{equation*}
    \states\step{\recv{\chani}{\datumd[0]}}
      \states_{i,0}
        \step{\recv{\chani}{\datumd[1]}}\cdots\step{\recv{\chani}{\datumd[k]}}
      \states_{i,k}
        \step{\recv{\chani}{\tpsep}}
      \states_{o,0}
        \step{\send{\chano}{\datume[0]}}
      \states_{o,1}
        \step{\send{\chano}{\datume[1]}}\cdots\step{\send{\chano}{\datume[\ell]}}
      \states_{o,\ell+1}
        \step{\send{\chano}{\tpsep}}
      \states'
  \enskip;
  \end{equation*}
\item for every state $\states$ of $\LTST$ and $w_i\in\Dataseqs$ such
  that
    $w_i=\datumd[0],\datumd[1],\dots,\datumd[k]$,
  let us write $\states\step{w_i/\divout}\divstate$ if there exist
  states $\states_{i,0},\dots,\states_{i,k}$ and an infinite sequence
  of states $(\states_{o,\ell})_{\ell\in\N}$ such that
  \begin{gather*}
    \states\step{\recv{\chani}{\datumd[0]}}
      \states_{i,0}
        \step{\recv{\chani}{\datumd[1]}}\cdots\step{\recv{\chani}{\datumd[k]}}
      \states_{i,k}
        \step{\recv{\chani}{\tpsep}}
      \states_{o,0}
\intertext{and}
    \states_{o,i}\step{\silent}\states_{o,i+1}\ \text{for all
      $i\in\N$}
  \enskip.
  \end{gather*}
\end{enumerate}
  The ITS $\its{\LTST}$ associated with $\LTST$ has as set of states
  the states of $\LTST$ reachable by macrosteps from the initial state
  of $\LTST$  (possibly including $\divstate$), as transitions all
  macrosteps between reachable states, and as initial state the
  initial state of $\LTST$.

\begin{theorem}
  Let $\ITSI$ be an ITS; then $\its{\lts{\ITSI}}$ is isomorphic to $\ITSI$.
\end{theorem}

\begin{remark}
  The purpose of our encoding $\its{\_}$ on transition systems is only
  to illustrate that the encoding $\lts{\_}$ on ITSs is lossless; our
  goal has not been to define the most general encoding. Indeed, the
  above procedure could be generalised by allowing unobservable
  activity modelled as transitions labelled $\silent$ in the
  transition system.
\end{remark}

Via the interpretation of PTMs as ITSs, we can associate with every
PTM $\TM{}$ a transition system $\lts{\TM{}}$. Clearly, the transition
system $\lts{\TM{}}$ associated with $\TM{}$ is effective, and hence,
as a corollary to our main result (Theorem~\ref{thm:BBCTSS}) in
Section~\ref{sec:expr}, there exists an \RTM{} that simulates it up to 
branching bisimilarity. 
\begin{corollary}
  For every PTM \TM{} there exists an \RTM{} $\TM'$ such that
    $\lts{\TM{}}\bbisim\lts{\TM'}$.
\end{corollary}


\fullonly{%
\section{Concluding remarks} \label{sec:conclusion}

Our reactive Turing machines extend conventional Turing machines with
a notion of interaction. Interactive computation has been studied
extensively in the past two decades (see, e.g.,
\cite{BGRR07,GSAS04,LW01}). The goal in these works is mainly to
investigate to what extent interaction may have a beneficial effect on
the power of sequential computation. These models essentially adopt a
language- or stream-based semantics; in particular, they abstract from
moments of choice in (internal) computations. Furthermore, interaction
is added through special input-output facilities of the Turing
machines, rather than as an orthogonal notion. We have discussed the
notion of \emph{persistent Turing machine} from \cite{GSAS04} in
Section~\ref{sec:PTM}; below we briefly discuss the notion of
\emph{interactive Turing machine} from \cite{LW01}.

\paragraph{Interactive Turing machines}

Van Leeuwen and Wiedermann proposed \emph{interactive Turing machines}
(ITMs) in \cite{LW01} (the formal details are worked out by Verbaan in
\cite{Ver05}). An ITM is a conventional Turing machine endowed with an
input port and an output port. In every step the ITM may input a
symbol from some finite alphabet on its input port and outputs a
symbol on its output port. ITMs are not designed to halt; they compute
translations of infinite input streams to infinite output streams.

Already in \cite{LW01}, but more prominently in subsequent work (see,
e.g., \cite{WL08}), van Leeuwen and Wiedermann consider a further
extension of the Turing machine paradigm, adding a notion of advice
\cite{KL82}. An \emph{interactive Turing machines
  with advice} is an ITM that can, when needed, access some advice
function that allows for inserting external information into the
computation. It is established that this extension allows the
modelling of non-uniform evolution. It is claimed by the authors that
non-uniform evolution is essential for modelling the Internet, and
that the resulting computational paradigm is more powerful than that
of conventional Turing machines.

Our \RTM{}s are not capable of modelling non-uniform evolution. We
leave it as future work to consider an extension of \RTM{}s with
advice. In particular, it would be interesting to consider and
extension with behavioural advice, rather than functional advice,
modelling advice as an extra parallel component representing the
non-uniform behaviour of the environment with which the system
interacts.


\paragraph{Expressiveness of process calculi}

In \cite{BBK87}, Baeten, Bergstra and Klop prove that computable
process graphs are finitely definable in \ACPt{} up to weak
bisimilarity; their proof involves a finite specification of a
(conventional) Turing machine. Their result was extended by Phillips
in \cite{Phi93}, who proved that all recursively enumerable process
graphs are finitely definable up to weak bisimilarity. We have further
extended these results by adopting a more general notion of
\emph{final state} and more refined notions of behavioural
equivalence.

\todo{Explain that our more general notion of final state actually
  complicates the specification of the tape. In \cite{BBK87}, Baeten, Bergstra and Klop present a finite
specification of a Turing machine is given in \ACPt{} to simulate 
computable transition systems up to bisimilarity; the conventional Turing
machine is simulated using finite control in parallel with two stacks. 
Their approach to model a tape as two stacks cannot be reused in our
setting, which allows for states that can be terminating and have outgoing
transitions at the same time.
Their specification of the stack does not allow for intermediate
termination, and it is not clear how to adapt it so that it does.
Instead, we model the tape using a (first-in first-out) queue, which does
allow for intermediate termination.}

\todo{Also mention work on parallel-or, Kahn networks, etc.}

}%



\RTM{}s may prove to be a useful tool in establishing the
expressiveness of richer process calculi. For instance, the transition
system associated with a $\pi$-calculus expression is effective, so it
can be simulated by an \RTM{}, at least up to branching
bisimilarity. 
\todo{Say something about $\lambda$-calculus: The $\pi$-calculus can to some extent be seen as the interactive version of
the $\lambda$-calculus.}
We conjecture that the converse ---every executable
transition system can be specified by a $\pi$-calculus expression--- is
also true, but leave it for future work to work this out in detail.

Petri showed already in his thesis \cite{Pet62} that concurrency and
interaction may serve to bridge the gap between the theoretically
convenient Turing machine model of a sequential machine with unbounded
memory, and the practically more realistic notion of extendable
architecture of components with bounded memory. The specification we
present in the proof of Corollary~\ref{cor:execfinitespec} is another
illustration of this idea: the unbounded tape is modelled as an
unbounded parallel composition. It would be interesting to further
study the inherent tradeoff between unbounded parallel composition and
unbounded memory in the context of \RTM{}s, considering unbounded
parallel compositions of \RTM{}s with bounded memory.

\todo{Make a remark on what this contributes to the field of process
  algebra.  For example the $\Pi$-calculus (generates effective LTSes via
  SOS (de Simone/Vaandrager?)).  It is supposed to be
  super-Turing-powerful, but there was never a model to measure this.  The
  RTM is the counterpart of TM, it is process versus function/language. 
  The RTM is in the ``right world''.  Van Leeuwen's model cannot be used
  very well, ours suits better.}

\paragraph{Sequential versus parallel interactive computation}

Interestingly, in the conclusions of \cite{GSAS04} it is conjectured
that parallel composition does affect the notion of interactive
computability, in the sense that the parallel interactive computation
is more expressive than sequential interactive computation. To verify
that claim, one would need to define a notion of parallel interactive
computation; it is unclear to us how this should be done in the
setting of PTMs of \cite{GSAS04}. In our setting, however, it is
straightforward to define a notion of parallel composition on \RTM{}s,
and then our result at the end of Sect.~\ref{sec:expr} that the
parallel composition of \RTM{}s can be faithfully simulated by a
single \RTM{} shows that parallelism does not the enhance the
expressiveness of interactive computation as defined by the model of
\RTM{}s.

\paragraph{Acknowledgement}

We thank Herman Geuvers for discussions, and Clemens Grabmayer for
suggesting to us the term \emph{reactive Turing machine}, and the
anonymous referees for their useful comments.

\fullonly{\bibliographystyle{plain}
\bibliography{paper}}%
\fctonly{%

}

\fctsubmissiononly{%
\newpage

\appendix

\section{Divergence-preserving branching bisimilarity}\label{app:bbed}

We proceed to define the behavioural equivalences that we 
employ in this paper to compare transition systems. Let $\trel$ be
an $\Actt$-labelled transition relation on a set $\States$, and let
$\actt\in\Actt$; we write $\states\optstep{\actt}\statet$ if
$\states\step{\actt}\statet$ or $\actt=\silent$ and
$\states=\statet$.  Furthermore, we denote the transitive closure of
$\step{\silent}$ by $\stepsplus{}$, and we denote the
reflexive-transitive closure of $\step{\silent}$ by $\steps{}$.

\begin{definition}\label{def:bbisim}
  Let
    $\LTST[1]=(\StatesS[1],\trel[1],\init[1],\final[1])$
  and
    $\LTST[2]=(\StatesS[2],\trel[2],\init[2],\final[2])$
  be transition systems.
  A \emph{branching bisimulation} from $\LTST[1]$ to $\LTST[2]$ is a
  binary relation $\brelsym\subseteq\StatesS[1]\times\StatesS[2]$ and, for all states $\states[1]$ and
  $\states[2]$, $\states[1]\brel\states[2]$ implies
  \begin{enumerate}
  \item if $\states[1]\step[1]{\actt}\state[s_1']$, then there exist
      $\state[s_2'],\state[s_2'']\in\StatesS[2]$ such that
      $\states[2]\steps[2]{}\state[s_2'']\optstep[2]{\actt}\state[s_2']$, 
      $\states[1]\brel\state[s_2'']$
    and 
      $\state[s_1']\brel\state[s_2']$;
  \item if $\states[2]\step[2]{\actt}\state[s_2']$, then there exist
    $\state[s_1'],\state[s_1'']\in\StatesS[1]$ such that
      $\states[1]\steps[1]{}\state[s_1'']\optstep[1]{\actt}\state[s_1']$, 
      $\state[s_1'']\brel\states[2]$
    and
      $\state[s_1']\brel\state[s_2']$;
  \item\label{def:bbisim:term1}
    if $\term[1]{\states[1]}$, then there exists $\state[s_2']$
    such that
    $\states[2]\steps[2]{}\state[s_2']$,
    $\states[1]\brel\state[s_2']$ and
    $\term[2]{\state[s_2']}$;
    and
  \item\label{def:bbisim:term2}
    if $\term[2]{\states[2]}$, then there exists $\state[s_1']$
    such that
      $\states[1]\steps[1]{}\state[s_1']$,
      $\state[s_1']\brel\states[2]$
    and
      $\term[1]{\state[s_1']}$.
  \end{enumerate}
  The transition systems $\LTST[1]$ and $\LTST[2]$ are
  \emph{branching bisimilar} (notation: $\LTST[1]\bbisim\LTST[2]$)
  if there exists a branching bisimulation from $\LTST[1]$ to
  $\LTST[2]$ such that $\init[1]\brel\init[2]$.

  A branching bisimulation $\brel$ from $\LTST[1]$ to $\LTST[2]$ is
  \emph{divergence-preserving} if, for all states $\states[1]$ and
  $\states[2]$, $\states[1]\brel\states[2]$ implies
  \begin{enumerate}
  \addtocounter{enumi}{4}
  \item if there exists an infinite sequence
    ${(\states[1,i])}_{i\in\N}$ such that $\states[1]=\states[1,0]$,
    $\states[1,i]\step{\silent}\states[1,i+1]$ and
    $\states[1,i]\brel\states[2]$ for all $i\in\N$, then there
    exists a state $\state[s_2']$ such that
    $\states[2]\stepsplus{}\state[s_2']$ and
    $\states[1,i]\brel\state[s_2']$ for some $i\in\N$; and
  \item if there exists an infinite sequence
    ${(\states[2,i])}_{i\in\N}$ such that $\states[2]=\states[2,0]$,
    $\states[2,i]\step{\silent}\states[2,i+1]$ and
    $\states[1]\brel\states[2,i]$ for all $i\in\N$, then there
    exists a state $\state[s_1']$ such that
    $\states[1]\stepsplus{}\state[s_1']$ and
    $\state[s_1']\brel\states[2,i]$ for some $i\in\N$.
  \end{enumerate}
  The transition systems $\LTST[1]$ and $\LTST[2]$ are
  \emph{divergence-preserving branching bisimilar} (notation:
  $\LTST[1]\bbisimed\LTST[2]$) if there exists a divergence-preserving
  branching bisimulation from $\LTST[1]$ to $\LTST[2]$ such
  that $\init[1]\brel\init[2]$.
\end{definition}

The notions of branching bisimilarity and divergence-preserving
branching bisimilarity originate with~\cite{GW96}. The particular
divergence conditions we use to define divergence-preserving
branching bisimulations here are discussed in~\cite{GLT09}, where it
is also proved that divergence-preserving branching bisimilarity is
an equivalence.
 
\begin{definition}
 Let $\LTST$ be a transition system and let $\states$ and $\statet$ be two
 states in $\LTST$.
 A $\silent$-transition $\states\step{\silent}\statet$ is \emph{inert} if
 $\states$ and $\statet$ are related by the maximal
 divergence-preserving branching bisimulation
 on $\LTST$.
\end{definition}

If $\states$ and $\statet$ are distinct states, then an inert
$\silent$-transition $\states\step{\silent}\statet$ can be
\emph{eliminated} from a transition system, e.g., by removing all outgoing
transitions of $\states$, changing every outgoing transition
$\statet\step{\actt}\stateu$ from $\statet$ to an outgoing transition
$\states\step{\actt}\stateu$, and removing the state $\statet$.  This
operation yields a transition system that is divergence-preserving
branching bisimilar to the original transition system.

An unobservable transition of an \RTM{}, i.e., a transition labelled
with $\silent$, may be thought of as an internal computation
step. Divergence-preserving branching bisimilarity allows us to
abstract from internal computations as long as they do not discard
the option to execute a certain behaviour. The following notion is
technically convenient in the remainder of the paper.
\begin{definition}\label{def:comp}
 Given some transition system $\LTST$, an \emph{internal computation from
 state $\states$ to $\states'$} is a sequence of states
   $\states[1], \dots, \states[n]$
 in $\LTST$ such that 
   $\states=\states[1]\step{\silent}\dots\step{\silent}\states[n]=\states'$.
 An internal computation is called \emph{fully deterministic} iff, for
 every state $\states[i]$ ($1\leq i < n$),
 $\states[i]\step{\acta}\states[i]'$ implies $\acta=\silent$ and
 $\states[i]'=\states[i+1]$.
 We write $\states\fdcomp{}\states'$ if there exists a fully
 deterministic computation from $\states$ to $\states'$.
\end{definition}

\begin{lemma}
  Let $\LTST$ be a transition system and let $\states$ and $\statet$
  be two states in $\LTST$. If $\states\fdcomp{}\states'$, then
  $\states$ and $\states'$ are related by the maximal
  divergence-preserving branching bisimulation on $\LTST$.
\end{lemma}

\section{Example~\ref{ex:notexecbbisim} from
  Sect.~\ref{sec:expr}}\label{app:notexecbbisim}

\begin{example}
  Assume that $\Act=\{\acta,\actb,\actc\}$ and consider the
  $\Act$-labelled transition system
     $\LTST[0]=(\StatesS[0],\trel[0],\init[0],\final[0])$
  with $\StatesS[0]$, $\trel[0]$, $\init[0]$ and $\final[0]$ defined by
  \begin{gather*}
    \StatesS[0]=\{\states,\statet,\stateu,\statev,\statew\}
                          \cup\{\states[x]\mid x \in \N\}\enskip,
\\
    \trel[0]=\{(\states,\acta,\statet),(\statet,\acta,\statet),(\statet,\actb,\statev),
                 (\states,\acta,\stateu),(\stateu,\acta,\stateu),(\stateu,\actc,\statew)\}\\
\qquad\qquad\mbox{}\cup
                 (\states,\acta,\states[0])\}
                    \cup
                 \{(\states[x],\acta,\states[x+1])\mid x\in\N\} \\
\qquad\qquad\mbox{}\cup
                 \{(\states[x],\acta,\statet),(\states[x],\acta,\stateu)\mid
                   \text{$\varphi_x$ is a total function}\}
\enskip,\\
    \init[0]=\states\enskip,\ \text{and}\\
    \final[0]=\{\statev,\statew\}\enskip.
  \end{gather*}
  (The transition system is depicted in Fig.~\ref{fig:ltst0}.)

  \begin{figure}
   \centering\small
   \begin{transsys}(75,30)(0,-5)
    \graphset{Nadjust=n,Nw=4,Nh=4,fangle=-180}
    \node[Nmarks=i](s)(15,10){$\states$}
    \node(t)(15,20){$\statet$}
    \node[Nmarks=f](v)(0,20){$\statev$}
    \node(s0)(30,10){$\states[0]$}
    \node(u)(15,0){$\stateu$}
    \node[Nmarks=f](w)(0,0){$\statew$}
    \node(s1)(45,10){$\states[1]$}
    \node(s2)(60,10){$\states[2]$}
    \node[Nframe=n,Nh=8,Nw=8](s_t)(75,10){}
    {\scriptsize
     \edge[ELside=r](s,t){$\acta$}
     \edge*[loopangle=90](t){$\acta$}
     \edge(t,v){$\actb$}
     \edge(s,u){$\acta$}
     \edge*[loopangle=-90](u){$\acta$}
     \edge(u,w){$\actc$}
     \edge(s,s0){$\acta$}
     \edge(s0,s1){$\acta$}
     \edge(s1,s2){$\acta$}
     \graphset{dash={1 0}{}}
     \edge(s2,s_t){$\acta$}
     \drawbpedge[ELpos=30,ELside=r](s0,135,10,t,-30,10){$\acta$}
     \drawbpedge[ELpos=30](s0,-135,10,u,30,10){$\acta$}
     \drawbpedge[ELpos=20,ELside=r](s1,135,10,t,0,10){$\acta$}
     \drawbpedge[ELpos=20](s1,-135,10,u,0,10){$\acta$}
     \drawbpedge[ELpos=10,ELside=r](s2,135,10,t,0,10){$\acta$}
     \drawbpedge[ELpos=10](s2,-135,10,u,0,10){$\acta$}
     \drawbpedge[ELside=r](s_t,120,10,t,0,10){$\acta$}
     \drawbpedge[dash={1 0}{}](s_t,-120,10,u,0,10){$\acta$}
    }
   \end{transsys}
   \caption{The transition system $\LTST[0]$.}\label{fig:ltst0}
  \end{figure}

  To argue that $\LTST[0]$ is not executable up to branching
  bisimilarity, suppose that it is.  Then $\LTS[T_0]$ is branching
  bisimilar to a computable transition system $\LTS[T_0']$. Then, in
  $\LTS[T_0']$, the set of states reachable by a path that contains
  exactly $x$ $\acta$-transitions ($x\in\N$) and from which both a
  $\actb$- and a $\actc$-transition are still reachable, is
  recursively enumerable. It follows that the set of states in
  $\LTS[T_0']$ branching bisimilar to $\states[x]$ ($x\in\N$) is
  recursively enumerable. But then, since the problem of deciding
  whether from some state in $\LTS[T_0']$ there is a path containing
  exactly one $a$-transition and one $b$-transition such that the
  $a$-transition precedes the $b$-transition, is also recursively
  enumerable, it follows that the problem of deciding whether
  $\varphi_x$ is a total function must be recursively enumerable too, quod non.
  We conclude that $\LTST[0]$ is not executable up to branching
  bisimilarity. Incidentally, note that the language associated with
  $\LTST[0]$ is $\{a^n b,a^n c\mid n\geq 1\}$, which \emph{is}
  recursively enumerable (it is even context-free).
\end{example}

\section{Example~\ref{ex:divergence} from Sect.~\ref{sec:expr}} 
\label{app:effectivedivergence}

\begin{example}
  Assume that $\Act=\{\acta,\actb\}$,
  and consider the transition system
    $\LTST[1]=(\StatesS[1],\trel[1],\init[1],\final[1])$
  with $\StatesS[1]$, $\trel[1]$, $\init[1]$ and $\final[1]$ defined by
  \begin{gather*}
    \StatesS[1]=\{\states[1,x],\statet[1,x]\mid x \in \N\}\enskip,\\
    \trel[1]=\{(\states[1,x],\acta,\states[1,x+1])\mid x\in\N\}\cup
             \{(\states[1,x],\actb,\statet[1,x])\mid x \in\N \}\enskip,\\
    \init[1]=\states[1,0]\enskip,\ \text{and}\\
    \final[1]=\{\statet[1,x]\mid\text{$\varphi_x(x)$ converges}\}\enskip.
  \end{gather*}
  (See also Fig.~\ref{fig:ex-lts-eff}).

  \begin{figure}
    \begin{center}
    \begin{transsys}(60,15)(0,5)
      \graphset{fangle=-90}
      \node[Nmarks=i](s1_0)(0,16){$\states[1,0]$}
      \node(s1_1)(15,16){$\states[1,1]$}
      \node(s1_2)(30,16){$\states[1,2]$}
      \node(s1_3)(45,16){$\states[1,3]$}
      \node[Nframe=n](s1_t)(60,16){}
      \node(t_0)(0,6){$\statet[0]$}
      \node(t_1)(15,6){$\statet[1]$}
      \node(t_2)(30,6){$\statet[2]$}
      \node(t_3)(45,6){$\statet[3]$}
      \graphset{Nframe=n,Nw=0,Nh=0,Nadjust=n}
      \node(t_0t)(0,0){}
      \node(t_1t)(15,0){}
      \node(t_2t)(30,0){}
      \node(t_3t)(45,0){}
      \edge(s1_0,s1_1){$\acta$}
      \edge(s1_1,s1_2){$\acta$}
      \edge(s1_2,s1_3){$\acta$}
      \edge[dash={1 0}{.2}](s1_3,s1_t){$\acta$}
      \edge(s1_0,t_0){$\actb$}
      \edge(s1_1,t_1){$\actb$}
      \edge(s1_2,t_2){$\actb$}
      \edge(s1_3,t_3){$\actb$}
      \graphset{dash={.5 0}{}}
      \edge(t_0,t_0t){}
      \edge(t_1,t_1t){}
      \edge(t_2,t_2t){}
      \edge(t_3,t_3t){}
    \end{transsys}
    \end{center}
    \caption{The transition system $\LTST[1]$.}\label{fig:ex-lts-eff}
  \end{figure}

  Now, suppose that $\LTST[2]$ is a transition system such that
    $\LTST[1]\bbisimed\LTST[2]$,
  as witnessed by some divergence-preserving branching bisimulation
  relation $\brel$; we argue that $\LTST[2]$ is not computable by
  deriving a contradiction from the assumption that it is.

  Clearly, since $\LTST[1]$ does not admit infinite
  sequences of $\silent$-transitions, if $\brel$ is
  divergence-preserving, then $\LTST[2]$ does not admit infinite
  sequences of $\silent$-transitions either. It follows that if
     $\states[1]\brel\states[2]$,
  then there exists a state $\state[s_2']$ in $\LTST[2]$ such that
    $\states[2]\steps[2]{}\state[s_2]$,
    $\states[1]\brel\state[s_2']$,
  and
    $\state[s_2']\nstep{\silent}$.
  Moreover, since $\LTST[2]$ is computable and does not admit infinite
  sequences of consecutive $\silent$-transitions, a state $\state[s_2']$
  satisfying the aforementioned properties is produced by the algorithm
  that, given a state of $\LTST[2]$, selects an enabled
  $\silent$-transition and recurses on the target of the transition
  until it reaches a state in which no $\silent$-transitions are enabled.

  But then we also have an algorithm that determines if $\varphi_x(x)$
  converges:
  \begin{enumerate}
  \item it starts from the initial state $\init[2]$ of $\LTST[2]$;
  \item it runs the algorithm to find a state without outgoing
      $\silent$-transitions, and then it repeats the following steps
      $x$ times:
    \begin{enumerate}
    \item execute the $\mathit{inc}$-transition enabled in the reached state;
    \item run the algorithm to find a state without outgoing
      $\silent$-transitions again;
    \end{enumerate}
    since $\init[1]\brel\init[2]$, this yields a state $\states[2,x]$
    in $\LTST[2]$ such that $\states[1,x]\brel\states[2,x]$;
  \item it executes the $\mathit{run}$-transition that must be enabled
    in $\states[2,x]$, followed, again, by the algorithm to find a state
    without outgoing $\silent$-transitions; this yields a state
    $\statet[2,x]$, without any outgoing transitions, such that
    $\statet[1,x]\brel\statet[2,x]$.
  \end{enumerate}

  From $\statet[1,x]\brel\statet[2,x]$ it follows that
  $\statet[2,x]\in\final[2]$ iff $\varphi_x(x)$ converges, so the problem
  of deciding whether $\varphi_x(x)$ converges has been reduced to the
   problem of deciding whether $\statet[2,x]\in\final[2]$. Since it is
  undecidable if $\varphi_x(x)$ converges, it follows that $\final[2]$ is
  not recursive, which contradicts our assumption that $\LTST[2]$ is
  computable.
\end{example}

\section{The operational semantics of \TCPt{}}\label{app:tcptos}
We use Structural Operational Semantics~\cite{Plo04a} to associate a
transition relation with \TCPt{} process expressions:
let $\trel$ be the $\Actt$-labelled transition relation induced on the set
of process expressions by operational rules in Table~\ref{tbl:sos-tcpt}.
Note that the operational rules presuppose a recursive specification
$\RSpec$. 


\begin{definition}\label{def:lts-rspec}
  Let $\RSpec$ be a recursive specification and let $\pexpp$ be a
  process expression.
  We define the labelled transition system
    $\lts[\RSpec]{\pexpp}=
      (\StatesS[\pexpp],\trel[\pexpp],\init[\pexpp],\final[\pexpp])$
  associated with $\pexpp$ and $\RSpec$ as follows:
  \begin{enumerate}
  \item the set of states $\StatesS[\pexpp]$ consists of all process
    expressions reachable from $\pexpp$;
  \item the transition relation $\trel[\pexpp]$ is the restriction
    to $\StatesS[\pexpp]$ of the transition relation $\trel$ defined
    on all process expressions by the operational rules in
    Table~\ref{tbl:sos-tcpt}, i.e.,
      $\trel[\pexpp]=\trel\cap
        (\StatesS[\pexpp]\times\Actt\times\StatesS[\pexpp])$.
  \item the process expression $\pexpp$ is the initial state,
    i.e. $\init[\pexpp]=\pexpp$; and
  \item the set of final states consists of all process expressions
    $\pexpq\in\StatesS[\pexpp]$ such that $\term{\pexpq}$, i.e.,
      $\final[\pexpp]=\final\cap\StatesS[\pexpp]$.
  \end{enumerate}
\end{definition}

\begin{framedtable}[htb]
 \begin{center}
 \begin{osrules}
   \osrule*{}{\term{\emp}} \qquad \qquad \qquad \qquad
     \osrule*{}{\pref{\act}{\pexp} \step{\act} \pexp} \\
   \osrule*{\pexpp \step{\act} \pexpp'}%
           {\pexpp \altc \pexpq \step{\act} \pexpp'} \qquad
     \osrule*{\pexpq \step{\act} \pexpq'}%
             {\pexpp \altc \pexpq \step{\act} \pexpq'} \qquad
     \osrule*{\term{\pexpp}}{\term{(\pexpp \altc \pexpq)}} \qquad
     \osrule*{\term{\pexpq}}{\term{(\pexpp \altc \pexpq)}} \\
   \osrule*{\pexpp \step{\act} \pexpp'}%
           {\pexpp \seqc \pexpq \step{\act} \pexpp' \seqc \pexpq} \qquad
     \osrule*{\term{\pexpp} & \pexpq \step{\act} \pexpq'}%
             {\pexpp \seqc \pexpq \step{\act} \pexpq'} \qquad
     \osrule*{\term{\pexpp} & \term{\pexpq}}%
             {\term{(\pexpp \seqc \pexpq)}} \\
   \osrule*{\pexpp \step{\act} \pexpp'& \neg\exists\chan\in\Chan',\datum\in\Datab.\ \act=\recv{\chan}{\datum}\vee \act= \send{\chan}{\datum}}%
           {\aep{\Chan'}{\pexpp \parc \pexpq} \step{\act} \aep{\Chan'}{\pexpp' \parc \pexpq} \quad \aep{\Chan'}{\pexpq \parc \pexpp} \step{\act} \aep{\Chan'}{\pexpq \parc \pexpp'}} \\
   \osrule*{\pexpp \step{\send{\chan}{\datum}} \pexpp' &
            \pexpq \step{\recv{\chan}{\datum}} \pexpq' & \chan\in\Chan',\ \datum\in\Datab}%
           {\aep{\Chan'}{\pexpp \parc \pexpq} \step{\silent} 
            \aep{\Chan'}{\pexpp' \parc \pexpq'}
              \quad
            \aep{\Chan'}{\pexpq \parc \pexpp} \step{\silent} 
            \aep{\Chan'}{\pexpq' \parc \pexpp'}}\\
   \osrule*{\pexpp \step{\act} \pexpp' &
            (\name \defeq \pexpp) \in \RSpec}%
           {\name \step{\act} \pexpp'} \qquad 
     \osrule*{\term{\pexpp} &
              (\name \defeq \pexpp) \in \RSpec}%
             {\term{\name}}
  \end{osrules}
 \caption{Operational rules for a recursive specification 
   $\RSpec$.}\label{tbl:sos-tcpt}
 \end{center}
\end{framedtable}%
}
\end{document}